\pdfoutput=1
\documentclass[final,3p,times]{elsarticle}
\usepackage{graphicx}
\usepackage{amssymb,xfrac}
\usepackage{amsthm,amsmath,dsfont}
\usepackage{thmtools}
\allowdisplaybreaks

\usepackage{lineno}

\usepackage[colorlinks=true, citecolor=blue]{hyperref}
\usepackage[algoruled,titlenotnumbered]{algorithm2e}
\usepackage{algorithmic}
\usepackage{wrapfig}
\usepackage{natbib}
\usepackage{comment}
\usepackage{multirow}
\usepackage{amsmath,bm}
\usepackage{graphicx}
\usepackage[lofdepth,lotdepth]{subfig}
\usepackage{float}
\usepackage{subfig}
\usepackage{wrapfig}
\usepackage{framed,xcolor}
\usepackage{newfloat}
\usepackage{amsthm}
\usepackage{thmtools}
\usepackage{stmaryrd}
\usepackage{mathtools}
\usepackage{tikz}
\usetikzlibrary{shapes.geometric,arrows,shapes,chains,calc,patterns,decorations.pathreplacing}

\usepackage{color,soulutf8}
\usepackage[colorinlistoftodos,textsize=footnotesize,textwidth=1.5in]{todonotes}

\newcommand{\RR}{\mathbb{R}}

\newcommand{\EE}{\mathbb{E}}

\newcommand{\GX}{\mathcal{G}}
\newcommand{\AX}{\mathcal{A}}

\newcommand{\NX}{\mathcal{N}}

\newcommand{\ZZ}{\mathbb{Z}}
\newcommand{\HH}{\mathbf{H}}

\newcommand{\cc}{\mathrm{C}}
\newcommand{\vv}{\mathrm{V}}
\newcommand{\Sonly}{\mathsf{S-only}}

\newtheoremstyle{theoremdd}% name of the style to be used
{\topsep}% measure of space to leave above the theorem. E.g.: 3pt
{\topsep}% measure of space to leave below the theorem. E.g.: 3pt
{\itshape}% name of font to use in the body of the theorem
{0pt}% measure of space to indent
{\fontfamily{cmss}\selectfont\bfseries}% name of head font
{.}% punctuation between head and body
{ }% space after theorem head; " " = normal interword space
{\thmname{#1}\thmnumber{ #2}\thmnote{ (#3)}}

\theoremstyle{theoremdd}
\newtheorem{theorem}{Theorem}
\newtheorem{lemma}{Lemma}

\newtheorem{definition}{Definition}

\newproof{pf1}{Proof of Theorem \ref{thm:mdp}}
\newproof{pf2}{Proof of Lemma \ref{Lem2}}
\newproof{pf}{Proof}

\usepackage{array}
\usepackage{enumitem}
\usepackage{threeparttable}
\makeatletter
\newcommand{\thickhline}{%
    \noalign {\ifnum 0=`}\fi \hrule height 1.5pt
    \futurelet \reserved@a \@xhline
}
\newcolumntype{"}{@{\hskip\tabcolsep\vrule width 1pt\hskip\tabcolsep}}
\makeatother

\journal{Transportation Research Part E (VSI: ISTTT24 Collected Papers)}

\biboptions{authoryear}

\usepackage{titlesec}
\titleformat*{\section}{\fontfamily{cmss}\selectfont\large\bfseries}
\titleformat*{\subsection}{\fontfamily{cmss}\selectfont\normalsize\bfseries}
\titleformat*{\subsubsection}{\fontfamily{cmss}\selectfont\normalsize}
\graphicspath{{./Figs/}}

\let\today\relax
\makeatletter
\def\ps@pprintTitle{%
	\let\@oddhead\@empty
	\let\@evenhead\@empty
	\def\@oddfoot{\footnotesize\itshape
		{This is a preprint of an article that was accepted for publication in Transportation Research Part E (VSI: ISTTT24 Collected Papers).} \hfill\today}%
	\let\@evenfoot\@oddfoot
}
\makeatother

\begin{document}
	
	\begin{frontmatter}
		
		
		\title{\fontfamily{cmss}\selectfont A real-time dispatching strategy for shared automated electric vehicles with performance guarantees}
		

		
\author[1,2]{Li Li}
\author[1]{Theodoros  Pantelidis}
\author[1]{Joseph Y.J. Chow}
\author[1,2]{Saif Eddin Jabari\corref{cor1}}
\cortext[cor1]{Corresponding author, e-mail: \url{sej7@nyu.edu}}
\address[1]{New York University Tandon School of Engineering, Brooklyn NY}
\address[2]{New York University Abu Dhabi, Saadiyat Island, P.O. Box 129188, Abu Dhabi, U.A.E.}

{ \fontfamily{cmss}\selectfont\large\bfseries		
\begin{abstract}
{ \normalfont\normalsize
	Car-sharing has emerged as a competitive technology for urban mobility. Combined with the upward trend in vehicle electrification and the promise of automation, it is expected that urban travel will change in fundamental ways in the near future.  Indeed, breakthroughs in battery technology and the incentive programs offered by governments worldwide have resulted in a continued increase in the market share of electric vehicles. Automation frees passengers from having to drive and seek parking, it also offers increased flexibility when selecting pick up locations.  These trends and incentives naturally suggest that shared automated electric vehicle (SAEV) systems will displace traditional gasoline-powered, human-driven car-sharing systems worldwide.
	
	Real-time vehicle dispatching operations in traditional car-sharing systems is an already computationally challenging scheduling problem. Electrification only exacerbates the computational difficulties as charge level constraints come into play.  To overcome this complexity, we employ an online \emph{minimum drift plus penalty} (\textsf{MDPP}) approach for SAEV systems that (i) does not require a priori knowledge of customer arrival rates to the different parts of the system (i.e. it is practical from a real-world deployment perspective), (ii) ensures the stability of customer waiting times, (iii) ensures that the deviation of dispatch costs from a desirable dispatch cost can be controlled, and (iv) has a computational time-complexity that allows for real-time implementation. Using an agent-based simulator developed for SAEV systems, we test the \textsf{MDPP} approach under two scenarios with real-world calibrated demand and charger distributions: 1) a low-demand scenario with long trips, and 2) a high-demand scenario with short trips. The comparisons with other algorithms under both scenarios show that the proposed online \textsf{MDPP} outperforms all other algorithms in terms of both reduced customer waiting times and vehicle dispatching costs.
}
\end{abstract}
}
		
		
\end{frontmatter}
	
	
%	\linenumbers
	
\section{Introduction}
\label{S:intro}
Car-sharing is gaining popularity throughout the world but especially in big cities with dense populations, such as New York City, Tokyo, Moscow, and Shanghai. Customers of a car-sharing system have access to private cars without having to bear the costs and responsibilities of car ownership. The fees for using car-sharing services are usually much lower than taxis, e.g. customers only need to pay approximately one-third to one-half of the taxi fee to complete the same trip using EVCARD in Shanghai. Customers in a traditional car-sharing system search for nearby (available) vehicles through an app, book the vehicle they like, and then walk to the location of the vehicle that they booked. The vehicles can be picked up from and returned to any location in a one-way non-electric car-sharing system, such as the car2go in New York City.  In an electric car-sharing system like EVCARD in Shanghai, however, companies usually require customers to return the vehicles to charging stations and connect the vehicles to chargers before they leave to ensure they are recharged. Naturally, customers can only pick vehicles up from a charging station as a result. There is no vehicle-to-customer assignment optimization in traditional car-sharing systems; pick up and drop off location choices are left to the customers.  However, companies do need to resolve the potential imbalance (in vehicle distribution in the network) that results, and they do so with vehicle rebalancing schedules.  This also leads to complex staff rebalancing problems.  Examples in the literature on vehicle relocation problems for real-world car-sharing systems include \citep{smith2013rebalancing, barrios2014fleet, zhao2018integrated, xu2018electric, wang2019optimization}. These papers typically solve both vehicle relocation and staff rebalancing problems.

The potential future use of automated vehicles (AVs) in car sharing systems, \emph{shared automated vehicles} (SAVs), has also received attention in the literature \citep{zhang2015exploring, krueger2016preferences, liu2017tracking, levin2017congestion, ma2017designing,jung2019large}. SAVs have also been referred to as \emph{autonomous taxis} \citep{dandl2017microsimulation, burghout2015impacts, bischoff2016simulation} and \emph{automated mobility-on-demand} \citep{spieser2014toward, azevedo2016microsimulation}. With SAVs, companies no longer need human staff to relocate the vehicles. Instead of requiring customers to walk to the pick up locations, AVs can drive to the customer locations \citep{fagnant2014travel}. For such systems, vehicle-to-customer assignment becomes advantageous and, arguably, required.

Studies on SAVs have either focused on the vehicle-to-customer assignment problem \citep{seow2009collaborative, marczuk2015autonomous, bischoff2016simulation, hanna2016minimum, boesch2016autonomous, hyland2018dynamic} or vehicle relocation \citep{pavone2012robotic, volkov2012markov, marczuk2016simulation, zhang2016control, sayarshad2017non, wen2017rebalancing}. Others have considered the joint vehicle assignment and relocation problem \citep{fagnant2014travel, burghout2015impacts, fagnant2015operations, spieser2016vehicle, zhang2016model, gueriau2018samod}.  \cite{marczuk2015autonomous} assign the nearest available vehicle to customers, and customers are served on a first-come first-serve (FCFS) basis, while \cite{seow2009collaborative} aggregate the customer requests in a queue and dispatch the same number of vehicles to the queue when it reaches a threshold size.  Others propose the use of historical customer demand information (arrival rates) as part of rebalancing strategies. Among them, \cite{pavone2012robotic} propose a continuous-time fluid model for rebalancing that converges to a (stationary) system state where no customers are waiting and the number of vehicles used to rebalance the system is minimized. \cite{volkov2012markov} approach the rebalancing problem from an incident management perspective; their approach aims to ensure that the steady state service rate (``taxi'' assignment) exceeds or meets the steady state rate customer arrivals throughout the network.  \cite{zhang2016model} propose a model predictive control (MPC) approach for assignment of SAVs to customers and their rebalancing; their simulation results suggest that their approach outperforms the approaches in \citep{marczuk2015autonomous, seow2009collaborative, pavone2012robotic, volkov2012markov}. \cite{zhang2016model} also prove that their approach ensures stability of the queuing dynamics, i.e., that queues will not grow indefinitely. However, their MPC approach turns out to be a MILP problem which scales poorly to network size, and can only handle tens of nodes.

To cater to the increasing needs for electric vehicles brought by the breakthroughs in battery technology, \cite{zhang2016model} also develop a number of charging constraints with which the MPC could be integrated to deal with Shared Automated Electric Vehicles (SAEVs).  In spite of this, the charging is only considered as constraints and not optimized in the objective function. \cite{iacobucci2019optimization} extend the original MPC approach of \cite{zhang2016model} to minimize both waiting time and electricity cost through charging optimization. Results show that the modified model could significantly reduce charging cost without a big influence on the waiting time. However, the extended MPC is even more complicated and hence is still only feasible to small systems. There are some other literature dealing with big SAEV systems, and they usually use greedy/heuristic vehicle assigmnent/rebalancing algorithms and recharging rules in their simulations. For example, \cite{chen2016operations} use a greedy search algorithm to look for the closest available SAEV within a 5-minute travel time radius for each customer based on a FCFS rule. The simulation step is 5 minutes, and in each time step available vehicles are rebalanced if not assigned to customers, otherwise the system checks whether the vehicles have sufficient range to serve customers, those that do not have sufficient range are charged. A rebalancing algorithm used for SAVs in \cite{ fagnant2014travel} requiring knowledge of arrival rates is implemented in \cite{chen2016operations} along with a range check. While charging, vehicles are simply assumed to drive to the nearest charging stations and stay there until fully charged.

\cite{loeb2018shared} extend the agent-based SAV simulation tools presented in \citep{boesch2016autonomous,chen2016operations} to include a more precise monitoring of the real-time vehicle battery consumption,  hence simulating SAEVs.  Their framework assigns customers to the nearest available vehicles, and no rebalancing strategy is employed.  They introduce three conditions under which a vehicle recharges: (i) when the charge level falls below 5\%, (ii) when the vehicle is idle for more than 30 minutes, and (iii) when vehicles receive requests that they cannot fulfill because of low range and charge level is below 80\%. Moreover, unlike \citep{chen2016operations}, the framework in \citep{loeb2018shared} allows a charging vehicle to be assigned to customers but only when all other eligible vehicles are unavailable. Similarly, \cite{bauer2018cost} present a framework that allows vehicles to recharge in spurts in between trip requests; their approach treats charging vehicles and fully charged vehicles the same way.  This type of simulation setting is suitable for the scenario that is studied in \citep{bauer2018cost}: Manhattan taxi trips with an average distance of 3 kilometers. However, the rules that they use to decide where and when a vehicle should recharge depend on demand predictions, which is not realistic in practice.

A common feature of all of the approaches above is that they all treat vehicle recharging as an independent problem, independent of both the vehicle-to-customer assignment problem and the vehicle rebalancing problem.  In other words, vehicles are not permitted to charge en route to picking up customers or during rebalancing. This excludes the possibility of co-optimization for recharging, assignment and rebalancing.  Some recent papers \citep{ma2019optimal, li2019agent, pantelidis2020node} have addressed this limitation by allowing vehicles to recharge en route to their rebalancing destinations. However, these approaches either do not scale well computationally \citep{ma2019optimal, pantelidis2020node} or involve the use of heuristics and offer no guarantees of performance \citep{jung2014stochastic,li2019agent}. 

We propose a methodology that is particularly suitable for real-time operations of SAEVs and that comes with theoretical guarantees of performance.  We combine the computational simplicity of heuristic approaches with the mathematical rigor of optimization-based approaches. Specifically, we model the dispatching problem as a stochastic queuing network and employ Lyapunov optimization techniques to derive a policy that ensures stability of waiting times in the network while also accounting for dispatch costs.  We employ a minimum drift plus penalty (MDPP) framework \citep{neely2010stochastic}, in which the vehicle assignment and recharging problems are jointly optimized. The objective function seeks to minimize a combination of the vehicle dispatch cost and customer waiting times.  We provide a rigorous proof of the stability of customer waiting times within the network along with a theoretical bound on the deviation of vehicle dispatch costs from a desirable level.  Our approach is an online approach, which both simplifies the problem and has practical advantages (real-time operation). The system state (travel times) are calculated in a distributed manner by the unused vehicles in the system and updated on a periodic basis (e.g., once every 5 minutes), we show that the time complexity of these periodic operations are favorable.  The online scheduling approach has a time complexity that is linear in the number of customer queues in the system (number of waiting head-of-line customers).  The online algorithm does not require a priori knowledge of customer arrivals to the system, which renders it naturally applicable in real-world settings. We compare the proposed approach with several other algorithms from the literature using an agent-based simulator \citep{li2019agent}. Both a low-demand scenario with long trip distances and a high-demand scenario with short trip distances were tested. The simulation results show that the \textsf{MDPP} approach is superior to all other algorithms tested in terms of customer waiting times and dispatch costs with a proper choice of penalty constant. 

The rest of the paper is organized as follows: Sec.~\ref{sec:method} describes the queuing dynamics and formulates \textsf{MDPP}, and Sec.~\ref{sec:stability} provides formal proofs of the claimed performance guarantees of our model. A toy illustrative example is given in Sec.~\ref{S:example}, and numerical experiments and comparisons are provided in Sec.~\ref{S:simulation}. Sec.~\ref{S:Conc} concludes the paper.

\section{Methodology}
\label{sec:method}

\subsection{Network construction}
\label{SS:nc}
In a SAEV system, a customer enters the network, requests a vehicle, and waits for the system to assign a vehicle to them. Once a vehicle is assigned to the customer, the vehicle will automatically drive to the customer's pick up location and drive them to their destination. Once the assignment is made, the vehicle is committed to the customer over the duration of this process.  That is, the assignment is not changed if a new vehicle becomes available or a new customer enters the systems that might be closer to the customer or the assigned vehicle, respectively. We divide the network into small geographic zones. While customers may be picked up or dropped off in any zone, charging stations may only exist in some zones. We assume that those vehicles that drop customers off at zones with charging stations are automatically connected to chargers; while those vehicles that drop customers off at zones without charging stations will remain at the drop-off locations until re-assigned.

Let the graph $\GX = (\NX_{\cc}, \NX_{\vv}, \AX)$ represent the car-sharing network, where $\NX_{\cc}$ and $\NX_{\vv}$ are graph nodes that represent customers and vehicles, respectively, while $\AX$ is a set of arcs.  A customer node is defined as a zone with a specific charge level. Fig.~\ref{F:network_a} shows a simple network example with 10 zones. Each zone can be classified into multiple customer nodes with different charge levels; a 5-level example is shown in Fig.~\ref{F:network_b}.
\begin{figure}[h!]
	\centering
	\subfloat[]{\includegraphics[scale=0.45]{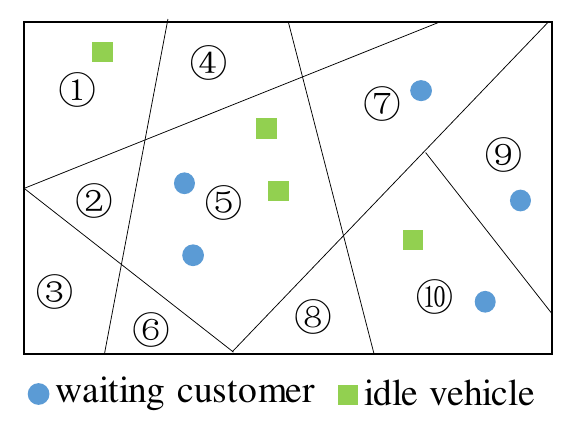}\label{F:network_a}} 
	\subfloat[]{\includegraphics[scale=0.45]{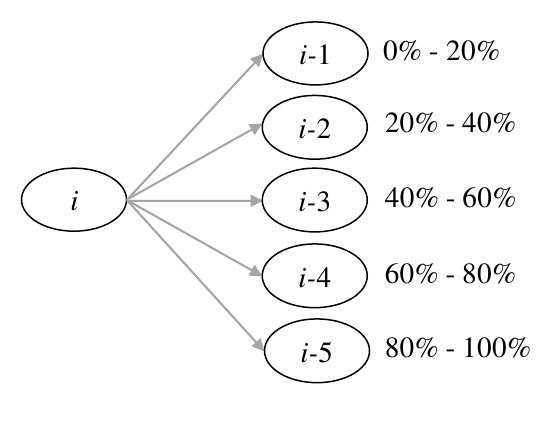}\label{F:network_b}}
	\caption{Network representation: (a) an example of a small network, and (b) conversion from a zone to multiple customer nodes.\label{F:network}}		
\end{figure}
An arc from a vehicle node $v \in \NX_{\vv}$ to a customer-charge node $c \in \NX_{\cc}$ only exists when vehicle $v$ is idle, as shown in Fig.~\ref{F:c2v_a}. The arc weight, denoted as $C_{vc}$, is the travel time (plus charge time, if charging is required) from $v$ to $c$. For example, if we assume that three of the zones in Fig.~\ref{F:network_a} include charging stations, as shown in Fig.~\ref{F:c2v_b}, then the arc weight for the $(v,c)$ pair in Fig.~\ref{F:c2v_b} equals to the travel time if $v$ has enough charge to serve customer-charge node $c$ and then drive to the nearest charging station after drop-off; otherwise $C_{vc}$ is equal to the travel time plus the required time to charge $v$.  With this representation, vehicle $v$ is considered to be idle whenever it is not being assigned to a customer, and a vehicle that is charging can be idle (as long as it is not assigned to a customer).  
\begin{figure}[h!]
	\centering
	\subfloat[]{\includegraphics[scale=0.42]{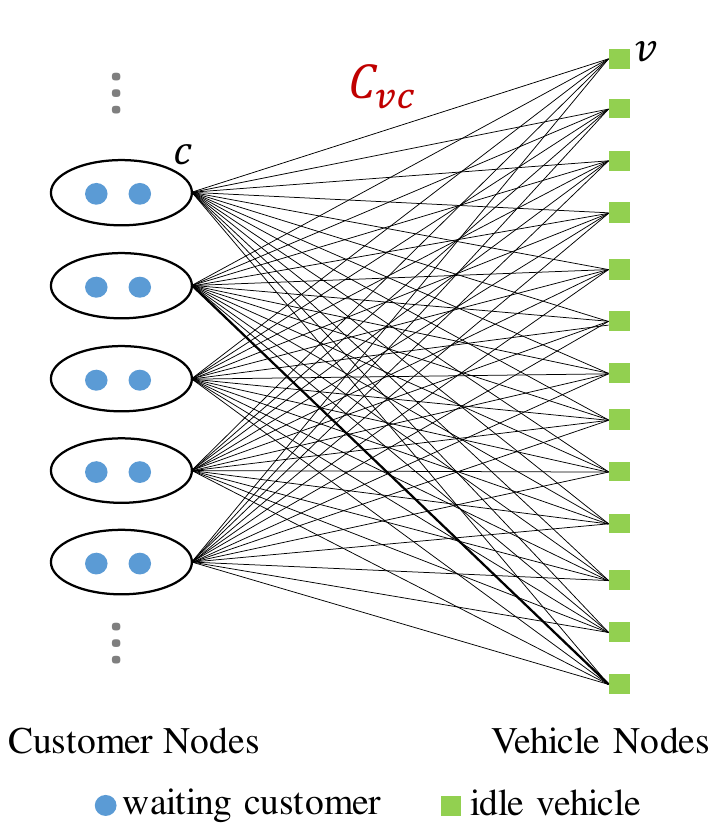}\label{F:c2v_a}} 
	\subfloat[]{\includegraphics[scale=0.42]{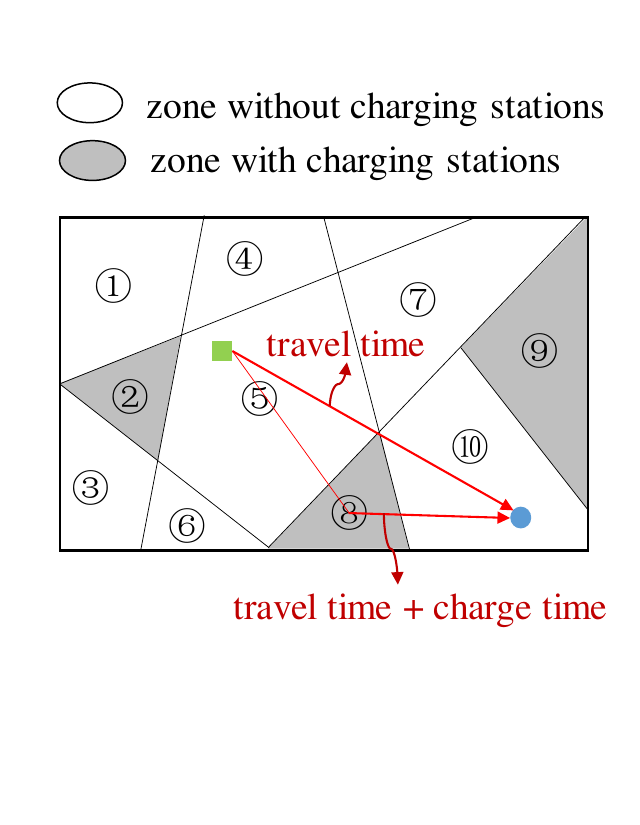}\label{F:c2v_b}}
	\caption{Arc weight: (a) a network with nodes connected by arcs, and (b) an example of arc weight calculation for a $(v,c)$ pair.\label{F:c2v}}		
\end{figure}

Note that the arcs in Fig.~\ref{F:c2v_b} do not represent the routes that the vehicles traverse, only the directed zone-to-zone end-points of those routes.  The real vehicle routes depend on the real-time traffic information from the road network, including the network road topology (physical distances) and traffic condition (travel times). We assume that every vehicle $v$ is capable of assessing $C_{vc}$ for every customer-charge level pair $c \in \NX_{\cc}$ and report the costs $\{C_{vc}\}_{c \in \NX_{\cc}}$ to the system operator. This is a reasonable assumption, given the capabilities of modern-day navigation systems.  From a computational standpoint, it is also feasible as we shall discuss next.  Let $\GX_{\mathrm{roads}} = (\mathcal{V},\mathcal{L})$ denote the road network with vertices $\mathcal{V}$ representing the network intersections and $\mathcal{L}$ the set of road segments in the network. We decompose the set $\mathcal{V}$ into two sets, $\mathcal{V}_{\mathrm{S}}$ and $\mathcal{V} \setminus \mathcal{V}_{\mathrm{S}}$, where $\mathcal{V}_{\mathrm{S}}$ are (the nearest intersections to) the network charging stations \emph{without customers}.  In the absence of charging requirements, the costs $\{C_{vc}\}_{c \in \NX_{\cc}}$ are determined by performing a \emph{single-source shortest-path} (SSSP) search. This can be achieved with a time complexity of $O(|\mathcal{L}|+|\mathcal{V}| \log \log \min\{|\mathcal{V}|,w^{\mathsf{max}}\})$, where $w^{\mathsf{max}}$ is the largest link weight (distance or cost) in the network \citep{thorup2004integer}.  We note that (i) this is a deterministic worst-case bound that (ii) does not make any assumptions about the structure of the network, and (iii) the underlying algorithm offers an exponential improvement over Dijkstra's algorithm.  There also exist speed-up techniques for road networks that are suitable for real-time implementations (see, e.g., \cite{madkour2017survey} and references therein).
	
For $(v,c)$ pairs that require charging, the shortest path from the vehicle to the location of the customer consists of a shortest path from the vehicle to a charging station (already determined by the SSSP for all nodes in $\mathcal{V}_{\mathrm{S}}$) and then a shortest path from the charging station to the customer.  One only needs to determine the charging station but this can be achieved with little overhead while searching for the shortest paths from the charging stations to the customers. Hence, employing the algorithm in \citep{thorup2004integer}, we have an overall complexity of $O(|\mathcal{V}_{\mathrm{S}}||\mathcal{L}|+|\mathcal{V}_{\mathrm{S}}||\mathcal{V}| \log \log \min\{|\mathcal{V}|,w^{\mathsf{max}}\})$ to determine $\{C_{vc}\}_{c \in \NX_{\cc}}$. Note that in most settings $|\mathcal{V}_{\mathrm{S}}| \ll |\mathcal{V}|$, rendering the resulting complexity quite favorable.  Note that, while $\{C_{vc}\}_{c \in \NX_{\cc}}$ vary with time, this variation is relatively slow.  In other words, in practice one only needs to perform periodic updates of the estimated travel costs (e.g., once every 5 minutes). 

\subsection{Queuing dynamics}
\label{SS:qd}

We discretize time into intervals of equal length, and use time $t$ to refer to the beginning of interval $[t, t+1)$. Let $H_c(t)$ denote the waiting time of the head-of-line (HOL) customer at (customer-charge) node $c \in \NX_{\cc}$ at time $t$.  The waiting time dynamics follow
\begin{equation}
\label{qd}
H_c(t+1) = \chi_c(t) \big[H_c(t) + 1 - x_c(t)\tau_c \big]^+ + \big( 1 - \chi_c(t) \big) A_c(t),
\end{equation}
where $[\cdot]^+ \equiv \max\{\cdot,0\}$, $\chi_c(t)$ is a binary variable equal to 0 if node $c$ is empty at time $t$ and 1 otherwise, $x_c(t)$ is a binary decision variable which is 1 if the HOL customer at node $c$ is served at time $t$ and 0 otherwise, $\tau_c$ is the inter-arrival time between the HOL customer and the following customer at node $c$, and $A_c(t) \in \{0,1\}$ is a (random) number of arrivals to node $c$ during time interval $[t, t+1)$.  We denote the average arrival rate to node $c$ by $\lambda_c$. Since arrivals are binary in each time step, $\lambda_c$ is also the probability of an arrival in a time step.  

A key advantage of the proposed approach is that \emph{stability of the waiting time process implies stability of the customer queuing process}.  We elaborate on this in Sec.~\ref{sec:stability} below. 
According to \eqref{qd}, a customer that arrives at time $t$  joins the queue at time $t+1$ with a waiting time of 1 [unit time].  We make the following observations about the dynamics:
\begin{enumerate}[label=(\roman*)]
	\setlength\itemsep{-.2em}
	\item If $\chi_c(t)=0$, then $H_c(t)=0$, $\tau_c=0$, and $H_c(t+1) = A_c(t)$.
	\item If $\chi_c(t)=1$, then $H_c(t) \in \ZZ_+$ and $\tau_c \in \ZZ_+$. It can be easily proven that inter-arrival times ($\tau_c$ in particular) are geometrically distributed with mean $\lambda_c^{-1}$.
	\item If $\chi_c(t) = 1$ and $x_c(t) = 0$, then $H_c(t+1) = H_c(t) + 1$. 
	\item If $\chi_c(t) = 1$ and $x_c(t) = 1$, then $H_c(t+1) = \big[H_c(t) + 1 - \tau_c \big]^+$.
\end{enumerate}
Observation (ii) also implies that the second moment of inter-arrival times is finite: 
\begin{equation}
	\EE \tau_c^2 = \frac{2 + \lambda_c}{\lambda_c^2}. \label{eq_secondMoment}
\end{equation}
This will be used to provide performance guarantees below.  
Observation (iii) says that if the HOL customer at node $c$ is not served at time $t$, then the HOL waiting time at $c$ is increased by 1 unit at time $t+1$. Point (iv) says that if the HOL customer at node $c$ is served at time $t$, then the HOL waiting time at node $c$ at time $t+1$ becomes the waiting time of the subsequent customer (which is the new HOL customer at time $t+1$) or 0 if the subsequent customer has not arrived (i.e., the inter-arrival time is greater than the waiting time of the HOL customer that is being served: $\tau_c > H_c(t) + 1$).  Equation \eqref{qd} and the observations (i)-(iv) above describe the dynamics of waiting \emph{before} customers enter service.  There is also a service time process, which depends on the time required for vehicles that have been assigned to reach their designated customers ($C_{vc}$ if $v$ is assigned to $c$).  We consider the ``service time'' in the assignment decisions below.

\subsection{Minimum drift plus penalty framework}
\label{ss:mdpp}
The minimum drift plus penalty (\textsf{MDPP}) framework proposed is applied as follows: At every time step $t$, observe the HOL waiting times $\HH(t) \equiv \{H_c(t): c \in \NX_{\cc}\}$ and solve the following optimization problem:
\begin{align}
\mbox{Minimize}  & \quad V \sum_{c \in \NX_{\cc}} \sum_{v \in \NX_{\vv}} C_{vc}(t) y_{vc}(t) - \sum_{c \in \NX_{\cc}} H_c(t) \sum_{v \in \NX_{\vv}} y_{vc}(t) \label{mdp_o} \\
\mbox{s.t.} \quad &  \quad \sum_{c \in \NX_{\cc}} y_{vc}(t) \leq 1, ~~ v \in \NX_{\vv} \label{mdp_c1} \\
&\quad \sum_{v \in \NX_{\vv}} y_{vc}(t) \leq 1, ~~  c \in \NX_{\cc} \label{mdp_c2} \\
&\quad y_{vc}(t) \in \{0,1\}, ~~ c \in \NX_{\cc}, v \in \NX_{\vv} \label{mdp_c3}
\end{align}
where $V \ge 0$ is a penalty constant associated with the dispatch cost and $y_{vc}(t)$ is a binary decision variable which equals 1 if vehicle $v$ is assigned to serve the HOL customer at node $c$ at time step $t$ and 0 otherwise. $C_{vc}(t)$ is the arc weight (cost) associated with arc $(v,c) \in \AX$, which is the travel time from $v$ to $c$ and (as described above) can include charging time, if charging is required before customer pick-up. Inequality \eqref{mdp_c1} ensures that one vehicle serves at most one customer in each time step and \eqref{mdp_c2} ensures that each customer is served by at most one vehicle in each time step.  The objective function \eqref{mdp_o} is interpreted as one that aims to assign vehicles to customers in a way that minimizes dispatch costs ($C_{vc}$), while giving priority to HOL customers with the highest waiting times.  We denote the solution of \eqref{mdp_o}--\eqref{mdp_c3} by $\{y^{\mathsf{MDPP}}_{vc}(t)\}$ and the binary assignment decisions obtained are 
\begin{equation}
\label{x_y}
x^{\mathsf{MDPP}}_c(t) = \sum_{v \in \NX_{\vv}} y^{\mathsf{MDPP}}_{vc}(t).
\end{equation}
Hence, the optimization problem \eqref{mdp_o} - \eqref{mdp_c3} produces the control variables $\{x^{\mathsf{MDPP}}_c(t)\}_{c \in \NX_{\cc}}$ which impact the network dynamics via \eqref{qd}. 

A disadvantage of making assignment decisions based on HOL waiting times is that this overlooks the possibility that a non-HOL customer at one station may have been waiting longer than a HOL customer at another. The formulation \eqref{mdp_o} - \eqref{mdp_c2} gives priority to the HOL customers regardless of how long non-HOL customers may have been waiting at other stations.  There also exist computational challenges: despite its tractable form\footnote{The constraints \eqref{mdp_c1} - \eqref{mdp_c3} have two features: (i) the right hand side constants are all integers (namely, ones) and (ii) the matrix formed from the left hand side coefficients is totally unimodular. Consequently, the linear relaxation $\mathbf{y}(t) \in [0,1]^N$, where $\mathbf{y}(t) \equiv \{y_{vc}(t)\}$ and $N \equiv |\NX_{\cc}| \times |\NX_{\vv}|$, produces binary solutions.}, the problem \eqref{mdp_o} - \eqref{mdp_c2} does not naturally decompose by node\footnote{Similar techniques have recently been employed in traffic signal control, there the problems naturally decompose by node \citep{li2019position} rendering them quite attractive in that setting.} so that for large networks, obtaining solutions in a real-time framework can be prohibitive.  We simultaneously overcome these two shortcomings by updating the solution frequently, i.e., an online approach.  The trick we employ is one where we update the solution in \emph{multiples of} a sufficiently small time increment $\Delta t$.  For example, setting
\begin{equation}
	\Delta t \equiv \Big( \sum_{c \in \NX_{\cc}} \lambda_c \Big)^{-1}
\end{equation}
we have that the probability that the number of arrivals to any node in the network exceeds 1 over a time interval of length $\Delta t$ is $o(\Delta t)$.  Thus, in each time step we can safely assume that either one customer arrives somewhere in the network or one vehicle is returned to the network (when one vehicle drops off its customer at the destination). In essence, we propose to operate in continuous time and we show below that the updates (naturally) need only take place when certain discrete events occur.

\textbf{Solution approach.} The main advantage of operating in continuous time is that \emph{two or more events occur simultaneously with probability zero}.  We exploit this below.  First, the objective function \eqref{mdp_o} can equivalently be written as a maximization objective:
\begin{equation}
	\mbox{Maximize } \sum_{c \in \NX_{\cc}} \sum_{v \in \NX_{\vv}} \big(H_c(t) - VC_{vc}(t) \big) y_{vc}(t),
\end{equation}
and we immediately see that $y_{vc}(t) = 0$ whenever $H_c(t) < VC_{vc}(t)$. This is interpreted as vehicle $v$ not being assigned to customer-charge node $c$ if the penalized dispatch cost exceeds the current waiting time of the customer.  Hence, \emph{the zeros in the solution are determined without solving the problem.} (When $H_c(t) = VC_{vc}(t)$ it is reasonable to set $y_{vc}(t) = 0$ to allow for the possibility that a better assignment can present itself at time $t + \Delta t$. This comes at a negligible cost since $\Delta t$ is small.)  One can interpret a viable assignment as one where the dispatch cost is directly proportional to the waiting time of the customer, with $V$ as the constant of proportionality.  When $VC_{vc}(t) > H_c(t)$, the assignment $(v,c)$ is deemed nonviable and customer $c$ might be better off waiting a little longer, allowing for the possibility of another vehicle being returned to the system that can be dispatched to $c$ at lower cost.  The variable $V$ should be fine-tuned; we compare the performance of the system under different values of $V$ in Sec.~\ref{S:simulation}.

Suppose that customers begin to arrive to the system at time $t=0$ and let time $t = t^{\mathsf{fp}} > 0$ mark the instant at which  $H_c(t^{\mathsf{fp}}) > VC_{vc}(t^{\mathsf{fp}})$ for some node $c$ (a \emph{first passage time}).  Prior to time $t^{\mathsf{fp}}$ no vehicles are in use in the system; hence, this passage event is triggered by the waiting time of some customer exceeding the threshold.  Since two or more events occur simultaneously with probability zero, the set
\begin{equation}
	\{c \in \NX_{\cc}: H_c(t^{\mathsf{fp}}) > VC_{vc}(t^{\mathsf{fp}}) \mbox{ for any } v \in \NX_{\vv}\}
\end{equation}
is a singleton set. Denoting this singleton node by $c^*$, the assignment problem \eqref{mdp_o} - \eqref{mdp_c2} at time $t^{\mathsf{fp}}$ simplifies to a 0-1 knapsack problem:
\begin{align}
\underset{y_{vc^*}(t^{\mathsf{fp}}) \in \{0,1\},~ v \in \NX_{\vv}}{\mbox{Maximize}}  & \quad \sum_{v \in \NX_{\vv}} \Big(H_{c^*}(t^{\mathsf{fp}}) - VC_{vc^*}(t^{\mathsf{fp}}) \Big) y_{vc^*}(t^{\mathsf{fp}}) \label{ks_o} \\
\mbox{s.t.} \qquad \quad & \quad \sum_{v \in \NX_{\vv}} y_{vc^*}(t^{\mathsf{fp}}) \le 1.  \label{ks_c} 
\end{align}
The problem \eqref{ks_o} - \eqref{ks_c} is solved by simply sorting the positive objective function coefficients ($H_{c^*}(t^{\mathsf{fp}}) - VC_{vc^*}(t^{\mathsf{fp}})$) and setting $y_{vc^*}(t^{\mathsf{fp}}) = 1$ for the element with the largest objective function coefficient.  This corresponds to assigning the vehicle with smallest dispatch cost, that is, letting $v^*$ be the optimal vehicle to assign to node $c^*$, we have that $y_{v^*c^*}(t^{\mathsf{fp}}) = 1$, where
\begin{equation}
\label{solveV*}
	v^* = \underset{v \in \NX_{\vv}}{\arg \min} ~ C_{vc^*}(t^{\mathsf{fp}}).
\end{equation}
This is even easier than sorting with a linear time complexity of $O(|\NX_{\vv}|)$, but since the system is tracking the passage events, it is already known which vehicle $v^*$ is and the time-complexity of this operation is practically $O(1)$.   It is important to note that it may appear that the assignment is decided entirely by the dispatch costs $\{C_{vc^*}(t^{\mathsf{fp}})\}_{v \in \NX_{\vv}}$ but this is not true. $H_{c^*}(t^{\mathsf{fp}})$ and $V$ still have an important role to play in deciding which vehicles to consider in the assignment, those that result in non-positive objective coefficients are not considered in the optimization problem.
	
The instant that customer $c^*$ enters service, the system reverts to one where $H_c(t) < VC_{vc}(t)$ and a certain amount of time will need to elapse (albeit small for a congested network) until the next event (e.g., a new arrival or a vehicle being returned to the system).  The event of interest is when the system makes the next passage into a state where $H_c(t) > VC_{vc}(t)$ 
for some $(v,c)$-pair.  Denote this subsequent passage time by $t^{\mathsf{sp}}$. The subsequent passage event can be triggered by either a passage of the waiting time of some HOL customer or a new vehicle being returned to the system.  In the first case, the time-complexity is $O(1)$ as in the first passage event.  If this new passage time is triggered by a new vehicle entering the system, $v^*$ is known and is the new vehicle in the system (again, since two or more events occur simultaneously with probability zero). In other words, 
\begin{equation}
	\{v^* \in \NX_{\vv}: H_c(t^{\mathsf{sp}}) > VC_{v^*c}(t^{\mathsf{sp}}) \mbox{ for any } c \in \NX_{\cc}\}
\end{equation}
is a singleton set and $v^*$ is known.  The assignment problem \eqref{mdp_o} - \eqref{mdp_c2} at time $t^{\mathsf{sp}}$ simplifies to a 0-1 knapsack problem:
\begin{align}
\underset{y_{v^*c}(t^{\mathsf{sp}}) \in \{0,1\},~ c \in \NX_{\cc}}{\mbox{Maximize}}  & \quad \sum_{c \in \NX_{\cc}} \Big(H_c(t^{\mathsf{sp}}) - VC_{v^*c}(t^{\mathsf{sp}}) \Big) y_{v^*c}(t^{\mathsf{sp}}) \label{kss_o} \\
\mbox{s.t.} \qquad \quad & \quad \sum_{c \in \NX_{\cc}} y_{v^*c}(t^{\mathsf{sp}}) \le 1.  \label{kss_c} 
\end{align}
This is also a 0-1 knapsack problem, and since the right-hand side of \eqref{kss_c} is 1, it can be simply solved by setting $y_{v^*c^*}(t^{\mathsf{sp}})=1$ for customer $c^*$ with the largest objective coefficient:
\begin{equation}
\label{solveC*}
	c^* = \underset{c \in \NX_{\cc}}{\arg \max} ~ \big(H_c(t^{\mathsf{sp}}) - VC_{v^*c}(t^{\mathsf{sp}}) \big).
\end{equation}
This operation has a time-complexity of $O(|\NX_{\cc}|)$.  This is also the maximal time-complexity of any assignment operation , since (in general) $|\NX_{\cc}| > |\NX_{\vv}|$. We summarize the overall framework of the MDPP solution approach as a flowchart shown in Fig.~\ref{F:flowchart}.

It is worth mentioning that we do not consider vehicle rebalancing in the MDPP framework. We leave vehicle rebalancing to future research, and focus on the joint optimization of vehicle assignment and recharging in this paper.

\begin{figure}[h!]
	\scriptsize
	\tikzstyle{startstop} = [rectangle, rounded corners, draw,thin,align=center,text centered,minimum width=1cm,minimum height=0.6cm]
	\tikzstyle{format}=[rectangle,draw,thin,align=center,minimum width=1cm,minimum height=0.6cm]%
	\begin{center}
		\begin{tikzpicture}[node distance=2cm,
		auto,>=latex',
		thin,
		start chain=going below,
		every join/.style={norm},]
		\node[draw, diamond, aspect=3.5, text width=2.9cm, align=center, yshift= 0.5cm](C){\footnotesize $\exists~(v,c)$-pair: $H_c(t) > VC_{vc}(t)$?};
		\node[draw, diamond, aspect=2.0,below of=C,text width=2.5cm,align=center,yshift=-0.5cm] (D){\footnotesize Is $v$ a newly returned vehicle?};
		\node[below of=D,format, text width=4.2cm, yshift=-0.15cm, xshift = -3cm](E){\footnotesize Mark the customer $c$ as $c^*$, and solve \eqref{solveV*} to get $v^*$};
		\node[below of=D,format, text width=4.0cm,yshift=-0.15cm,xshift = 3cm](F){\footnotesize Mark the vehicle $v$ as $v^*$, and solve \eqref{solveC*} to get $c^*$};
		\node[below of=D,yshift=-1.5cm,circle, draw,align=center](G){};
		\node[below of=G,format,yshift=1.2cm](H){\footnotesize Assign $v^*$ to $c^*$};
		\node[below of=H,format,yshift=1.0cm](I){\footnotesize $t \mapsfrom$ time of next event};
		\draw[->] (C.south) -- node[anchor=west]{\footnotesize yes} (D.north);
		\draw[->] (C.west) -- node[anchor=south]{\footnotesize no}(-5.7,0.5) -- (-5.7,-7.3) -- (I.west);
		\draw[->] (D.west) -- node[anchor=south]{\footnotesize no}(-3,-2.0) -- (E.north);
		\draw[->] (D.east) -- node[anchor=south]{\footnotesize yes}(3,-2.0) -- (F.north);
		\draw[->] (E.south) -- (-3,-5.5) -- (G.west);
		\draw[->] (F.south) -- (3,-5.5) -- (G.east);
		\draw[->] (G.south) -- (H.north);
		\draw[->] (H.south) -- (I.north);
		\draw[->] (I.east) -- node[anchor=south]{\footnotesize}(5.7,-7.3) -- (5.7,0.5) -- (C.east);
		\end{tikzpicture}
	\end{center}
	\caption{Flowchart of the MDPP solution approach.}\label{F:flowchart}
\end{figure}
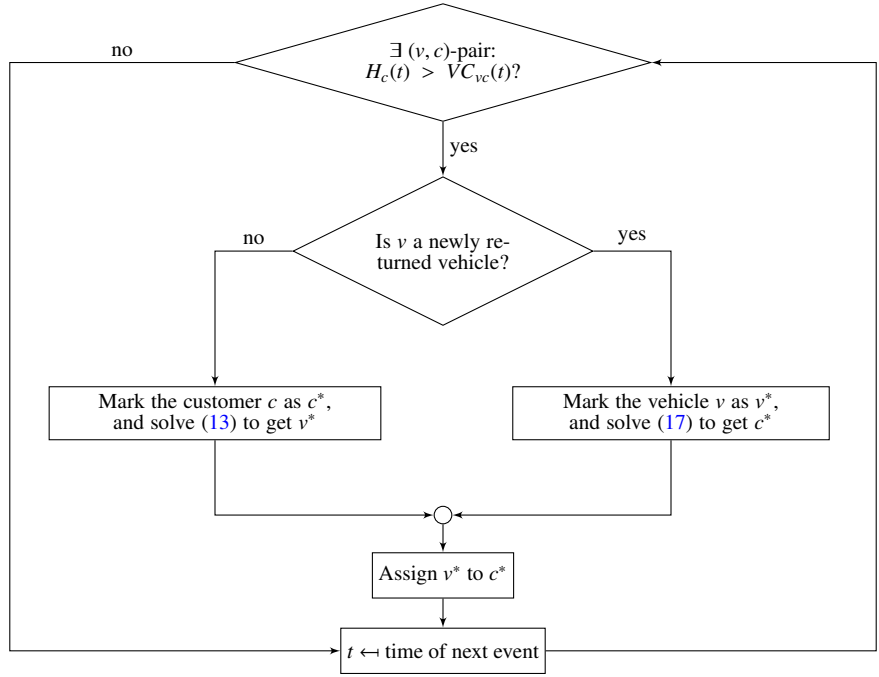

\section{Performance guarantees} \label{sec:stability}
This section provides guarantees of performance of the proposed approach in terms of stability of waiting times and a bound on the deviation from optimal time-averaged dispatch costs.  Here, stability is interpreted as customer waiting times not growing without bounds in the limit.  We formally define it next.

\begin{definition}[Stability] \label{def_stability}
The SAEV network is said to be \textit{strongly stable} if %\citep[Definition 2.7]{neely2010stochastic}:
\begin{equation}
	\underset{T \rightarrow \infty}{\lim \sup} \frac{1}{T}\sum_{t=0}^{T-1}\sum_{c \in \NX_{\cc}}\EE |H_c(t)| < \infty. \label{eq_stability}
\end{equation}
\end{definition}
\noindent Since $H_c(t) \ge 0$ with probability 1, we shall drop the modulus ($|\cdot|$) below.

\begin{definition}[Maximal Throughput Region] \label{d:mtr}
	The maximal throughput region of the network, denoted by $\boldsymbol{\Lambda}$, is defined as the closure, i.e., the convex hull, of the set of all arrival rate vectors $\boldsymbol{\lambda} \in \RR_+^{|\NX_{\cc}|}$ for which there exists a stabilizing scheduling algorithm.
\end{definition}

Definition \ref{d:mtr} implies that if an arrival rate vector lies outside of $\boldsymbol{\Lambda}$, then there does not exist a scheduling algorithm capable of stabilizing the network. It is hence natural to assume that the arrival rate lies in $\boldsymbol{\Lambda}$. In the present system, \emph{service} is determined by the SAEV vehicle fleet.  For any customer demand, $\boldsymbol{\lambda}$, as long as there exists a way to assign vehicles in the fleet to satisfy the demand, then $\boldsymbol{\lambda} \in \boldsymbol{\Lambda}$.  Otherwise, there does not exist an algorithm or schedule capable of satisfying such demand and $\boldsymbol{\lambda} \not\in \boldsymbol{\Lambda}$.  This implies that $\boldsymbol{\Lambda}$ is \emph{not} empty as long there are vehicles in the fleet.  In other words, $\boldsymbol{\Lambda}$ is the set of all demands that can be served by the SAEV vehicle fleet.
Denote by $\textbf{y}^{\mathsf{MDPP}}(t)$ the solution vector produced by our algorithm, i.e., the solution of \eqref{mdp_o} - \eqref{mdp_c3} and let $C^{\mathsf{MDPP}}(t)$ denote the network-wide vehicle dispatch costs under the \textsf{MDPP} algorithm at time $t$, $C^{\mathsf{MDPP}}(t)$ is expressed as:
\begin{equation}
C^{\mathsf{MDPP}}(t)=\sum_{c \in \NX_{\cc}} \sum_{v \in \NX_{\vv}} C_{vc}(t)y_{vc}^{\mathsf{MDPP}}(t).
\end{equation}
%
%\begin{pf}
	Our main result will be proved using Lyapunov stability techniques. We first define the vector of HOL waiting times, the Lyapunov function ($L$), and the Lyapunov drift ($\triangle$), respectively as
	\begin{equation}
	\HH(t) = [H_1(t) \cdots H_c(t) \cdots H_{|\NX_{\cc}|}(t)]^{\top},
	\end{equation}
%	the Lyapunov function ($L$)
	\begin{equation}
	\label{lyap}
	L\big(\HH(t)\big) \equiv \frac{1}{2}\sum_{c \in \NX_{\cc}}\lambda_c H_c(t)^2,
	\end{equation}
	and %the Lyapunov drift ($\triangle$) 
	\begin{equation}
	\label{drift}
	\triangle (\HH(t)\big) \equiv \EE \Big(L\big(\HH(t+1)\big) - L\big(\HH(t)\big) \big| \HH(t) \Big).
	\end{equation}
	We introduce the \textsf{S-only} algorithm \citep{neely2006energy,neely2013delay} as a baseline algorithm for our proofs.  We denote by $\mathbf{C}(t)$ the vector of travel costs for all $(v,c)$ pairs at time $t$. The 'S' in \textsf{S-only} refers to `service times' and  \textsf{S-only} algorithms assume that the system knows the entire probability distribution of the network service times for all $t$. That is, the probability distribution of network travel costs, represented by the density function $\pi_{\mathbf{C}}(\mathbf{c})$, is assumed to be known to the system operator, in contrast to the assumption that only point estimates are known to the operator in the proposed \textsf{MDPP}\footnote{The latter are easier to obtain: navigation systems that are (ubiquitously) available in vehicles and mobile devices today are capable of producing such estimates.}. \textsf{S-only} algorithms produce randomized scheduling policies that can stabilize the system.  That is, they seek to determine an \emph{optimal} probability distribution over the set of assignment decisions $\mathcal{X} \subseteq \{0,1\}^{\NX_{\cc}}$, denoted $\pi^*(\mathbf{x}|\mathbf{c})$, and find a constant $\epsilon^* \ge 0$ so that for any $\boldsymbol{\lambda} \in \boldsymbol{\Lambda}$
	\begin{equation}
		\EE \mathbf{x}^{\mathsf{S-only}}(t) = \sum_{\mathbf{x} \in  \mathcal{X}} \int_{\mathbf{c}}\mathbf{x} \pi^*(\mathbf{x}|\mathbf{c}) \pi_{\mathbf{C}}(\mathbf{c}) \mathrm{d}\mathbf{c} \ge \boldsymbol{\lambda} + \epsilon^* \mathbf{1}_{|\NX_{\cc}|}, \label{Tpf1}
	\end{equation}
	where the inequality is component-wise and $\mathbf{1}_{|\NX_{\cc}|}$ is a column vector of ones of size $|\NX_{\cc}|$.  If $\pi^*(\mathbf{x}|\mathbf{c})$ and $\epsilon^* \ge 0$ exist, then an \textsf{S-only} algorithm that observes $\mathbf{C}(t)$ and randomly selects a decision vector by sampling from $\pi^*\big(\mathbf{x} \big| \mathbf{c} = \mathbf{C}(t)\big)$ is guaranteed to stabilize the network waiting times \citep{neely2006energy,neely2010stochastic}.  We define the total system cost that is achieved with the \textsf{S-only} algorithm as
	\begin{equation}
		C^{\Sonly}(t) = \sum_{c \in \NX_{\cc}} \sum_{v \in \NX_{\vv}} C_{vc}^{\Sonly}(t)y_{vc}^{\Sonly}(t)
	\end{equation}
	and the long-run mean system cost under \textsf{S-only} scheduling as 
	\begin{equation}
		\overline{C}^{\Sonly} = \underset{T \rightarrow \infty}{\lim} \frac{1}{T} \sum_{t=0}^{T-1} \EE C^{\Sonly}(t). \label{eq_sonly}
	\end{equation}
	The total system cost $\overline{C}^{\Sonly}$ is interpreted as the system cost that can be achieved if (i) the entire probability distribution $\pi_{\mathbf{C}}$ was known to the operator (not just an estimate), and (ii) one is capable of finding an \emph{optimal} probability distribution from which to sample.  Both of these assumptions are prohibitive from a practical standpoint, especially for purposes of real-time scheduling.  While $\overline{C}^{\Sonly}$ may not be the globally minimal mean system cost that can be achieved by any scheduling algorithm, it is known to produce the smallest \emph{average} systems cost over all stabilizing policies \citep{neely2006energy}. It, thus, serves as an attractive target system cost to achieve.  
	
	\begin{theorem}[Stability of MDPP]
		\label{thm:mdp}
		Assume that arrival rates at time $t$, $\boldsymbol{\lambda}(t)$, lie in $\boldsymbol{\Lambda}$ and that $\EE L(\HH(0)) < \infty$.  Then applying the solution to \eqref{mdp_o} - \eqref{mdp_c3}, there exist finite constants $0 < K < \infty$ and $0 < \epsilon^* < \infty$ such that the time-averaged expected vehicle dispatch cost and waiting time of HOL customers in the network satisfy
		\begin{flalign}
		\mbox{(i) } \qquad &  \underset{T \rightarrow \infty}{\lim \sup} \frac{1}{T} \sum_{t=0}^{T-1} \EE C^{\mathsf{MDPP}}(t) \leq \overline{C}^{\Sonly} + \frac{K}{V}, \mbox{ and} \label{eq_thm1}&&\\
		\mbox{(ii) } \qquad &  \underset{T \rightarrow \infty}{\lim \sup} \frac{1}{T}\sum_{t = 0}^{T-1} \sum_{c \in \NX_{\cc}} \EE H_c(t) \leq \frac{K +  (\overline{C}^{\Sonly} - C^\mathsf{min}) V }{\epsilon^*}, \label{eq_thm2}&&
		\end{flalign}
		where $C^\mathsf{min}$ is the minimum system dispatch cost that can be achieved by any scheduling policy.
	\end{theorem}
	
	The second part of Theorem \ref{thm:mdp}, \eqref{eq_thm2}, states that as long as there exists a way to stabilize the network (the arrival rates are in $\boldsymbol{\Lambda}$), setting $V<\infty$ and solving the optimization problem \eqref{mdp_o} - \eqref{mdp_c3}, one ensures stability of the waiting times in the network (see Definition \ref{def_stability}).  The theorem also says that there exists a compromise between the dispatch cost and the queuing stability that depends on the value of $V$.  As $V$ gets large, the dispatch cost associated with the \textsf{MDPP} policy approaches the target dispatch cost $\overline{C}^{\Sonly}$ but the long-run expected waiting times may increase as a result.  Decreasing $V$ has the opposite effect.  Also, the proposed policy does not require knowledge of the demands or even the arrival rates; as long the latter lie in the maximal throughput region of the network ($\boldsymbol{\Lambda}$), stability of the waiting times (and the system travel costs) is ensured.  One only requires knowledge of the present state of the system and estimates of system travel costs, which is information that is typically available to SAEV operators.  
	
	The minimum dispatch cost $C^\mathsf{min}$ is one that would be achieved under any information setting (e.g., known demands, full history, and even future arrivals).  If target cost $\overline{C}^{\Sonly}$ approaches $C^\mathsf{min}$, one can set $V$ to a larger value while maintaining low waiting times, that is, in this case, there is no compromise between travel/dispatch costs and customer waiting times.
	
	Since the random arrival processes $\{A_c\}_{c \in \NX_{\cc}}$ have at most one arrival per time step, the customer queue sizes are always no greater than the waiting times of their corresponding HOL customers. Therefore, the stability of the HOL customers' waiting times implies the stability of customer queues. In other words, the proof of \eqref{eq_thm2} also proves the stability of customer queues throughout the network. The proof of Theorem \ref{thm:mdp} relies on an inequality, which we state as a Lemma (Lemma \ref{Lem2}) next.  Note that the inequality holds for any scheduling policy and any arrival pattern. 
	
	\begin{lemma}
		\label{Lem2}
		There exists a constant $0<K<\infty$ such that
		\begin{equation}
		\label{lem2_eq}
		\triangle (\HH(t)\big) \leq K - \sum_{c \in \NX_{\cc}} H_c(t) \EE \big( x_c(t) - \lambda_c \big| \HH(t) \big)
		\end{equation}
		holds for all $t > 0$.
	\end{lemma}
	\begin{pf}
		From \eqref{qd}, we have that
		\begin{equation}
		\label{Lem1_p1}
		H_c(t+1)^2 \leq \chi_c(t)^2\big(H_c(t)+1-x_c(t)\tau_c\big)^2 +\big(1-\chi_c(t)\big)^2 A_c(t)^2  
		+ 2\chi_c(t) \big[H_c(t)+1- x_c(t)\tau_c\big]^+ \big(1-\chi_c(t)\big)A_c(t).
		\end{equation}
		Since $\chi_c(t)$ is a binary variable, the third term of the right side of \eqref{Lem1_p1} is equal to 0,  $\chi_c(t)^2 = \chi_c(t)$, and $(1-\chi_c(t))^2 = 1-\chi_c(t)$. Since $A_c(t)$ is also binary, $A_c(t)^2 = A_c(t)$. Moreover, we have that $\chi_c(t)H_c(t) = H_c(t)$. Hence, 
		\begin{equation}
		\label{Lem1_p1_1}
		H_c(t+1)^2  
		\le H_c(t)^2 + \chi_c(t)\big(1-x_c(t)\tau_c\big)^2 - 2H_c(t)\big(x_c(t)\tau_c-1\big) +\big(1-\chi_c(t)\big)A_c(t).
		\end{equation}	
		Then, from \eqref{lyap} we have that
		\begin{multline}
		\label{Lem1_p2}
		L\big(\HH(t+1)\big) - L\big(\HH(t)\big) = \frac{1}{2}\sum_{c \in \NX_{\cc}}\lambda_c H_c(t+1)^2 - \frac{1}{2}\sum_{c \in \NX_{\cc}}\lambda_c H_c(t)^2  \\
		\le \frac{1}{2}\sum_{c \in \NX_{\cc}}\lambda_c\Big(\chi_c(t)\big(1-x_c(t)\tau_c\big)^2 + \big(1-\chi_c(t)\big)A_c(t)\Big) - \sum_{c \in \NX_{\cc}}\lambda_c H_c(t)\big(x_c(t)\tau_c-1\big).
		\end{multline}
		Conditioning on $\mathbf{H}(t)$ and taking expectations, we have from definition \eqref{drift} that
		\begin{equation}
		\label{Lem1_p3}
		\triangle (\HH(t)\big) \le \EE \big(K(t) \big| \HH(t) \big) - \sum_{c \in \NX_{\cc}} \lambda_c H_c(t) \EE \big(x_c(t) \tau_c - 1 \big| \HH(t) \big),
		\end{equation}
		where 
		\begin{equation}
		\label{Lem1_p4}
		K(t) \equiv \frac{1}{2} \sum_{c \in \NX_{\cc}} \lambda_c \chi_c(t) \big( 1 - x_c(t) \tau_c \big)^2 + \frac{1}{2} \sum_{c \in \NX_{\cc}} \lambda_c \big(1-\chi_c(t)\big) A_c(t).
		\end{equation}
		Since $x_c(t)$ and $\chi_c(t)$ are both binary, we have that
		\begin{equation}
		K(t) \le \frac{1}{2} \sum_{c \in \NX_{\cc}} \lambda_c \big( 1 + \tau_c^2 + A_c(t)\big). \label{eq_bound}
		\end{equation}
		Recalling the finite second moment of inter-arrival times and \eqref{eq_secondMoment}, we define the finite constant
		\begin{equation}
		K \equiv \frac{1}{2} \sum_{c \in \NX_{\cc}} \lambda_c \EE \big( 1 + \tau_c^2 + A_c(t)\big) = \frac{1}{2} \sum_{c \in \NX_{\cc}} \Big( \lambda_c^2 + \lambda_c + \frac{2}{\lambda_c} + 1 \Big).
		\end{equation}
		That $0 < K < \infty$ follows since $0 < \lambda_c < 1$ for all $c$. Then from \eqref{eq_bound}, we have that
		\begin{equation}
		\EE \big(K(t) \big| \HH(t) \big) \le K. \label{eq_bound1}
		\end{equation}
		Turning to the second term on the right-hand side of \eqref{Lem1_p3}, we have that
		\begin{equation}
		\lambda_c \EE \big(x_c(t) \tau_c - 1 \big| \HH(t) \big) = \EE \big(x_c(t)  - \lambda_c \big| \HH(t) \big), \label{eq_bound2}
		\end{equation}
		since the decisions $\{x_c(t)\}$ are made independently of the inter-arrival times $\{\tau_c\}$ and $\EE \tau_c = \lambda_c^{-1}$. Finally from \eqref{eq_bound1} and \eqref{eq_bound2}, we have that \eqref{Lem1_p3} implies that
		\begin{equation}
		\triangle (\HH(t)\big) \le K - \sum_{c \in \NX_{\cc}} H_c(t) \EE \big(x_c(t) - \lambda_c \big| \HH(t) \big).
		\end{equation}
		This completes the proof. \qed
	\end{pf}
	 
\begin{pf1}
	Since the MDPP algorithm observes $\HH(t)$ and then selects an assignment strategy, $\mathbf{x}^{\mathsf{MDPP}}(t)$, that minimizes 
	\begin{equation}
		VC(t)-\sum_{c \in \NX_{\cc}}H_c(t)x_c(t),
	\end{equation}
	we have, for each time step $t$, that
	\begin{equation}
	\label{Tpf3}
	V\EE\big(C^{\mathsf{MDPP}}(t) \big| \HH(t) \big) - \sum_{c \in \NX_{\cc}} \EE\big(H_c(t)x_c^{\mathsf{MDPP}}(t) \big| \HH(t) \big) 
	\le \EE\big(VC^{\Sonly}(t) \big| \HH(t) \big) - \sum_{c \in \NX_{\cc}} \EE\big(H_c(t)x_c^{\Sonly}(t) \big| \HH(t) \big).
	\end{equation}
	Since the \textsf{S-only} algorithm makes decisions independently of the $\HH(t)$ and by the properties of conditional expectation, we can rewrite \eqref{Tpf3} as
	\begin{equation}
	\label{Tpf4}
	V\EE\big(C^{\mathsf{MDPP}}(t) \big| \HH(t) \big) - \sum_{c \in \NX_{\cc}} H_c(t) \EE\big(x_c^{\mathsf{MDPP}}(t) \big| \HH(t) \big)  
	\le V\EE C^{\Sonly}(t)  - \sum_{c \in \NX_{\cc}}H_c(t)\EE x_c^{\Sonly}(t).
	\end{equation}
	Since \eqref{lem2_eq} holds for any policy, adding $V\EE\big(C^{\mathsf{MDPP}}(t) \big| \HH(t) \big)$ to both sides of the inequality yields
	\begin{equation}
	\label{Tpf5}
	\triangle (\HH(t)\big)+V\EE\big(C^{\mathsf{MDPP}}(t) \big| \HH(t) \big) 
	\le K+V\EE\big(C^{\mathsf{MDPP}}(t) \big| \HH(t) \big) - \sum_{c \in \NX_{\cc}}H_c(t)\EE\big(x_c^{\mathsf{MDPP}}(t)-\lambda_c \big| \HH(t) \big).
	\end{equation}
	From \eqref{Tpf4} and \eqref{Tpf5}, and since $\lambda_c$ is independent of $\HH(t)$, we have that
	\begin{equation} 
	\triangle (\HH(t)\big)+V\EE\big(C^{\mathsf{MDPP}}(t) \big| \HH(t) \big) 
	\le K +  
	V\EE C^{\Sonly}(t)  - \sum_{c \in \NX_{\cc}}H_c(t)\EE \big(x_c^{\Sonly}(t) - \lambda_c \big). \label{Tpf6}
	\end{equation}
	Hence,
	\begin{equation}
	\label{Tpf7}
	\triangle (\HH(t)\big)+V\EE\big(C^{\mathsf{MDPP}}(t) \big| \HH(t) \big)
	 \le K + V\EE C^{\Sonly}(t)  - \epsilon^* \sum_{c \in \NX_{\cc}}H_c(t)
	\end{equation}
	by appeal to \eqref{Tpf1}.  Expanding the Lyapunov drift on the left-hand side of \eqref{Tpf7}, taking expectations, and summing both sides over $t=0,\hdots,T-1$: %to $t=T-1$, we have 
%	\begin{multline}
%	\label{eq_mainBound}
%	\EE L\big(\HH(t+1)\big) \big| \HH(t) \big) - \EE \big(L\big(\HH(t)\big) \big| \HH(t) \big) +V\EE\big(C^{\mathsf{MDPP}}(t) \big| \HH(t) \big) \\
%	\le K + V\EE C^{\Sonly}(t)  - \epsilon^* \sum_{c \in \NX_{\cc}}H_c(t).
%	\end{multline}
	\begin{equation}
	\label{eq_mainBound}
	\EE L\big(\HH(T)\big) - \EE L\big(\HH(0)\big) + \sum_{t=0}^{T-1}V \EE C^{\mathsf{MDPP}}(t)
	\le TK + V \sum_{t=0}^{T-1} \EE C^{\Sonly}(t)  - \epsilon^* \sum_{t=0}^{T-1} \sum_{c \in \NX_{\cc}} \EE H_c(t).
	\end{equation}
	(Recall that $\EE \EE( C^{\mathsf{MDPP}}(t)| \HH(t)) = \EE C^{\mathsf{MDPP}}(t)$ by the properties of conditional  expectation). We rearrange terms in \eqref{eq_mainBound} and divide both sides by $VT$ to obtain
	\begin{equation}
		\frac{1}{T} \sum_{t=0}^{T-1}\EE C^{\mathsf{MDPP}}(t) \le \frac{1}{T} \sum_{t=0}^{T-1} \EE C^{\Sonly}(t) + \frac{K}{V} + \frac{1}{VT}\EE L\big(\HH(0)\big),
	\end{equation}
	where we have dropped the first term on the left-hand side and the last term on the right-hand side since they are both non-negative and the inequality, thus, holds without them.  Taking the $\lim \sup$ as $T \rightarrow \infty$, we get the first result \eqref{eq_thm1}. 
	Next, we rearrange the terms in \eqref{eq_mainBound} in a different way and divide by $T\epsilon^*$ to get
	\begin{equation}
		\frac{1}{T} \sum_{t=0}^{T-1} \sum_{c \in \NX_{\cc}} \EE H_c(t)
		 \le \frac{K}{\epsilon^*} + \frac{V}{T\epsilon^*} \sum_{t=0}^{T-1} \EE \big( C^{\Sonly}(t) - C^{\mathsf{MDPP}}(t) \big) + \frac{1}{T\epsilon^*} \EE L\big(\HH(0)\big).
	\end{equation}
	Noting that $\frac{1}{T} \sum_{t=0}^{T-1} \EE C^{\mathsf{MDPP}}(t) \ge C^{\mathsf{min}}$ 
%	\begin{equation}
%		\frac{1}{T} \sum_{t=0}^{T-1} \EE C^{\mathsf{MDPP}}(t) \ge C^{\mathsf{min}}
%	\end{equation}
	for all $T$ and from \eqref{eq_sonly}, upon taking the $\lim \sup$ as $T \rightarrow \infty$ on both sides, we get the second result \eqref{eq_thm2}, which completes the proof. 
	\qed
\end{pf1}

\section{Illustrative examples}
\label{S:example}
Consider a simple grid network with two charging stations as shown in Fig.~\ref{F:ieV01}. Vehicles with insufficient charge will charge at the suitable charging station en route to pick up the customer.  Consider the following setting: Two vehicles are in the system at time $t=0$.  Vehicle 1 has 55\% charge and Vehicle 2 has 60\% charge.  The first customer arrives at time $t=0$ and requests a vehicle with a charge level of 30\% (or more).  Two other customers arrive at $t=5$ and $t=13.6$ minutes, and request vehicles with charge levels of (no less than) 45\% and 80\%, respectively. Three vehicles are returned to the system at $t=5.6$, $t=15.8$, and $t=18$ minutes, their charge levels are 52\%, 93\%, and 90\%, respectively.  The dispatch costs $\mathbf{C} = \{C_{vc}\}$ are constant and given (in minutes) by
\begin{equation}
	\mathbf{C} = \begin{bmatrix}
	30 & 25 & 40 \\
	15 & 20  & 35 \\
	 18 & 4  &  NA \\
	 NA & NA & 22  \\
	 NA & NA &  3
	\end{bmatrix}.
\end{equation}
%Each row corresponds to a vehicle, each column corresponds to a customer, and 
``NA'' are costs that are not needed for this example.  The sequence of events are shown in Fig.~\ref{F:ieV01} when we set $V = 0.1$.
\begin{figure}[h!] 
	\centering
	
	\resizebox{1.0\textwidth}{!}{%
		\includegraphics[scale=0.17]{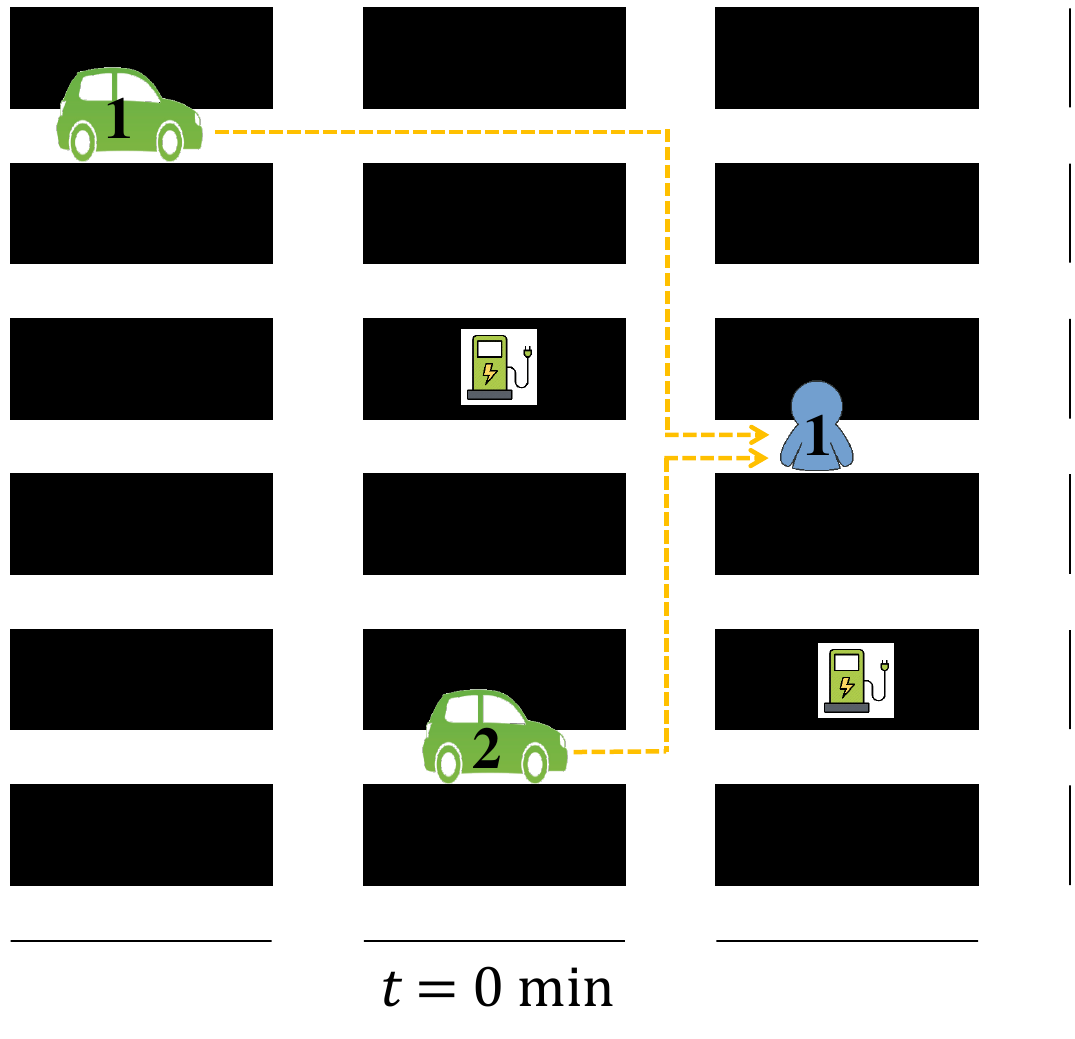}\hspace{0.2in}
		\includegraphics[scale=0.17]{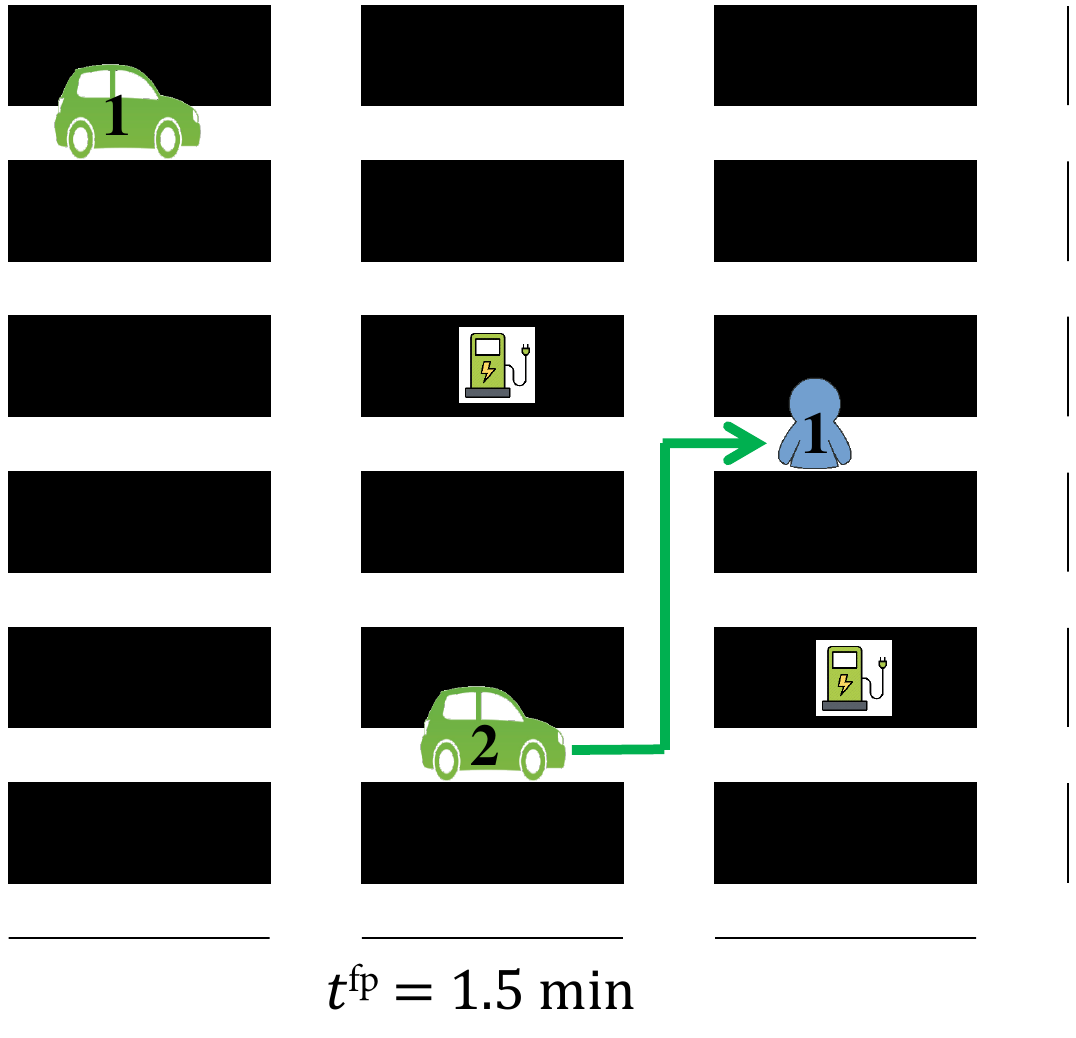}\hspace{0.2in}
		\includegraphics[scale=0.17]{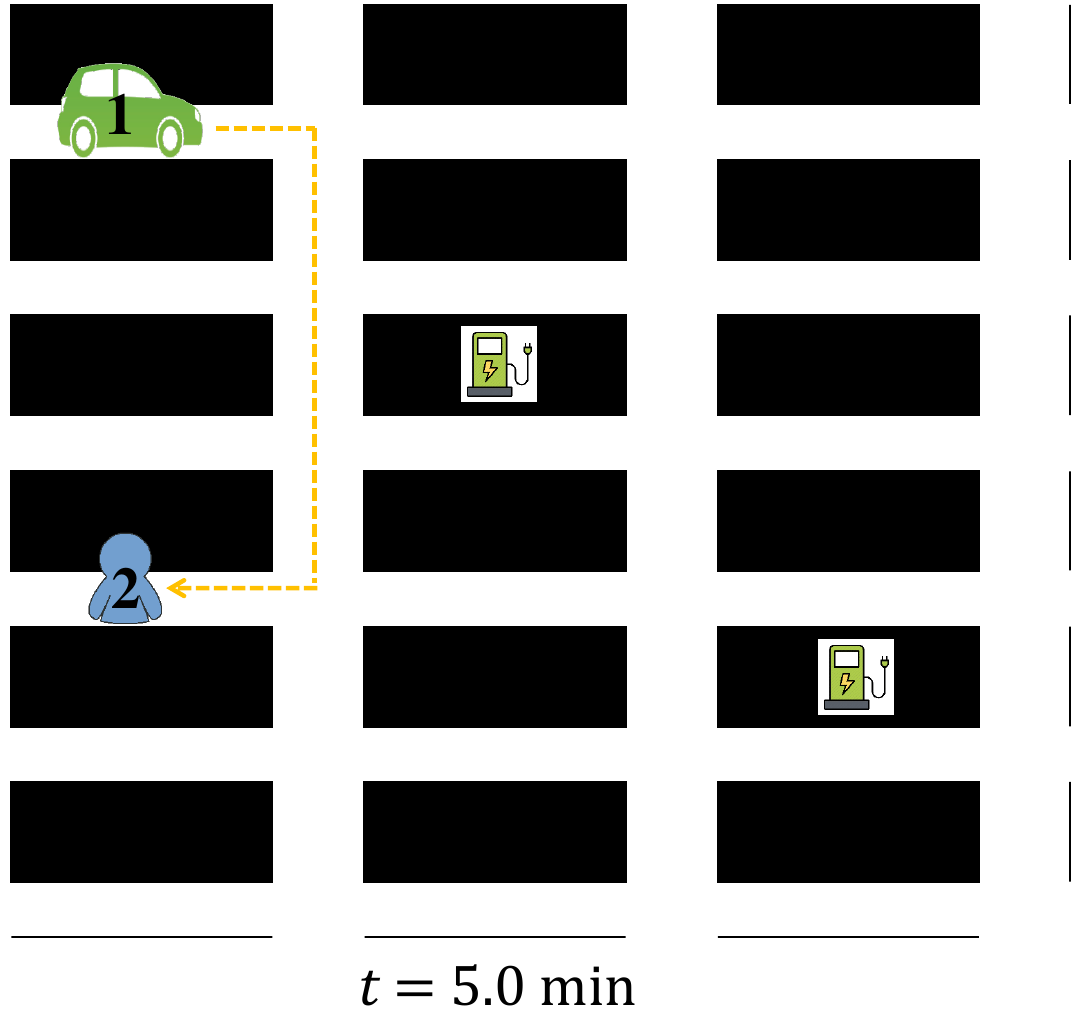}}	
	
	\resizebox{1.0\textwidth}{!}{%
		\includegraphics[scale=0.17]{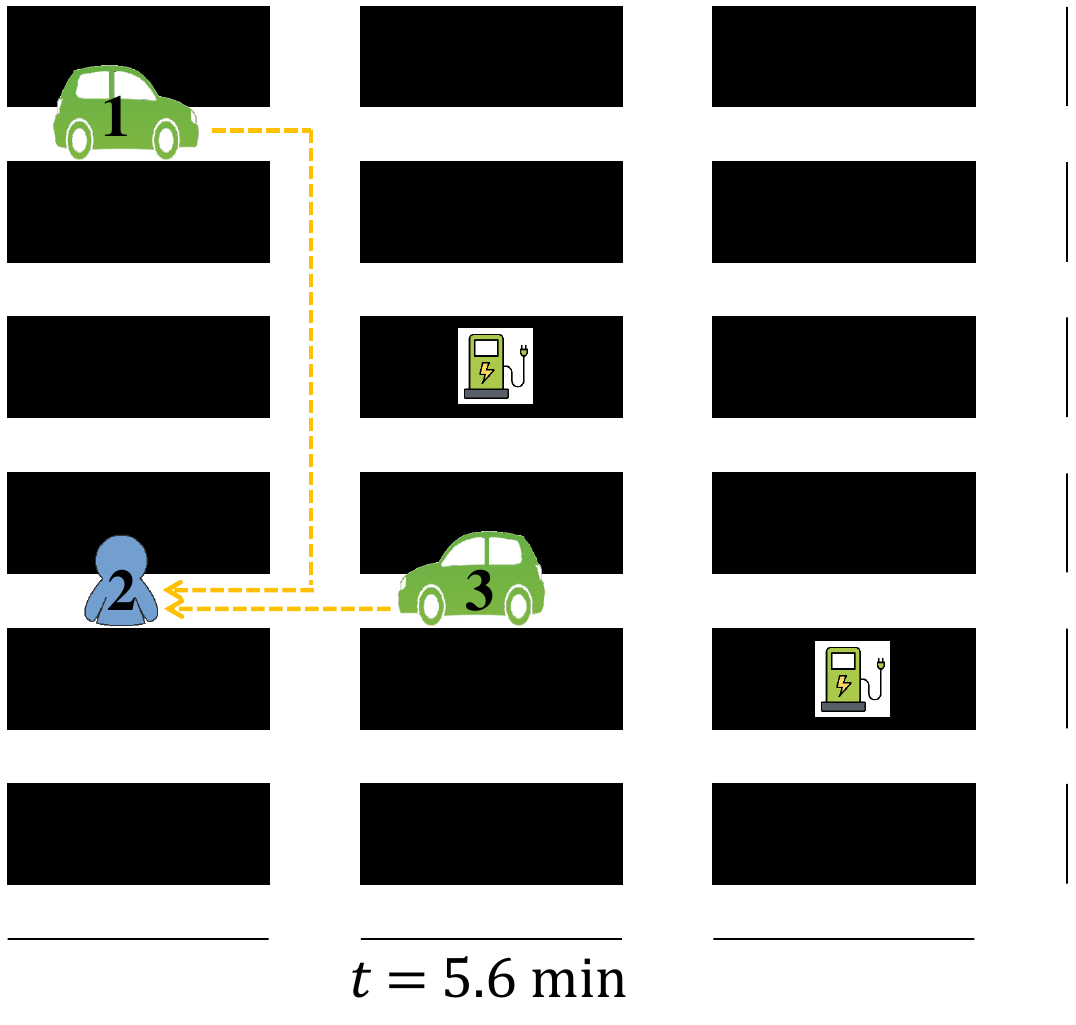}\hspace{0.2in}
		\includegraphics[scale=0.17]{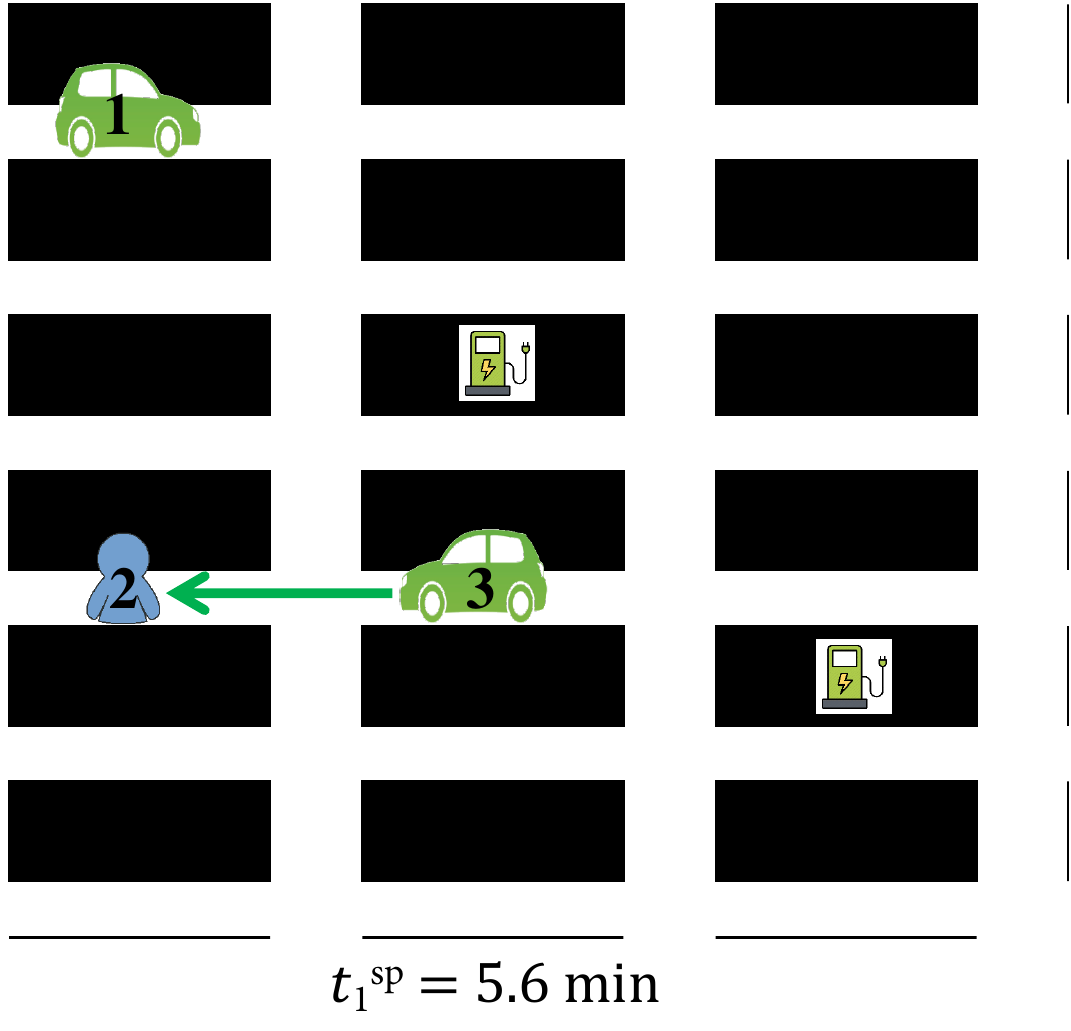}\hspace{0.2in}
		\includegraphics[scale=0.17]{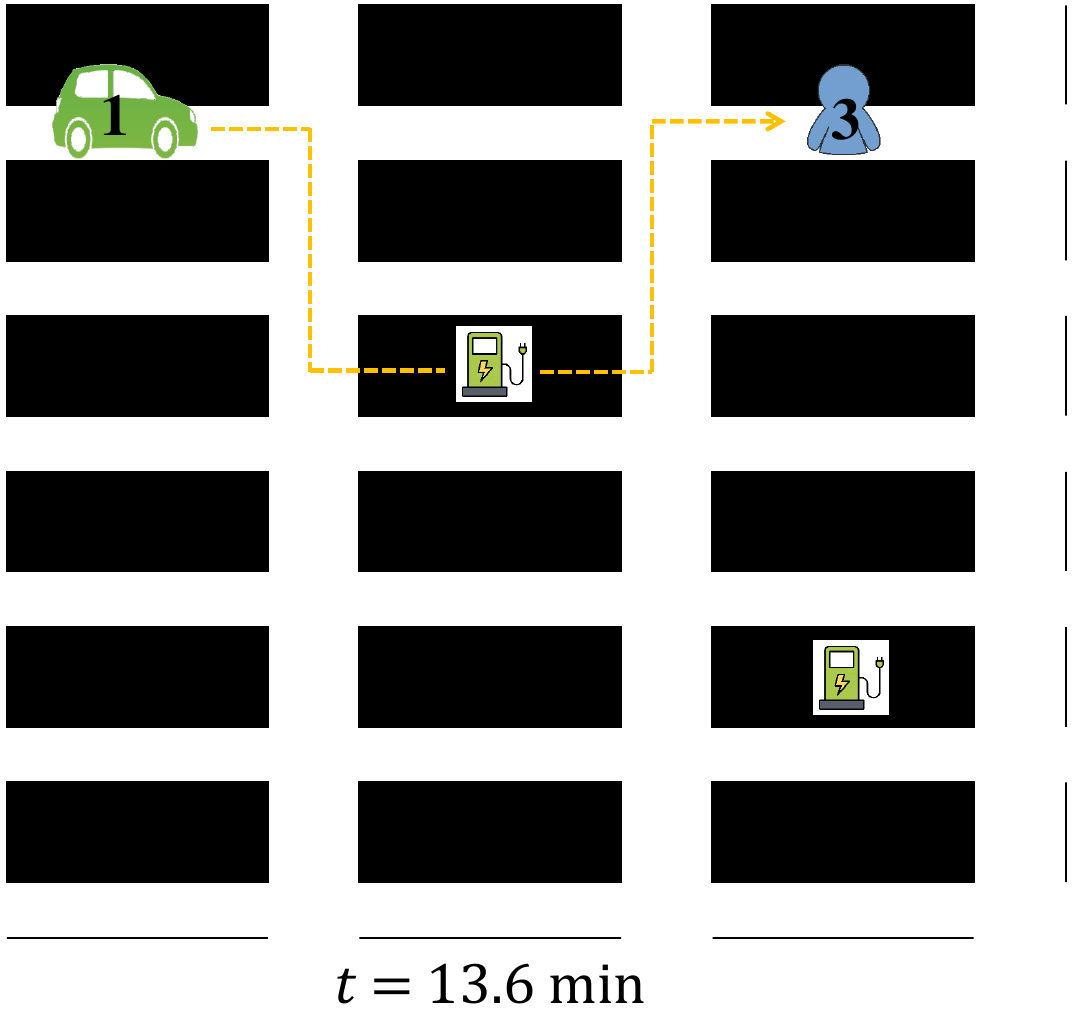}}
	
	\resizebox{1.0\textwidth}{!}{%
		\includegraphics[scale=0.17]{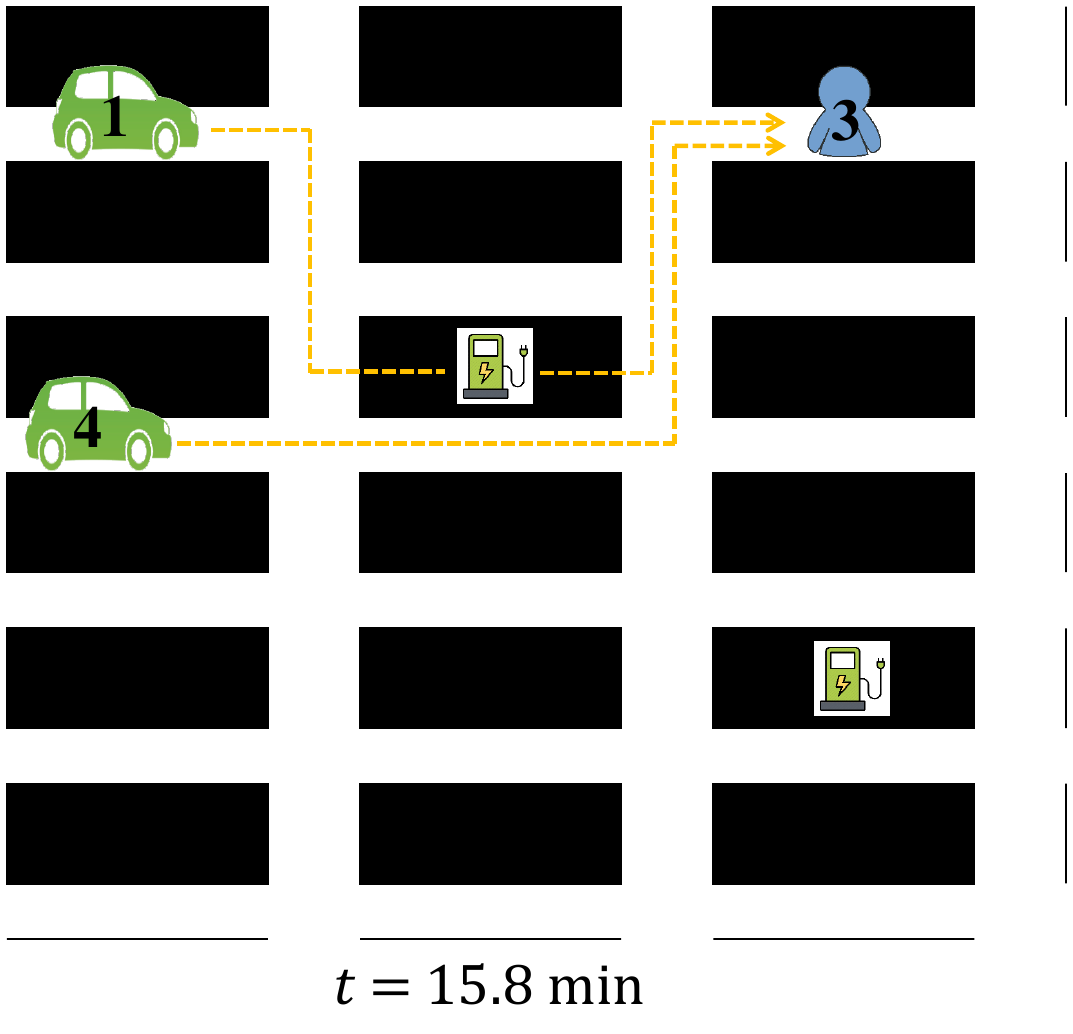}\hspace{0.2in}
		\includegraphics[scale=0.17]{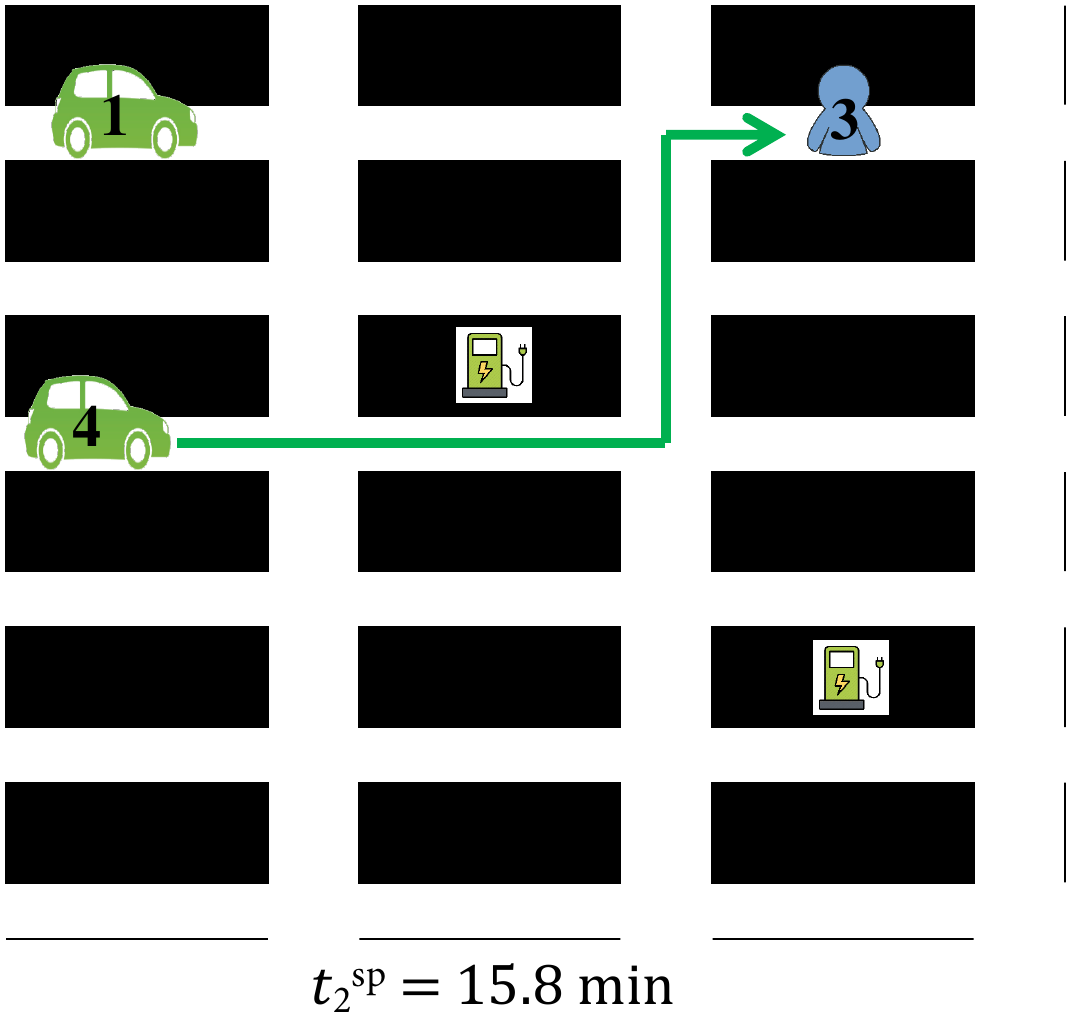}\hspace{0.2in}
		\includegraphics[scale=0.17]{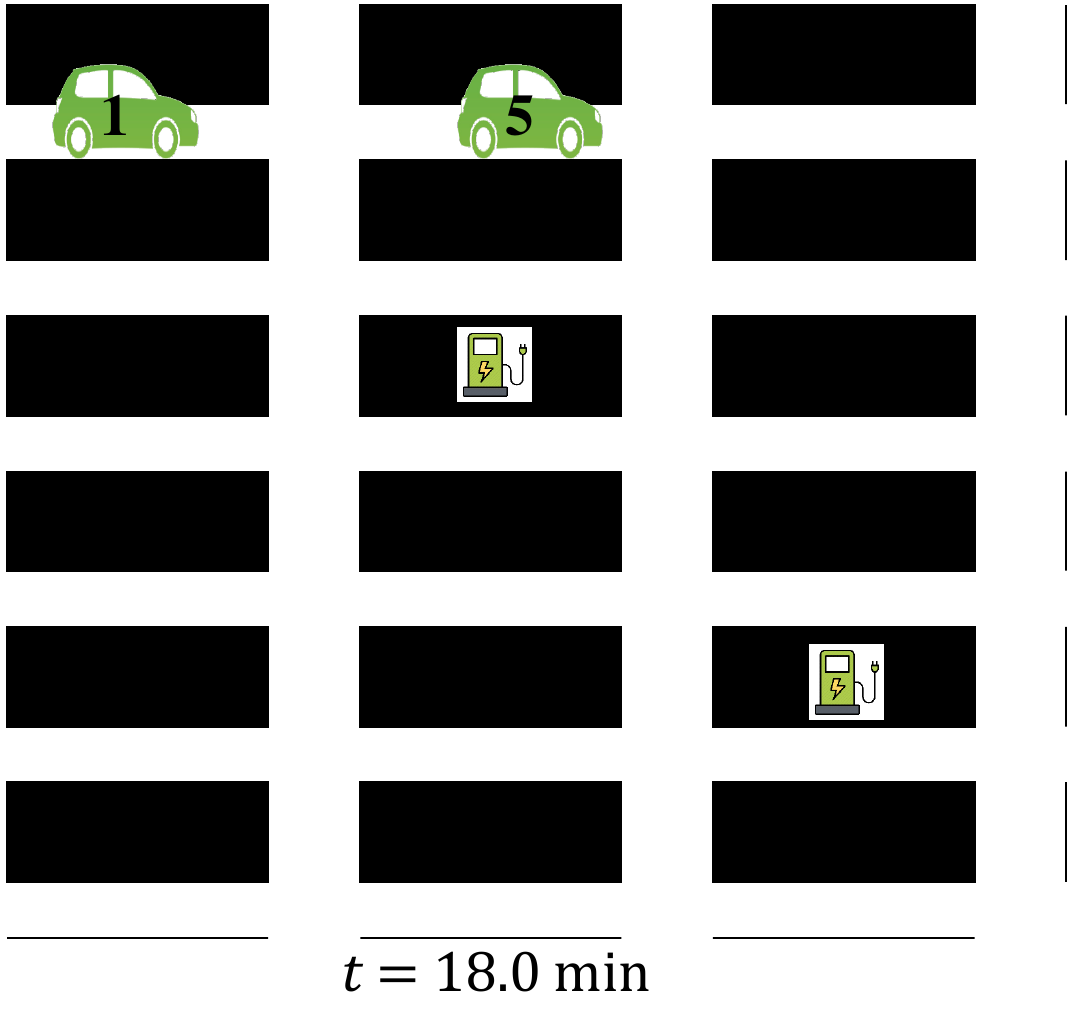}}
	
	\caption{Dynamics with \textsf{MDPP} when $V=0.1$. $t=0$: Customer 1 arrives, $t^{\mathsf{fp}}=1.5$: Vehicle 2 is assigned to Customer 1, $t=5$: Customer 2 arrives, $t=5.6$: Vehicle 3 is returned, $t=t_1^{\mathsf{sp}}=5.6$: first subsequent passage and Vehicle 3 is assigned to Customer 2, $t=13.6$: Customer 3 arrives, $t=15.8$: Vehicle 4 is returned, $t=t_2^{\mathsf{sp}}=15.8$: second subsequent passage and Vehicle 4 is assigned to Customer 3, $t=18$: Vehicle 5 is returned.}
	\label{F:ieV01}
\end{figure}
The passage times in the system are $t^{\mathsf{fp}} = 1.5$ minutes, $t_1^{\mathsf{sp}} = 5.6$ minutes, and $t_2^{\mathsf{sp}} = 15.8$ minutes.  These correspond to times that customers 1, 2, and 3 enter service, respectively.  All times outside excluding these three passage times are times for which $H_c(t) - VC_{vc}(t) < 0$ for all $(v,c)$-pairs.  Note the important role the parameter $V$ has to play in deciding these conditions. 

To further illustrate the importance of the parameter $V$, we present the resulting dynamics in Fig.~\ref{F:ieV1} with $V=1$ under the exact same customer arrivals and vehicle returns.  Clearly, the waiting times of the customers are longer in this case.
\begin{figure}[h!] 
	\centering
	
	\resizebox{1.0\textwidth}{!}{%
		\includegraphics[scale=0.17]{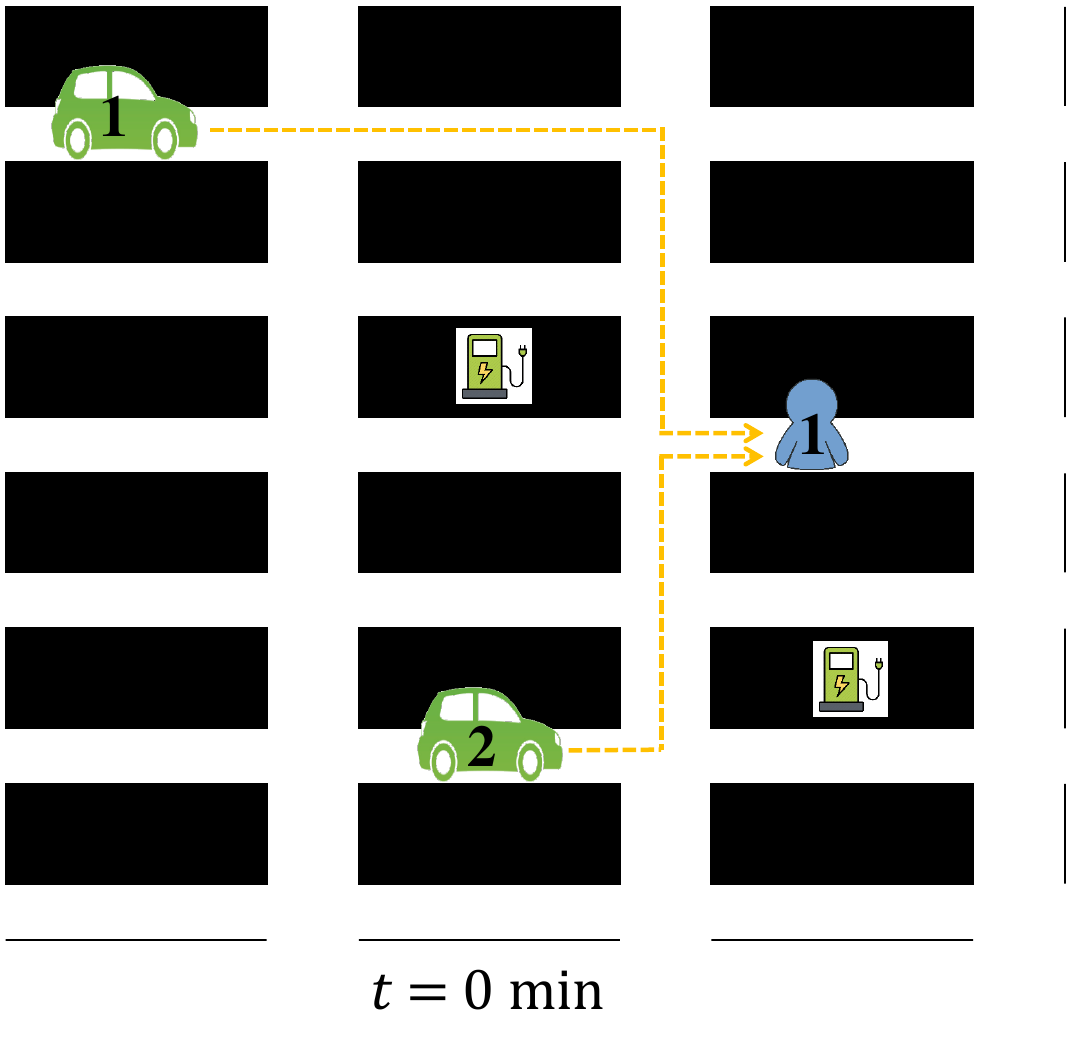}\hspace{0.2in}
		\includegraphics[scale=0.17]{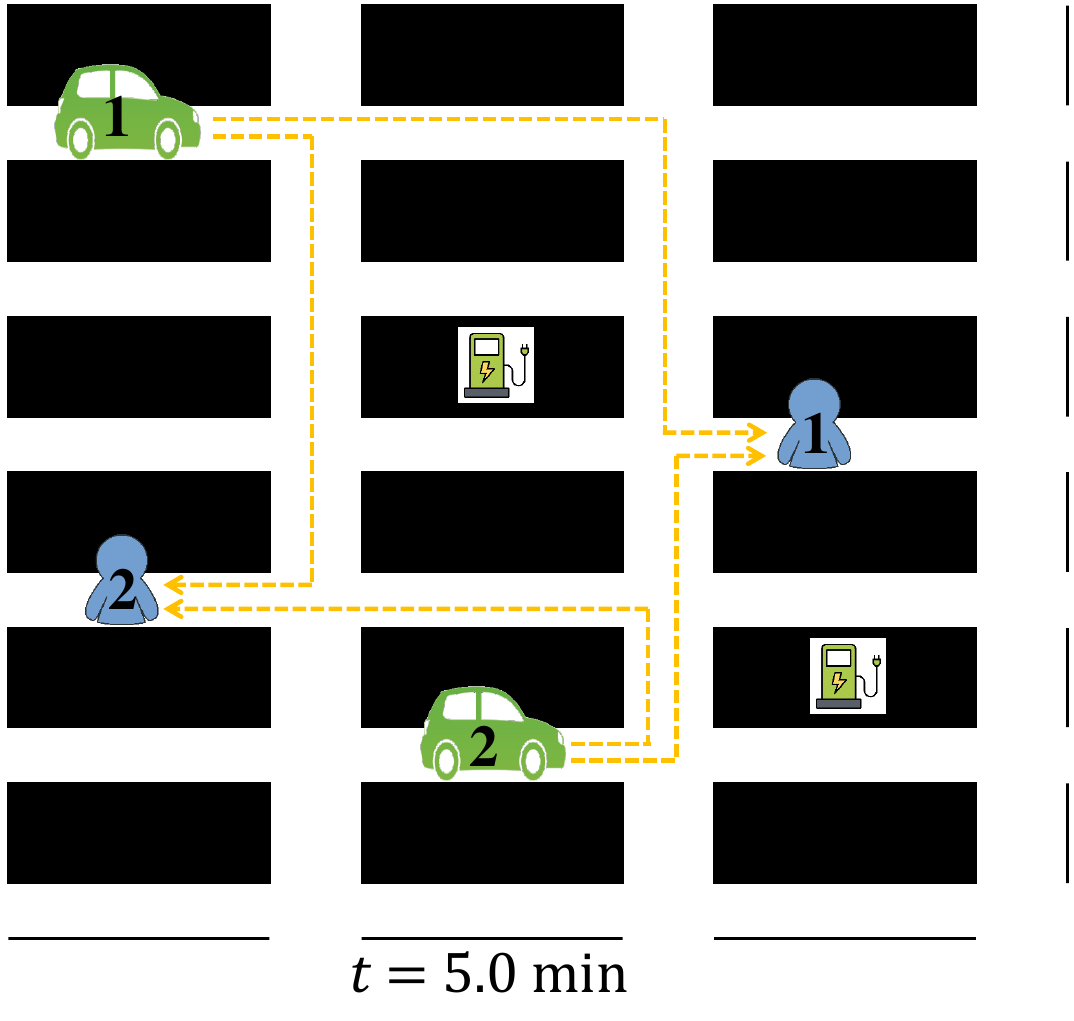}\hspace{0.2in}
		\includegraphics[scale=0.17]{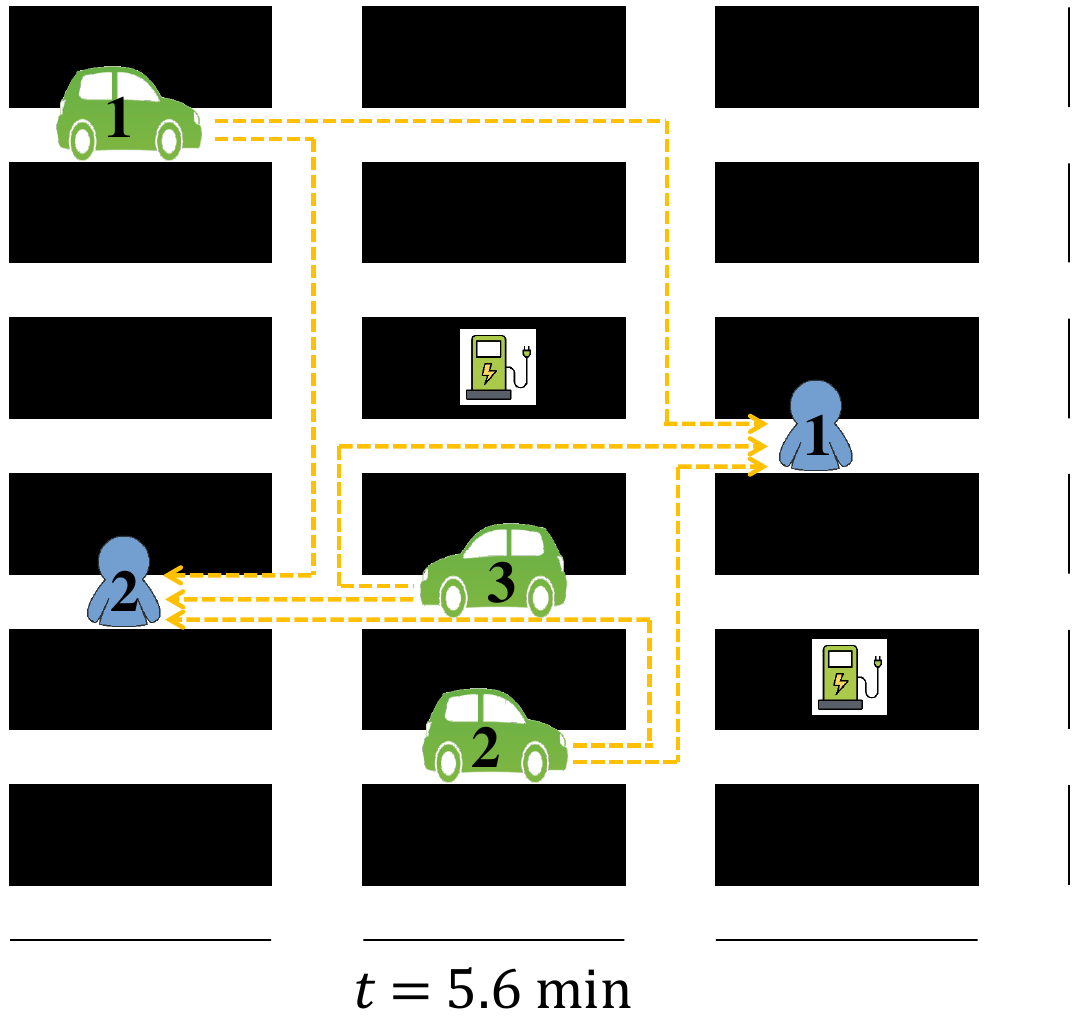}}	
	
	\resizebox{1.0\textwidth}{!}{%
		\includegraphics[scale=0.17]{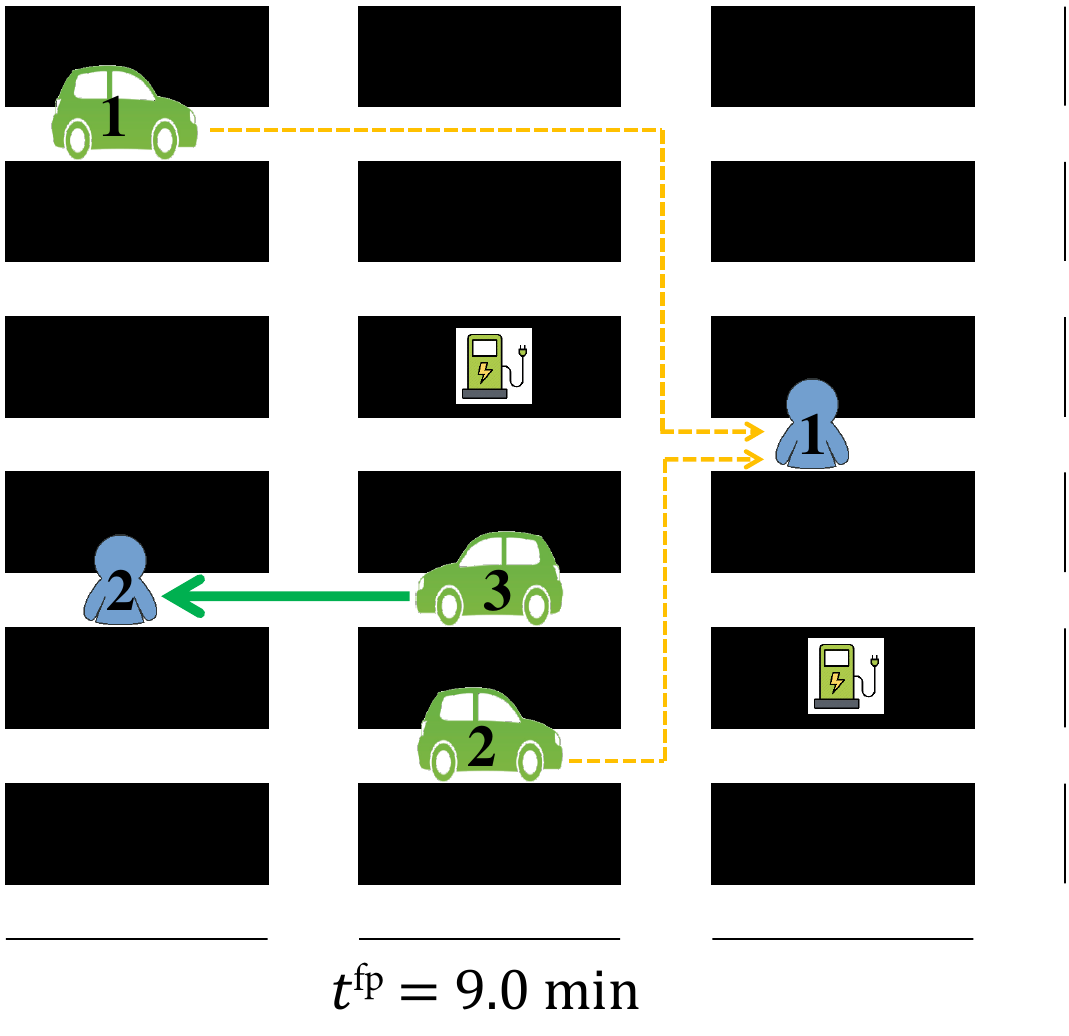}\hspace{0.2in}
		\includegraphics[scale=0.17]{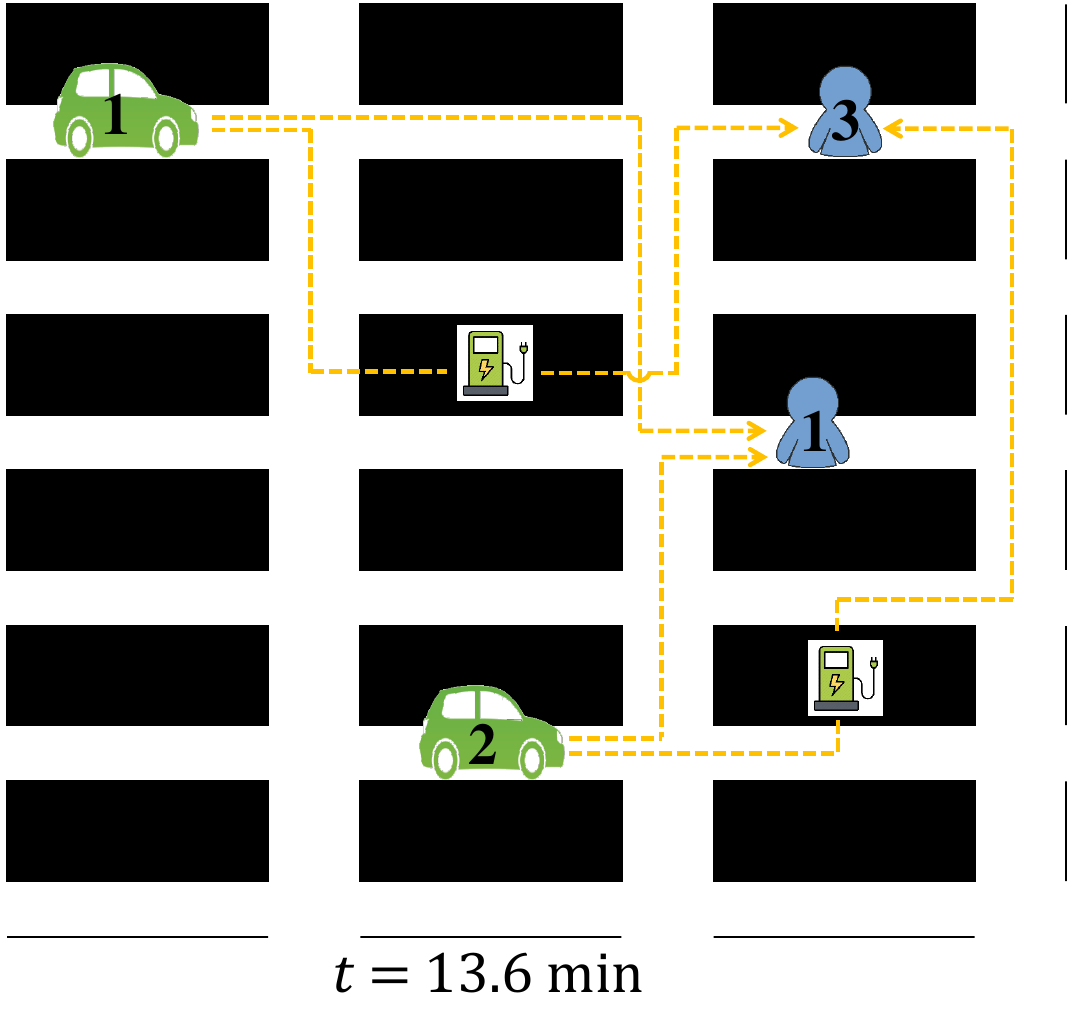}\hspace{0.2in}
		\includegraphics[scale=0.17]{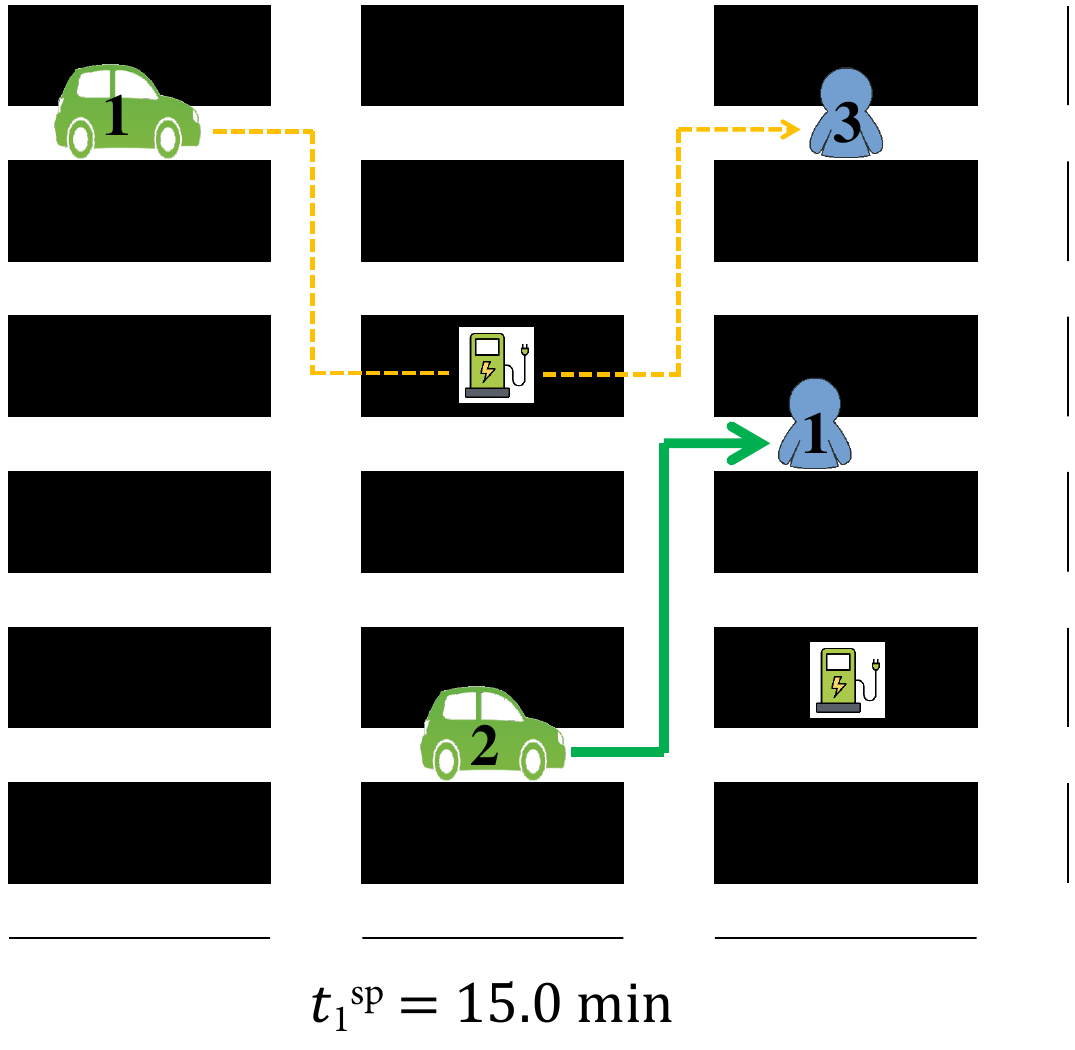}}
	
	\resizebox{1.0\textwidth}{!}{%
		\includegraphics[scale=0.17]{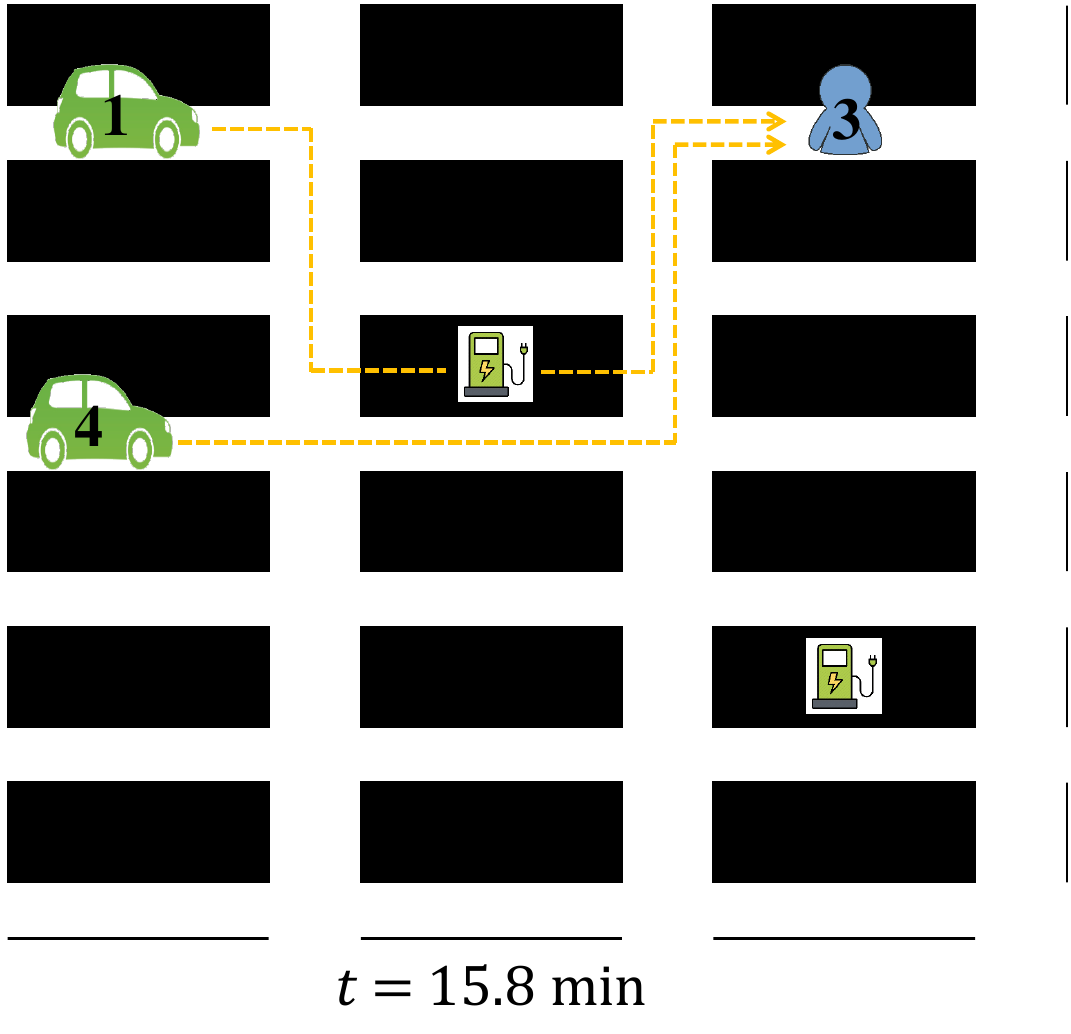}\hspace{0.2in}
		\includegraphics[scale=0.17]{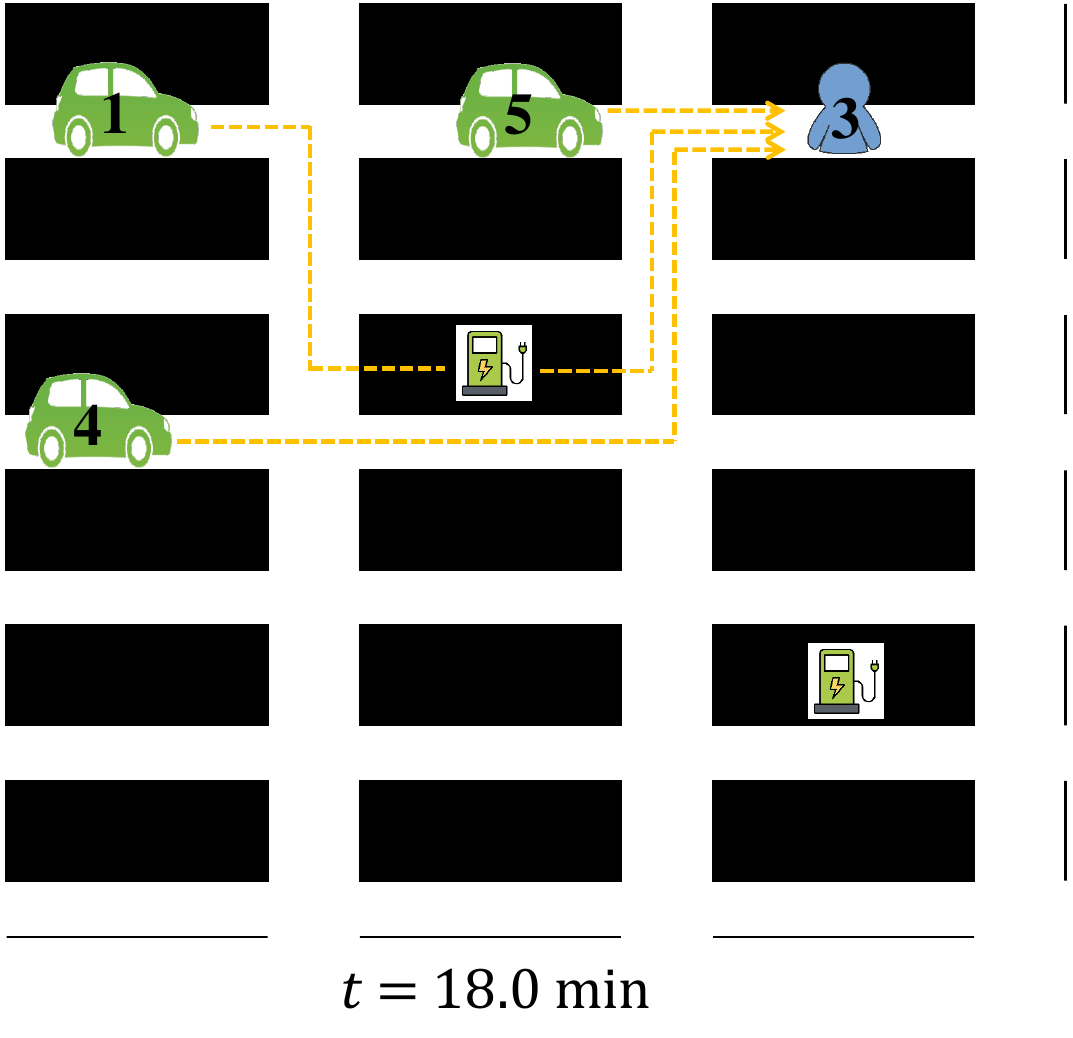}\hspace{0.2in}
		\includegraphics[scale=0.17]{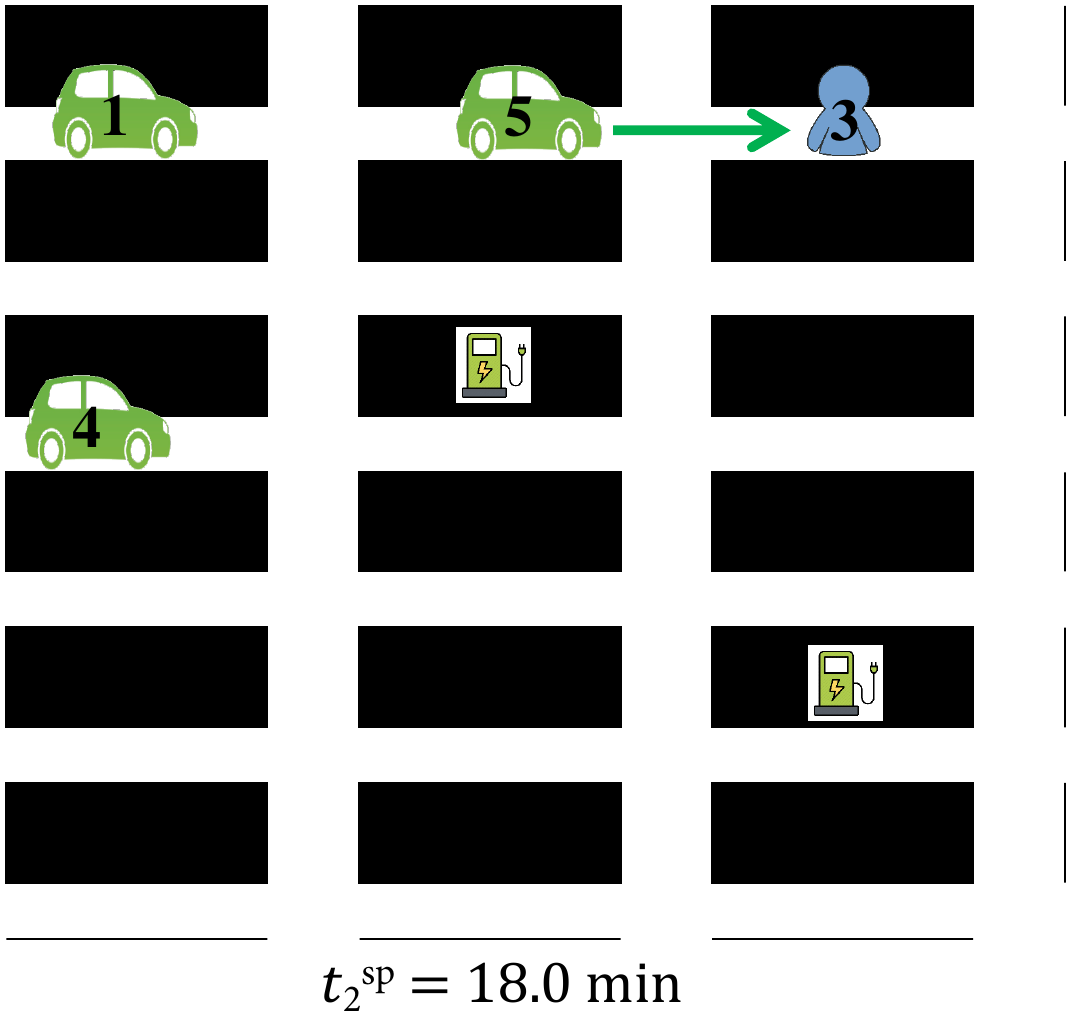}}
	
	\caption{Dynamics with \textsf{MDPP} when $V=1$. $t=0$: Customer 1 arrives, $t=5$: Customer 2 arrives, $t=5.6$: Vehicle 3 is returned, $t=t^{\mathsf{fp}}=9$: first passage and Vehicle 3 is assigned to Customer 2, $t=13.6$: Customer 3 arrives, $t=t_1^{\mathsf{sp}}=15$: first subsequent passage time and Vehicle 2 is assigned to Customer 1, $t=15.8$: Vehicle 4 is returned, $t=18$: Vehicle 5 is returned, $t=t_2^{\mathsf{sp}}=18$: second subsequent passage and Vehicle 5 is assigned to Customer 3.}
	\label{F:ieV1}
\end{figure}

We give another small example here to illustrate en-route charging in our system.  This example specifically highlights the utility of en-route charging as an additional option. Assume a customer arrives at $t = 0$ and requests a vehicle with a charge level of 80\% (or more). The only available vehicle in the system is Vehicle 1 with 60\% charge. Assuming a battery capacity of 60 kwh and a charge power of 120 kw, Vehicle 1 needs to charge for 6 minutes to reach 80\% charge and the en-route charging adds 10 minutes in travel time to reach the customer (due to the detour). At $t = 5$ minutes Vehicle 2 is returned to the system with 80\% charge, and it is 20 minutes away from the customer. Let's set $V = 0.1$. If en-route charging is allowed, then Vehicle 1 will be assigned to the customer at $t^{\mathsf{fp}} = 1.6$ minutes, and will pick up the customer at $t = 17.6$ minutes. Otherwise, if en-route charging is not allowed, then Vehicle 2 will be assigned to the customer at $t^{\mathsf{fp}} = 7$ minutes, and will pick up the customer at $t = 27$ minutes. Hence, allowing en-route charging saves the customer about 10 minutes in waiting time.

\section{Experiments}
\label{S:simulation}

An agent-based simulator was developed to evaluate different vehicle dispatching algorithms for the electric car-sharing systems, as detailed in \citep{li2019agent} and highlighted here. There are two kinds of agents: customers and vehicles. The status of customer agents and vehicle agents are updated every time step, according to the selected vehicle-to-customer assignment algorithm and vehicle recharge/rebalancing algorithm. We compare the MDPP algorithm with other algorithms, including some straightforward heuristics and some algorithms from the literature. We investigate two scenarios: a low-demand scenario with long trips, and a high-demand scenario with short trips. The mean customer waiting times, mean number of waiting customers at every time step, and total dispatch cost (from the vehicles to the customers, the charging stations, or the rebalancing stations) are selected to evaluate the performance of different algorithms.

\subsection{Low-demand scenario with long trips}
\label{ss:low_demand}

We use data obtained from the BMW ReachNow car-sharing operations in Brooklyn, NY in 2017 to simulate the low-demand-long-trips scenario. The project covers 303 Traffic Analysis Zones (TAZs) in Brooklyn, and only 18 of them include charging stations (according to data from \href{https://chargehub.com/en/charging-stations-map.html?lat=40.706109&lon=-74.01014800000002&locId=56029}{ChargeHub.com}), as shown in Fig.~\ref{F:BrooklynMaps}.  The number of chargers in each charging stations are calibrated based on the real data from \href{https://chargehub.com/en/charging-stations-map.html?lat=40.706109&lon=-74.01014800000002&locId=56029}{ChargeHub.com}. We have a total of 82 chargers in the network, with 2 to 11 Level 2 chargers in the charging stations. Each TAZ is treated as a zone in the car-sharing system. Each zone is then divided into 5 customer nodes with different charge levels, as shown in Fig.~\ref{F:network_b}.
\begin{figure}[h!]
	\centering
	\footnotesize
	\resizebox{0.85\textwidth}{!}{%
		\includegraphics[scale=0.045]{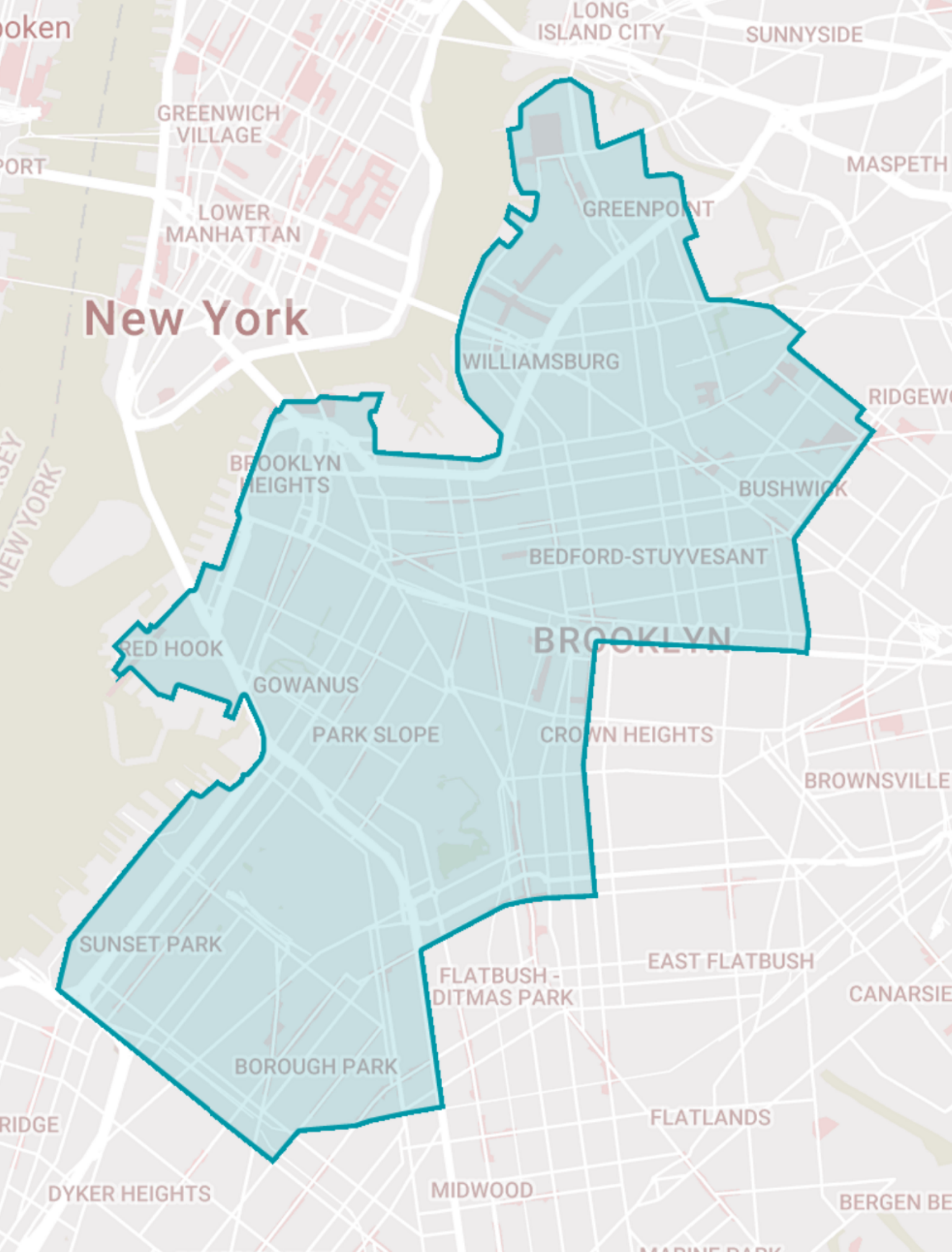} 
		\includegraphics[scale=0.112]{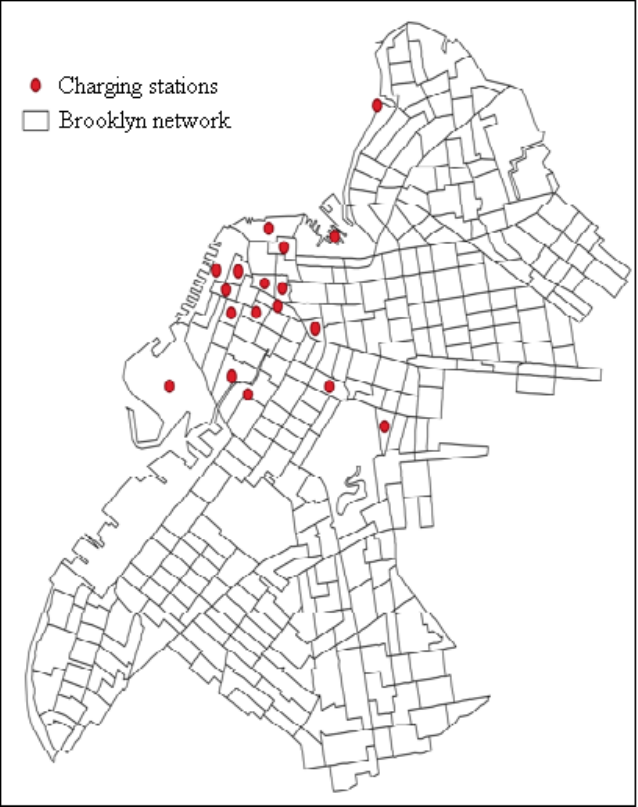}}
	
	(a) \hspace{2.5in} (b)
	\caption{(a) Map of project area, (b) charger distribution from \href{https://chargehub.com/en/charging-stations-map.html?lat=40.706109&lon=-74.01014800000002&locId=56029}{ChargeHub.com}.}
	\label{F:BrooklynMaps}
\end{figure}
The ReachNow project uses gasoline-powered vehicles, where a small number of the trips have very long distances. These trips were ignored in our simulation since they are beyond the battery ranges of the electric vehicles.

Our dataset includes all the trips in September, 2017. The average arrival rate is approximately 230 customers per day. The trip distance and trip duration distributions are shown in Fig.~\ref{F:brooklyn_demand}. 
\begin{figure}[h!]
	\centering
	\subfloat[]{\includegraphics[scale=0.38]{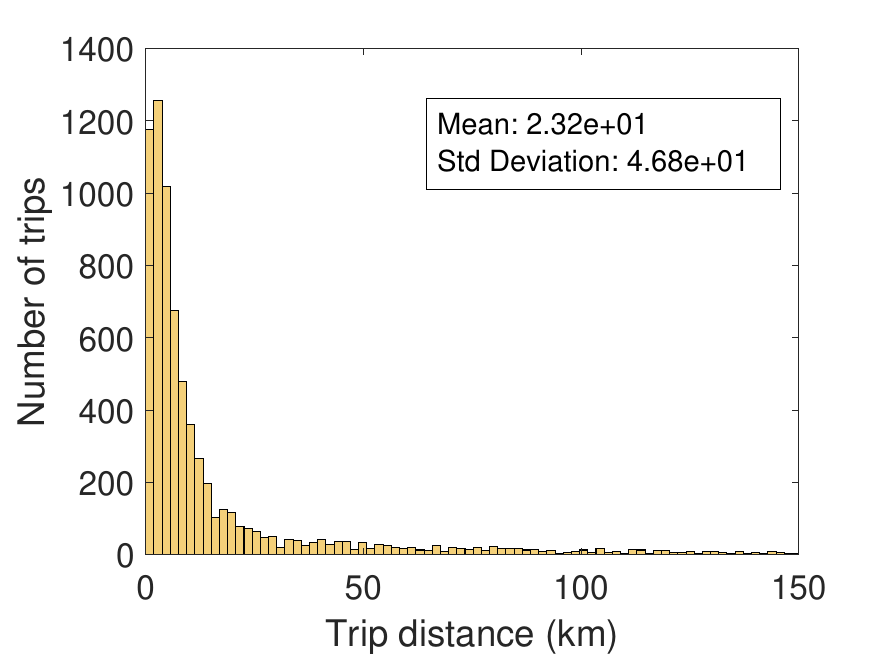}\label{F:brooklyn_demand_a}} 
	\subfloat[]{\includegraphics[scale=0.38]{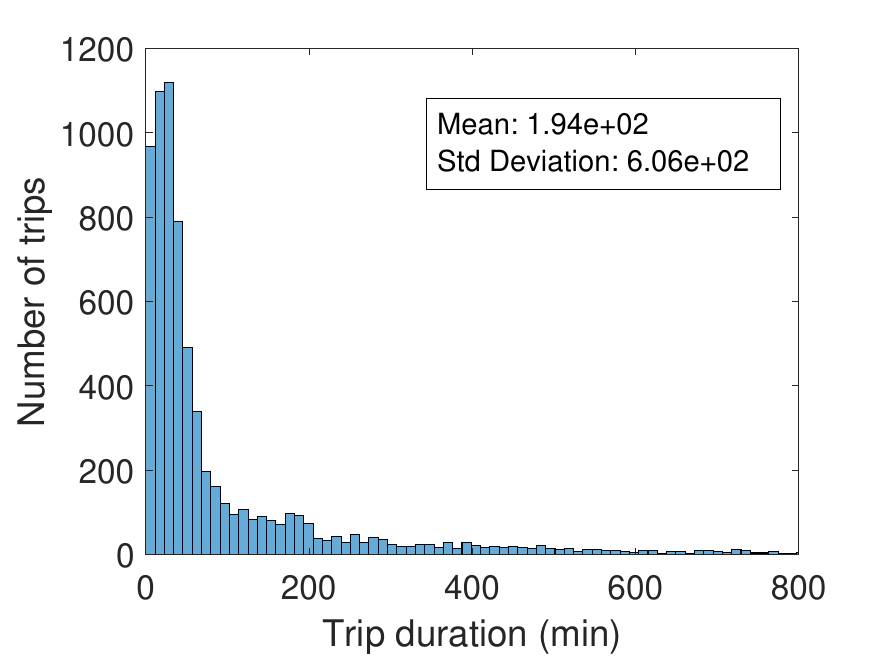}\label{F:brooklyn_demand_b}}
	\caption{BMW ReachNow project's (a) trip distance distribution, (b) trip duration distribution.\label{F:brooklyn_demand}}		
\end{figure}
The mean trip distance is 23.2 km, and the mean trip duration is 194 min. The  fleet size is 262 vehicles. We tested the following algorithms in this scenario:
\begin{enumerate}[label=(\roman*)]
	\item \textsf{MDPP} with different values of $V$, namely, $V=0$, 0.1, and 1.
	\item A na\"{i}ve ``charger chasing'' algorithm that assigns vehicles to the closest charging facility right after they drop customers off, and gives priority to vehicles with lower charge levels (to use chargers). 
	\item The vehicle recharging \& rebalancing heuristic \citep{li2019agent} developed for a SAEV system.
	\item The vehicle recharging rules \citep{loeb2018shared} designed for a SAEV system.
\end{enumerate} 
Note that the vehicle-to-customer assignment rule in \citep{marczuk2015autonomous}, which assigns the nearest available vehicle to the customer on a first-come-first-serve basis, is used when testing algorithms (ii)-(iv). \emph{For all algorithms, vehicles returned to zones with charging stations are automatically connected to chargers, if any is available.} The simulation time step length is 1 min, and the horizon is 30 days. There is no maximum waiting time, hence no customer will abandon the system. We compare the sensitivity of all the algorithms to battery capacity, charge power of the chargers, and fleet size.

Fig.~\ref{F:battery} compares the mean waiting time, mean number of waiting customers, and total dispatch cost under different battery capacities. 
\begin{figure}[h!]
	\centering
	\subfloat[]{\includegraphics[scale=0.35]{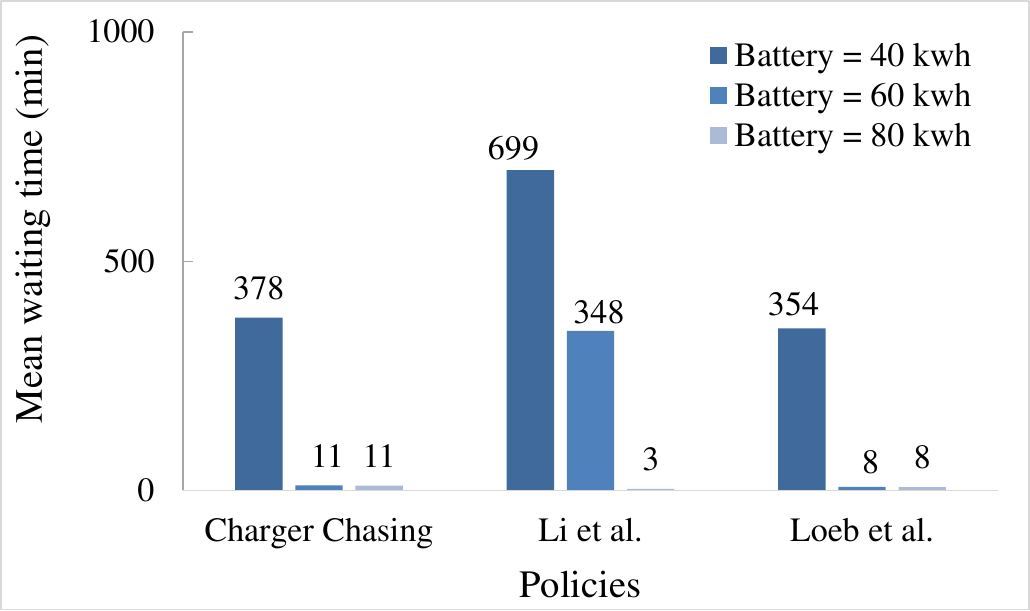}\label{F:battery_a}} 
	\subfloat[]{\includegraphics[scale=0.35]{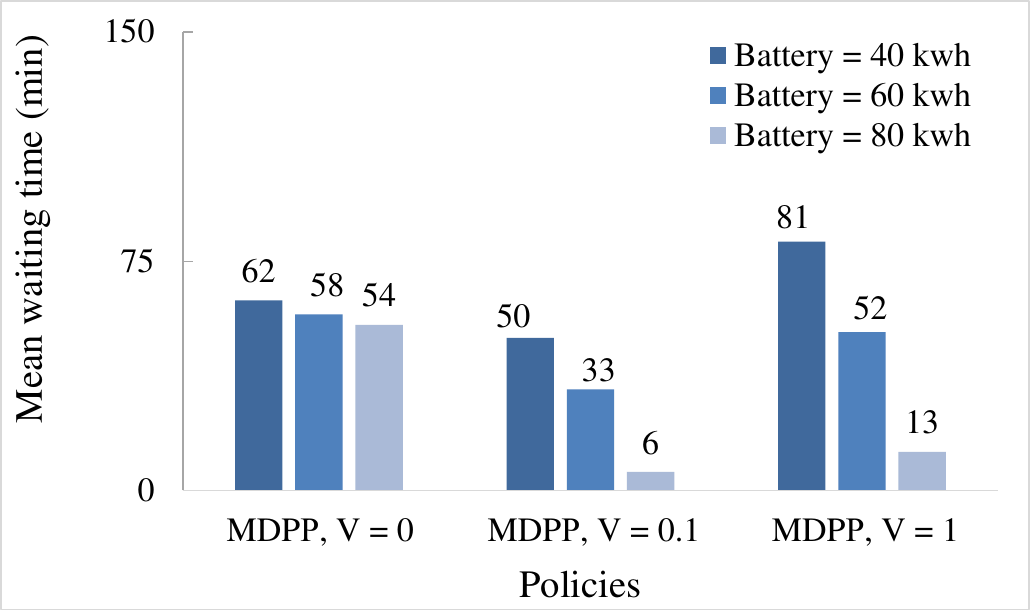}\label{F:battery_b}} \\
	
	\subfloat[]{\includegraphics[scale=0.35]{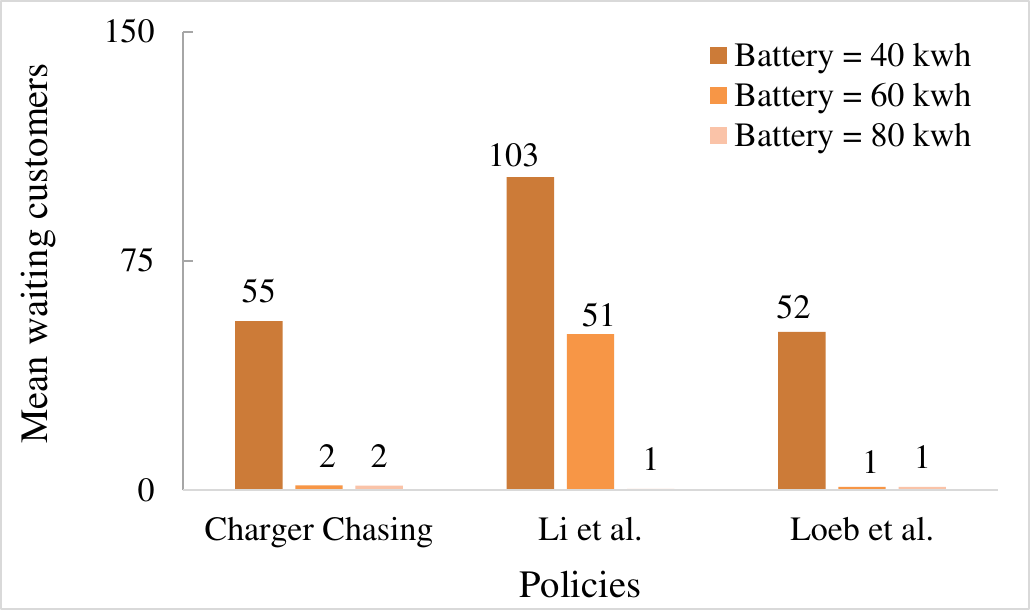}\label{F:battery_c}} 
	\subfloat[]{\includegraphics[scale=0.35]{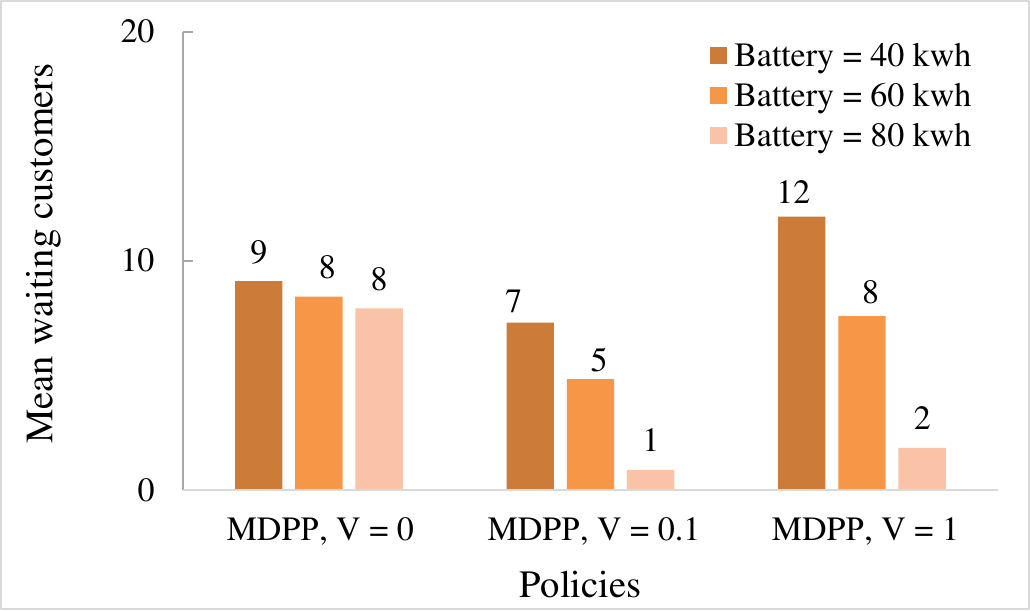}\label{F:battery_d}} \\
	
	\subfloat[]{\includegraphics[scale=0.35]{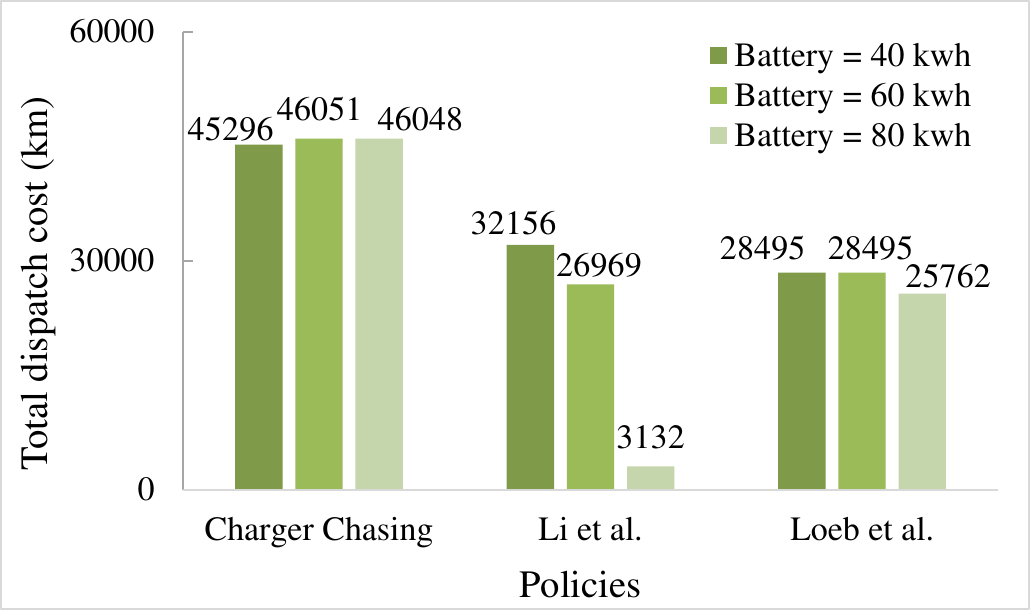}\label{F:battery_e}} 
	\subfloat[]{\includegraphics[scale=0.35]{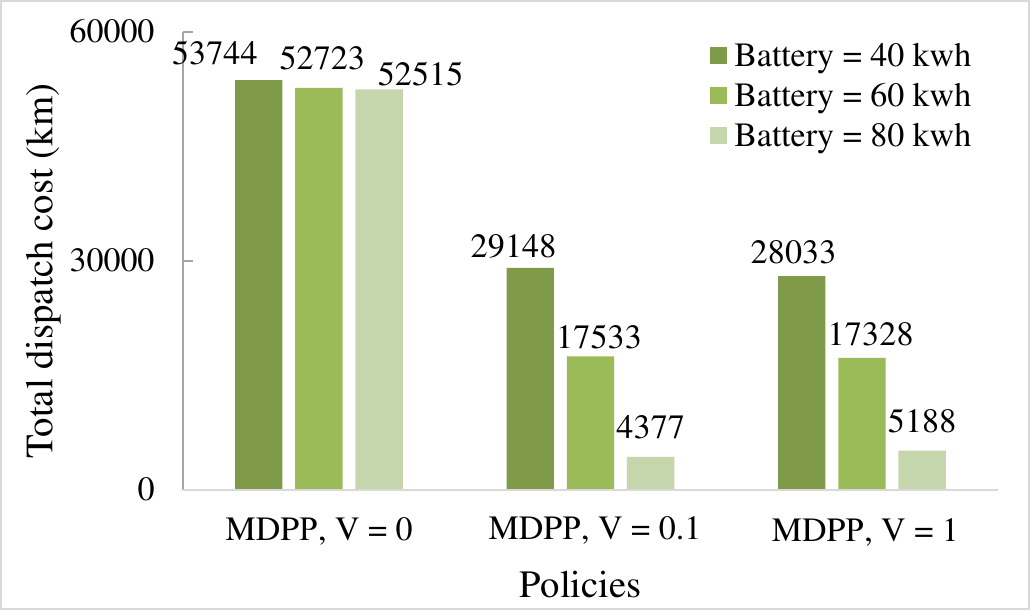}\label{F:battery_f}}
	\caption{Comparison under different battery capacities.\label{F:battery}}		
\end{figure}
The fleet size is set to 262, which is the same as the real vehicle fleet size; the chargers are all Level 2 chargers with a charge power of 7 kw (consistent with what is seen in practice). We tested three battery capacities: 40 kwh, 60 kwh, and 80 kwh. The range per kwh is assumed to be 7 km, hence, the corresponding full ranges are 280 km, 420 km and 560 km, respectively. As we can see, when the battery capacity is 40 kwh, \textsf{MDPP} performs much better than all the other approaches. The value of $V$ significantly influences the performance of \textsf{MDPP}, $V=0.1$ performs best: it achieves the shortest mean waiting time (50 min) and the smallest number of mean waiting customers (7) with a low total dispatch cost (29,128 km). When the battery capacity is increased to 60 kwh, the charger chasing policy and \citep{loeb2018shared} perform better than our \textsf{MDPP} in terms of the mean waiting time and the mean number of waiting customers, but these advantages come with much higher dispatch costs. 
When we increase battery capacity to 80 kwh, performance of the both the charger chasing algorithm and the approach in \citep{loeb2018shared} do not further improve, while the performance of the heuristic in \citep{li2019agent} improves significantly and becomes the best policy. The \textsf{MDPP} policy with $V=0.1$ has comparable performance to the heuristic in \citep{li2019agent} but with a slightly higher dispatch cost. In general, \textsf{MDPP} with $V=0.1$ is a satisfactory policy across all battery capacities.

Fig.~\ref{F:spc} shows the results of the comparisons under different charge powers. The fleet size is 262 vehicles, and the battery capacity is 40 kwh. 
\begin{figure}[h!]
	\centering
	\subfloat[]{\includegraphics[scale=0.35]{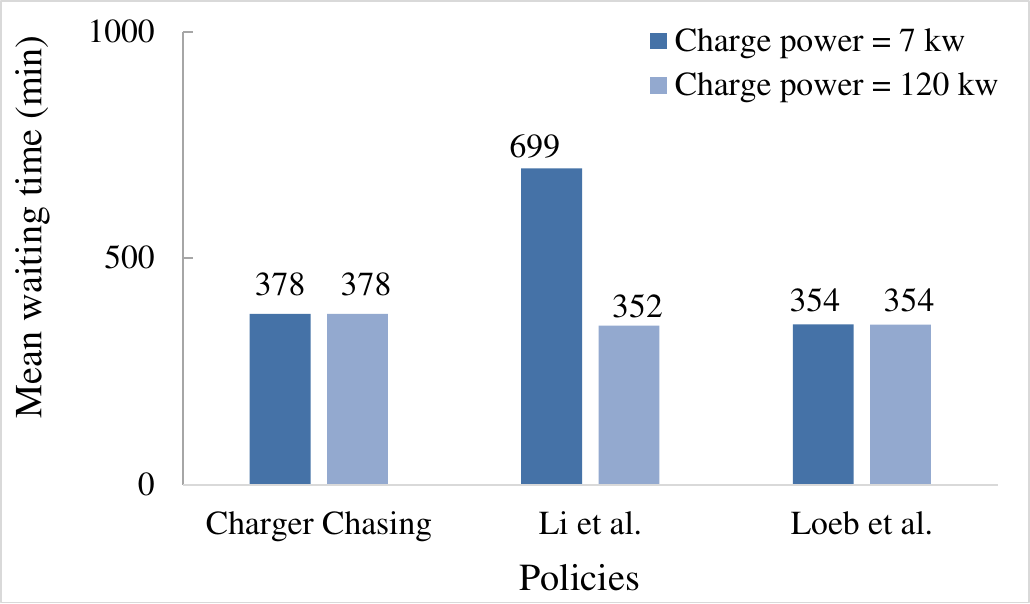}\label{F:spc_a}} 
	\subfloat[]{\includegraphics[scale=0.35]{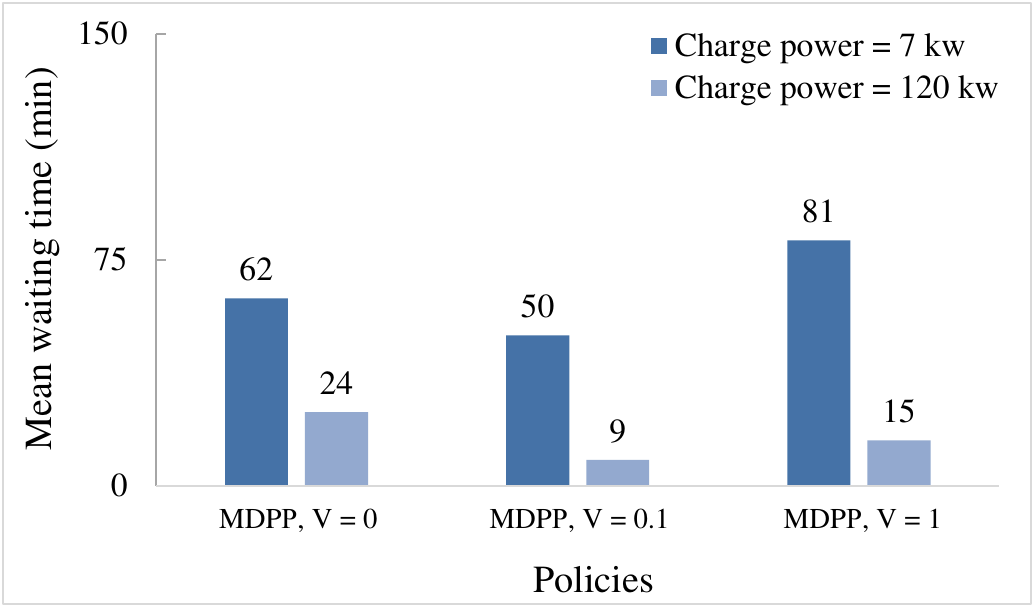}\label{F:spc_b}} \\
	
	\subfloat[]{\includegraphics[scale=0.35]{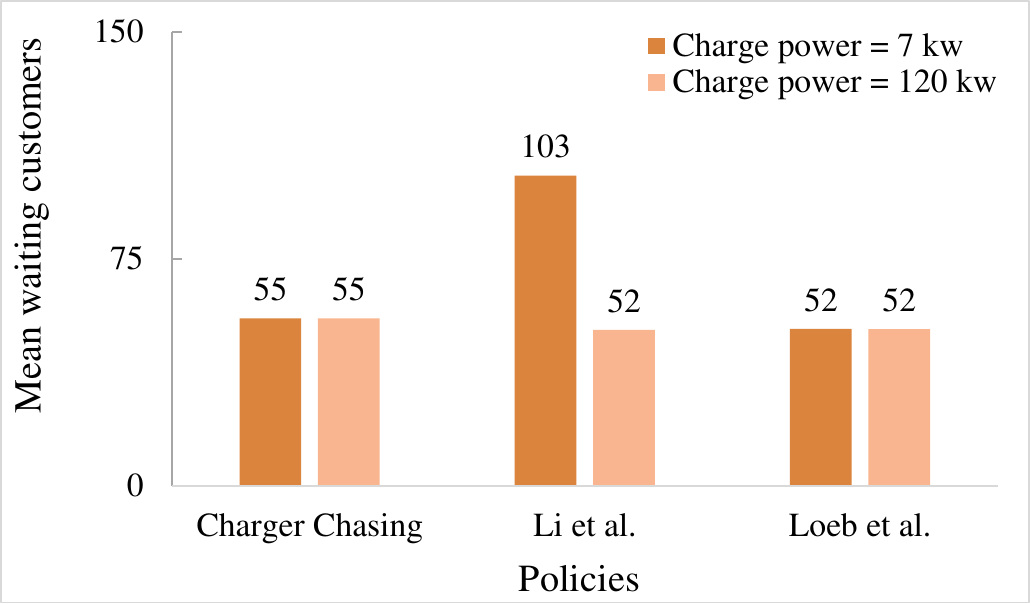}\label{F:spc_c}} 
	\subfloat[]{\includegraphics[scale=0.35]{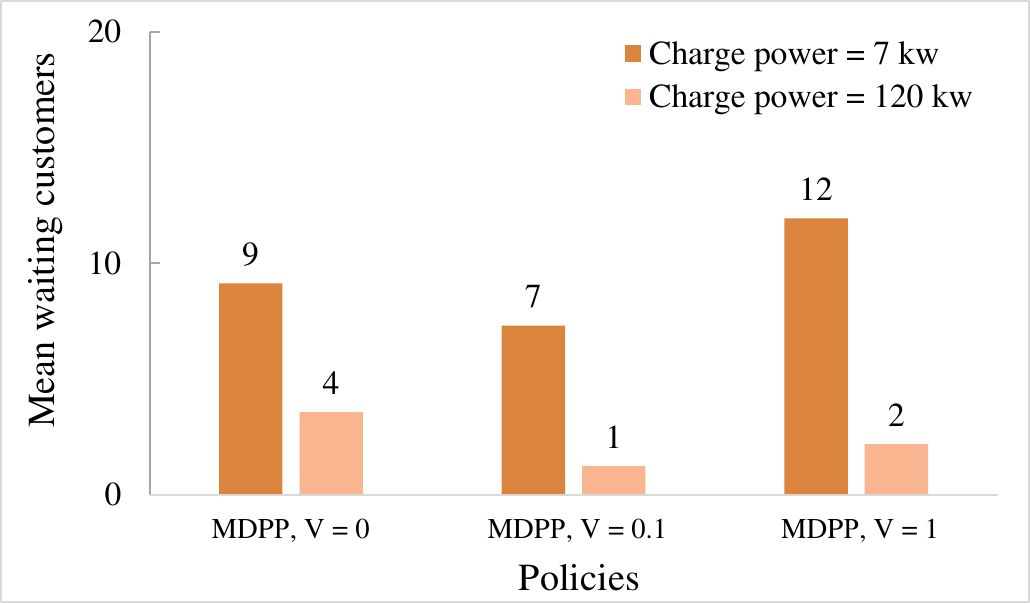}\label{F:spc_d}} \\
	
	\subfloat[]{\includegraphics[scale=0.35]{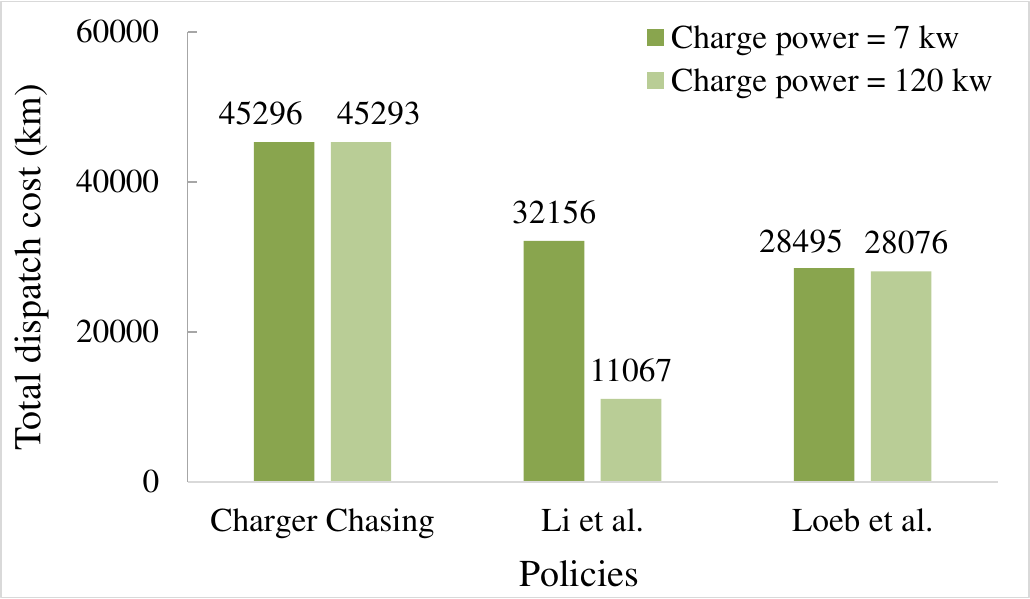}\label{F:spc_e}} 
	\subfloat[]{\includegraphics[scale=0.35]{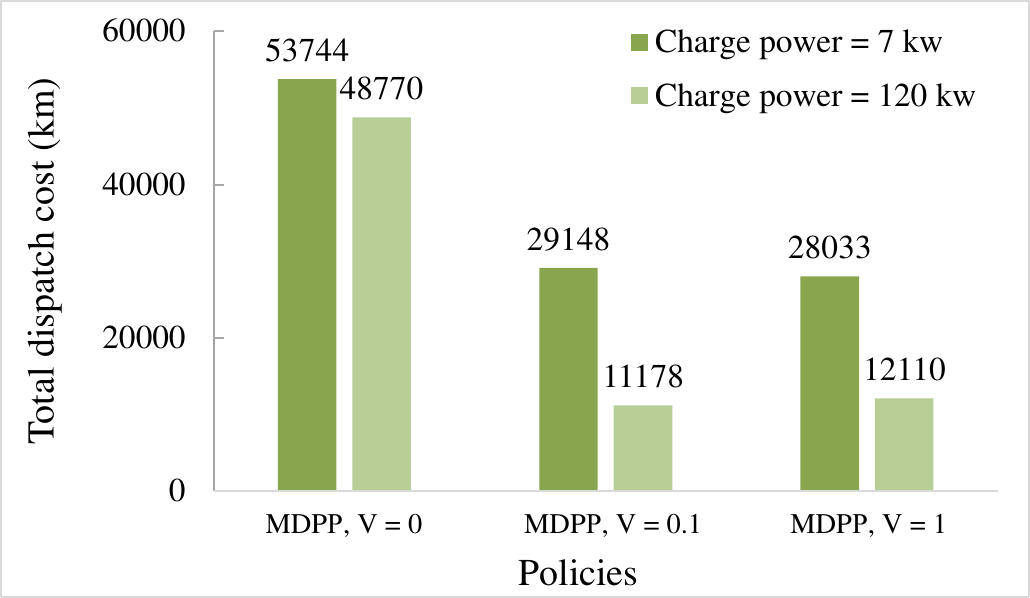}\label{F:spc_f}}
	\caption{Comparison under different charge powers.\label{F:spc}}		
\end{figure}
We have a basic setting with Level 2 chargers with 7 kw power, and a comparative setting with superchargers with 120 kw power. It turns out that charge power has little impact on the charger chasing policy and the approach proposed in \citep{loeb2018shared}, while it improves the performance of the heuristic in \citep{li2019agent} and \textsf{MDPP} substantially. Again, \textsf{MDPP} with $V=0.1$ is the best policy: it achieves a minimum mean waiting time of 9 min, and only has 1 waiting customer on average, with almost the lowest dispatch cost of 11,178 km, slightly higher than that of \citep{li2019agent}.

Fig.~\ref{F:lf} shows the results under different fleet sizes. 
\begin{figure}[h!]
	\centering
	\subfloat[]{\includegraphics[scale=0.35]{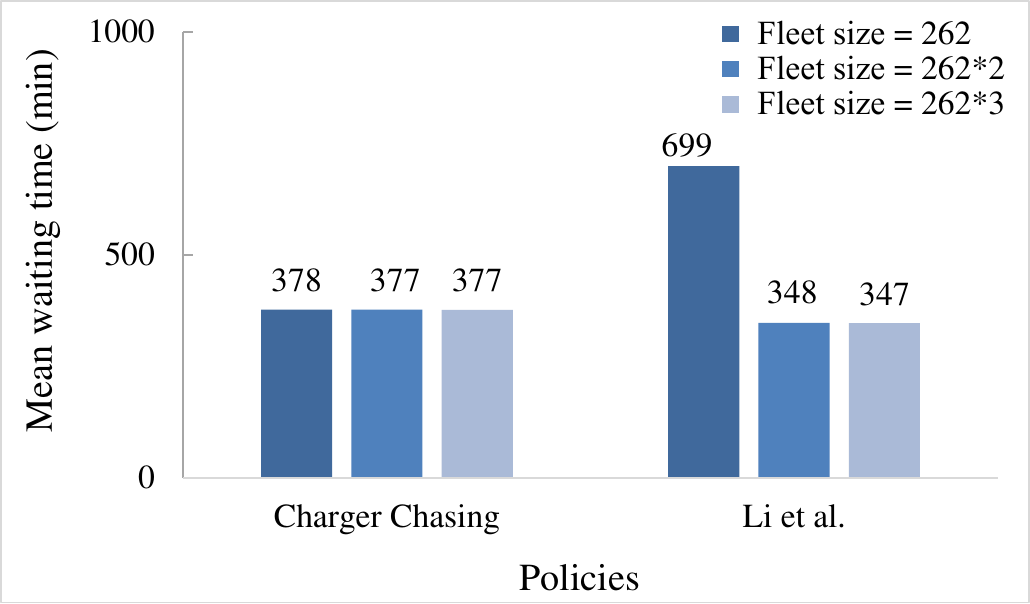}\label{F:lf_a}} 
	\subfloat[]{\includegraphics[scale=0.35]{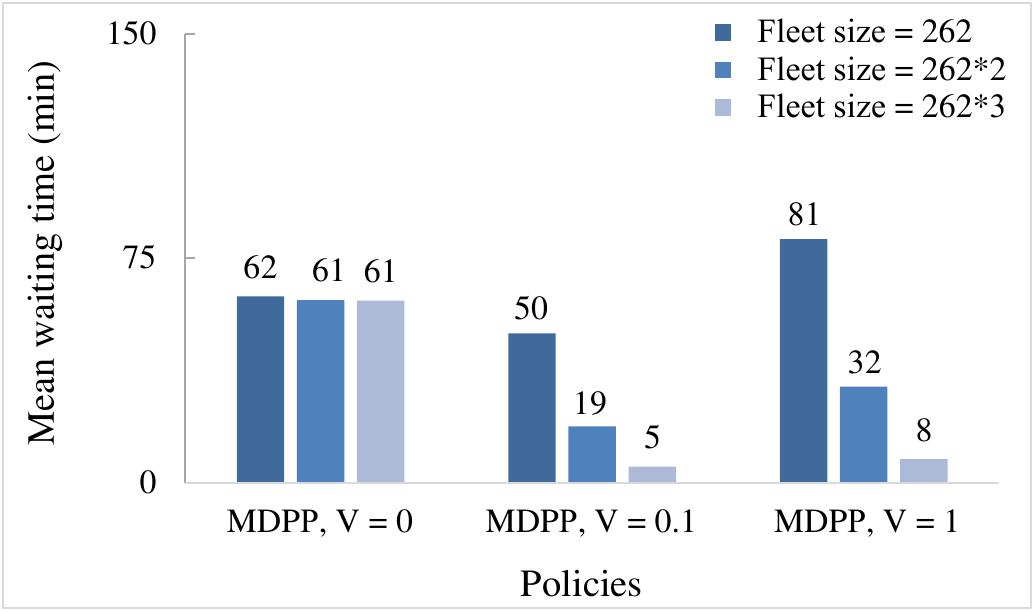}\label{F:lf_b}} \\
	
	\subfloat[]{\includegraphics[scale=0.35]{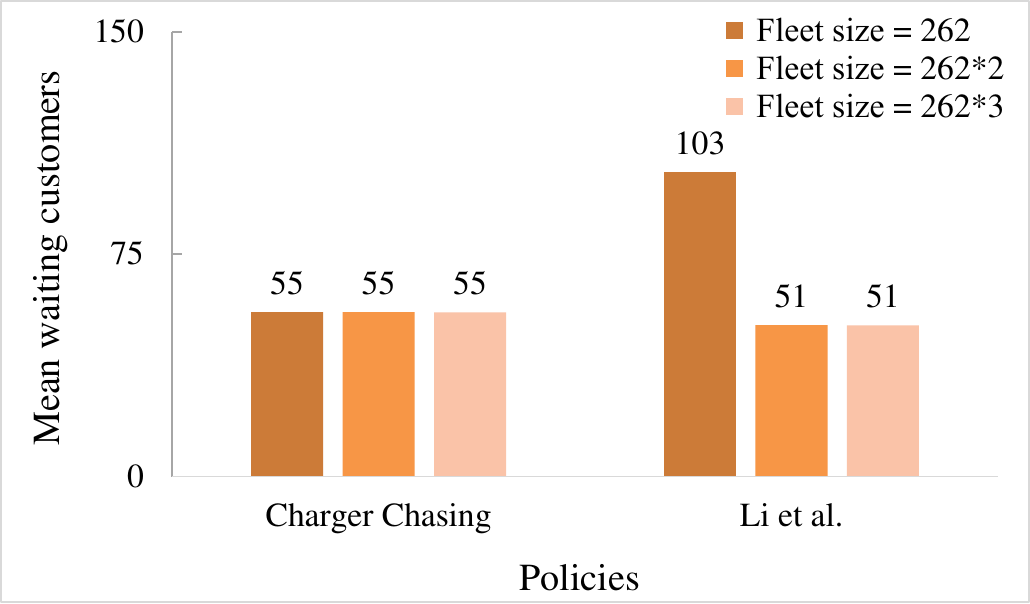}\label{F:lf_c}} 
	\subfloat[]{\includegraphics[scale=0.35]{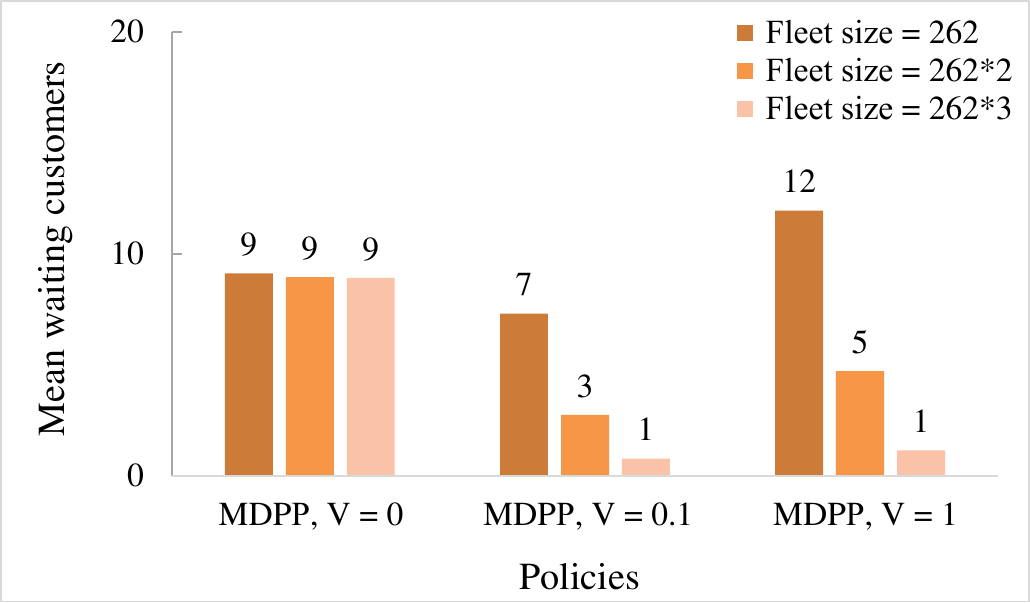}\label{F:lf_d}} \\
	
	\subfloat[]{\includegraphics[scale=0.35]{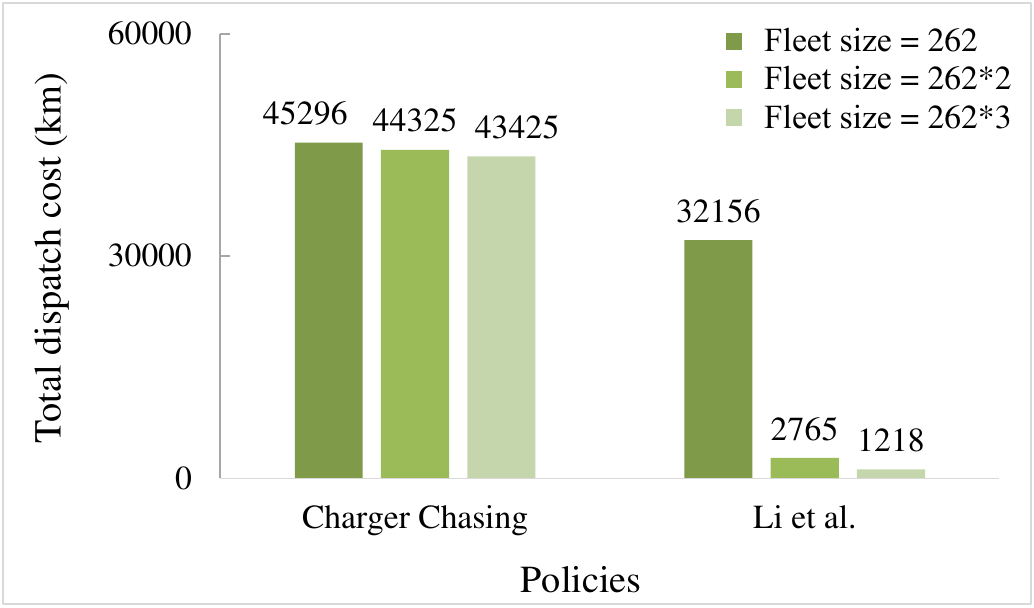}\label{F:lf_e}} 
	\subfloat[]{\includegraphics[scale=0.35]{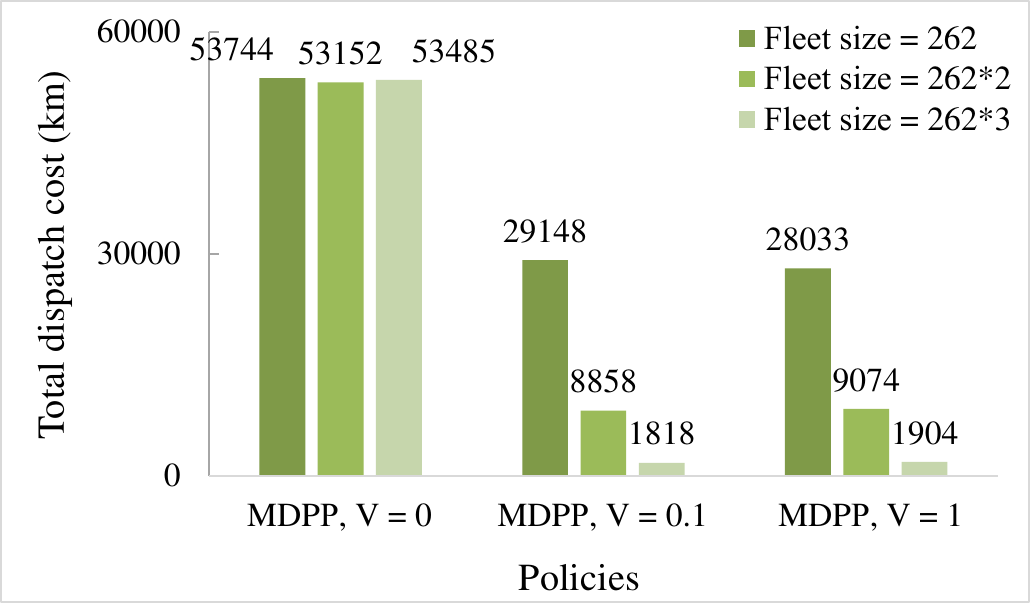}\label{F:lf_f}}
	\caption{Comparison under different fleet sizes.\label{F:lf}}		
\end{figure}
The chargers are assumed to be Level 2 chargers with charge power of 7 kw, the battery capacity is 40 kwh. We tested three fleet sizes: 262, 262$\times$2, and 262$\times$3. Note that the approach in \citep{loeb2018shared} is not included in Fig.~\ref{F:lf} because the simulation time for their approach is prohibitively long in these scenarios with more vehicles. As we can see, \textsf{MDPP} with $V=0.1$ is still the best policy as the fleet size increases. When the fleet size doubles, both the mean waiting time and the mean number of waiting customers under \textsf{MDPP} with $V=0.1$ are 60\% lower, and the total dispatch cost is reduced by 70\%. When we triple the fleet size, the mean waiting time decreases by 90\% (from 50 min to 5 min), the mean number of waiting customers decreases by 85\% (from 7 to 1), and the total dispatch cost decreases by approximately 94\% (from 29,148 km to 1,818 km).

To summarize, the value of $V$ influences the performance of \textsf{MDPP}, and among all the three tested $V$ values, $V=0.1$ delivers the best performance. \textsf{MDPP} with $V=0.1$ outperforms all other algorithms under almost all tested cases, and it is sensitive to improvements in battery capacity, charge power and fleet size. This implies that by employing the proposed \textsf{MDPP} approach with an appropriately chosen value for $V$, the operators can significantly improve system performance (including customer level of service and dispatch costs) by introducing vehicles with larger battery capacity, chargers with greater charge power, or more vehicles.

\subsection{High-demand scenario with short trips}

We use the network in Midtown Manhattan shown in Fig.~\ref{F:Manhattan_map_a} to test the high demand scenario with short trips. Yellow Cab demand data, which can be downloaded from the New
York City Taxi \& Limousine Commission website are used in the simulation. 
\begin{figure}[h!]
	\centering
	\subfloat[]{\includegraphics[scale=0.54]{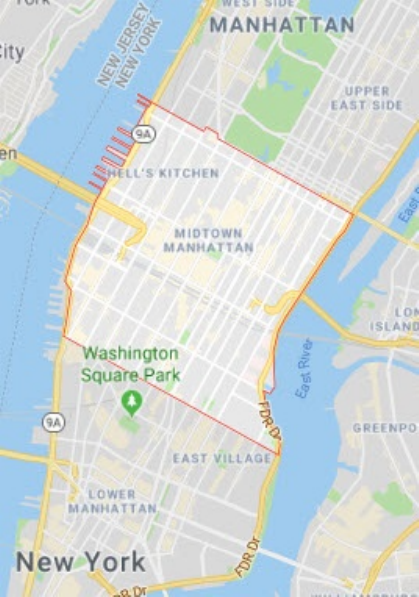}\label{F:Manhattan_map_a}} \qquad
	\subfloat[]{\includegraphics[scale=0.37]{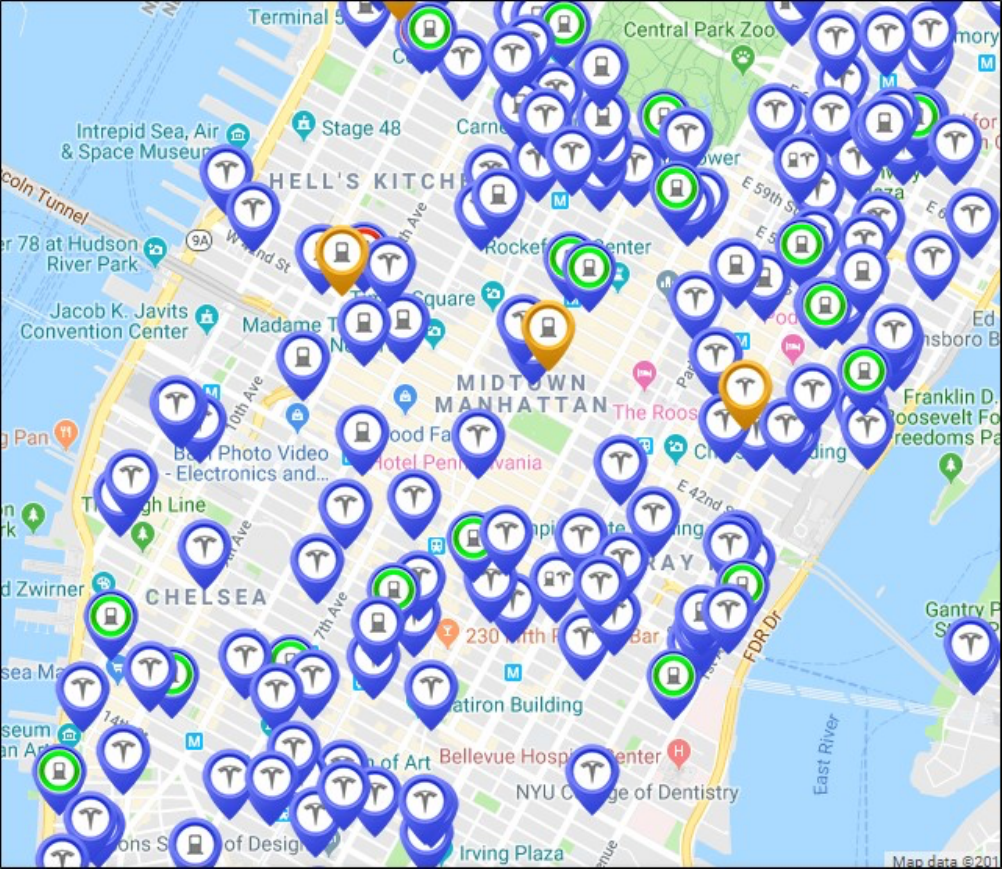}\label{F:Manhattan_map_b}}
	\caption{Midtown Manhattan: (a) map of network from \href{https://www.google.com/maps/place/Midtown+Manhattan,+New+York,+NY,+USA/@40.7575603,-73.9925782,13z/data=!4m5!3m4!1s0x89c25901a4127ca9:0xbecdcc9081d6cfdb!8m2!3d40.7549309!4d-73.9840195}{Google Maps}, and (b) charging station locations from \href{https://chargehub.com/en/charging-stations-map.html?lat=40.706109&lon=-74.01014800000002&locId=56029}{ChargeHub.com}.}
	\label{F:Manhattan_map}
\end{figure}
The charging station locations and the corresponding number of chargers are calibrated using data from \href{https://chargehub.com/en/charging-stations-map.html?lat=40.706109&lon=-74.01014800000002&locId=56029}{ChargeHub.com}. Unlike the Brooklyn network, Midtown Manhattan has a dense distribution of charging stations, with both Level 2 chargers and superchargers, as shown in Fig.~\ref{F:Manhattan_map_b}.  There are 19 TAZs in Midtown Manhattan, corresponding to 19 zones in the SAEV system. 17 of them have charging stations, and the number of chargers varies from 9 to 31, with a total number of 288 Level 2 chargers and 8 superchargers.

Data from June 1-7, 2018 are used to run a 1-week simulation. The average arrival rate is approximately 68,500 customers per day. Fig.~\ref{F:manhattan_demand} shows the distributions of trip distances and trip durations. The mean trip distance is 1.89 km, and the mean trip duration is 10.2 min. New York City had roughly 13,500 yellow cabs in 2018, and the trips in Midtown Manhattan account for approximately 10\% of all the yellow cab trips. We set 1,200 as the default fleet size in this scenario. The default battery capacity is set to be 20 kwh, since we have very short trips, the range per kwh is assumed to be 7 km.
\begin{figure}[h!]
	\centering
	\subfloat[]{\includegraphics[scale=0.38]{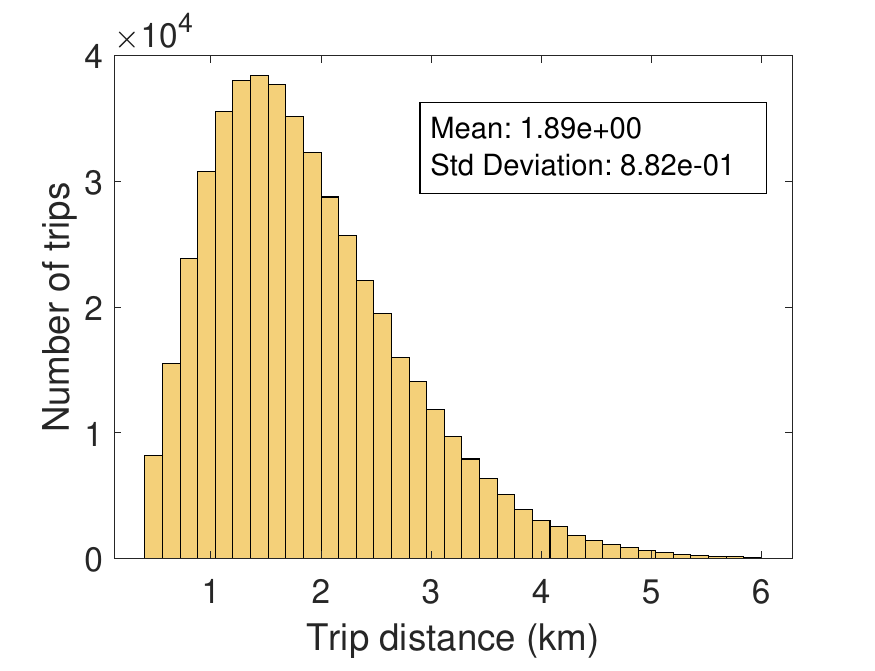}\label{F:manhattan_demand_a}} 
	\subfloat[]{\includegraphics[scale=0.38]{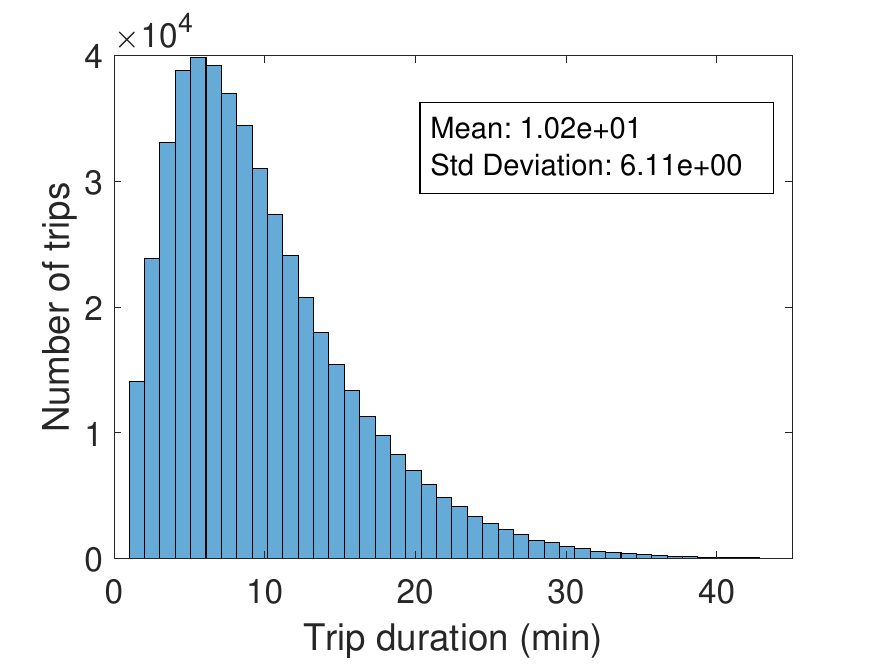}\label{F:manhattan_demand_b}}
	\caption{Midtown Manhattan taxis: (a) distribution of trip distances, (b) distribution of trip durations.\label{F:manhattan_demand}}		
\end{figure}
In addition to the algorithms that were tested in Sec.~\ref{ss:low_demand}, we also test the ``NonEV NoReb'' policy which is representative of gasoline-powered vehicles without rebalancing, since we have high arrival rates and dense charging stations in this scenario. We also test with more values for $V$ in the \textsf{MDPP} policy, namely, 0, 0.001, 0.01, 0.1, and 1. Again, the vehicle-to-customer assignment rule in \citep{marczuk2015autonomous} is used for all policies excluding the \textsf{MDPP} policy. The simulation step is 1 min, and the horizon is 7 days. A maximum waiting time of 30 min and no maximum waiting times are both tested. In the case of a maximum waiting time, customers will abandon the system when their waiting times reach 30 minutes; hence, we record the number of ``lost customers'' as well in this case.

In this scenario, we compare the performance of the scheduling policies under different fleet sizes and battery capacities, while keeping charge power constant for all tests (we use a value that is consistent with what is seen in practice). For the charging stations with both Level 2 chargers and superchargers, vehicles will give priority to using superchargers. 

Table~\ref{t1} compares the results of different policies with a fleet size of 1,200 vehicles and maximum waiting time of 30 min. 
\begin{table}[h!]
	\centering
	\caption{Comparison of different policies with a fleet size of 1,200 vehicles and a maximum waiting time of 30 minutes, Manhattan, NY.}
	\small
	\begin{tabular}{m{0.2\textwidth}>{\centering}m{0.14\textwidth}>{\centering}m{0.14\textwidth}>{\centering}m{0.14\textwidth}>{\centering\arraybackslash}m{0.15\textwidth}}
		\thickhline
		\multicolumn{1}{c}{Policies} & Mean wait time (min) & Mean no. waiting cust. & Mean no. lost cust. & Total dispatch cost (km) \\
		\hline
		\multicolumn{1}{c}{NonEV NoReb} & 21.7  & 1,034 & 24,017 & 592,869  \\
		\multicolumn{1}{c}{Charger Chasing} & 20.7 & 987 & 40,211 & 608,180 \\
		\multicolumn{1}{c}{\cite{li2019agent}} & 30.4 & 1,447 & 247,011 & 365,614 \\
		\multicolumn{1}{c}{\cite{loeb2018shared}} & 23.1 & 1,101 & 31,555 & 747,819 \\
		\multicolumn{1}{c}{\textsf{MDPP} with $V$=0} & 37.4 & 1,780 & 89,944 & 781,448\\
		\multicolumn{1}{c}{\textsf{MDPP} with $V$=0.001} & 12.4 & 591 & 1,993 & 325,850\\
		\multicolumn{1}{c}{\textsf{MDPP} with $V$=0.01} & 12.6 & 599 & 2,730 & 325,297\\
		\multicolumn{1}{c}{\textsf{MDPP} with $V$=0.1} & 12.9 & 612 & 1,958 & 324,187\\
		\multicolumn{1}{c}{\textsf{MDPP} with $V$=1} & 18.8 & 893 & 1,862 & 294,587\\
		\thickhline
	\end{tabular}
	\label{t1}
\end{table}
The battery capacity is 20 kwh. As we can see, the \textsf{MDPP} policy with $V=$0.001, 0.01, and 0.1 have similar results, and they outperform all the other policies, with the lowest mean waiting times, mean numbers of waiting customers, mean numbers of lost customers, and very low dispatch costs. When $V=0$, \textsf{MDPP} performs the worst, with the highest dispatch cost and the longest mean waiting time. This is because $V=0$ in \eqref{mdp_o} implies that \textsf{MDPP} does not care about the dispatch cost at all and only seeks to minimize the waiting time in every time step. This turns out to be counterproductive and it results in the longest mean waiting time. The total dispatch cost decreases as $V$ increases, and such reduction comes with a loss in the customers' level of service when $V$ increases from 0.1 to 1.

Table~\ref{t2} lists more results for different policies when we vary the fleet size, with a maximum waiting time of 30 min, and battery capacity of 20 kwh. 
\begin{table}[h!]
	\centering
	\caption{Comparison of different policies with different fleet sizes, maximum waiting time = 30 min, Manhattan, NY.}
	\small
	\begin{tabular}{m{0.06\textwidth}>{\centering}m{0.26\textwidth}>{\centering}m{0.12\textwidth}>{\centering}m{0.12\textwidth}>{\centering}m{0.1\textwidth}>{\centering\arraybackslash}m{0.13\textwidth}}
		\thickhline
		Fleet size & Policies & Mean wait time (min) & Mean no. waiting cust. & Mean no. lost cust. & Total dispatch cost (km) \\
		\hline
		\multicolumn{1}{c}{\multirow{6}{*}{1,400}} & NonEV NoReb & 10.3  & 491 & 0 & 609,475  \\
		& Charger Chasing & 12.0 & 570 & 7,246 & 637,335 \\
		& \cite{li2019agent} & 28.5 & 1,355 & 232,871 & 380,470 \\
		& \textsf{MDPP} with $V$=0.001 & 9.4 & 449 & 165 & 316,314\\
		& \textsf{MDPP} with $V$=0.01 & 9.0 & 430 & 357 & 316,071\\
		& \textsf{MDPP} with $V$=0.1 & 9.2 & 436 & 387 & 314,676\\
		\hline
		\multicolumn{1}{c}{\multirow{6}{*}{1,600}} & NonEV NoReb & 8.0  & 379 & 0 & 597,053  \\
		&Charger Chasing & 8.1 & 383 & 0 & 627,845 \\
		&\cite{li2019agent} & 27.3 & 1,301 & 215,064 & 399,478 \\
		&\textsf{MDPP} with $V$=0.001 & 6.8 & 324 & 0 & 307,934\\
		& \textsf{MDPP} with $V$=0.01 & 7.0 & 335 & 0 & 308,251\\
		& \textsf{MDPP} with $V$=0.1 & 7.4 & 354 & 0 & 307,861\\
		\hline
		\multicolumn{1}{c}{\multirow{6}{*}{1,800}} & NonEV NoReb & 7.9  & 377 & 0 & 594,885  \\
		& Charger Chasing & 7.9 & 377 & 0 & 622,733 \\
		& \cite{li2019agent} & 26.2 & 1,249 & 197,318 & 417,238 \\
		& \textsf{MDPP} with $V$=0.001 & 6.4 & 305 & 0 & 302,259\\
		& \textsf{MDPP} with $V$=0.01 & 6.4 & 304 & 0 & 301,835\\
		& \textsf{MDPP} with $V$=0.1 & 6.4 & 306 & 0 & 301,316\\
		\hline
		\multicolumn{1}{c}{\multirow{6}{*}{2,000}} & NonEV NoReb & 7.9  & 377 & 0 & 593,778  \\
		& Charger Chasing & 7.9 & 376 & 0 & 620,827 \\
		& \cite{li2019agent} & 24.8 & 1,180 & 178,695 & 436,694 \\
		& \textsf{MDPP} with $V$=0.001 & 6.2 & 297 & 0 & 296,638\\
		& \textsf{MDPP} with $V$=0.01 & 6.2 & 297 & 0 & 295,600\\
		& \textsf{MDPP} with $V$=0.1 & 6.3 & 298 & 0 & 296,765\\
		\thickhline
	\end{tabular}
	\label{t2}
\end{table}
Note that the method of \citep{loeb2018shared} is not included as it becomes computationally prohibitive for fleet sizes of 1400 vehicles and more. We also do not include the extreme cases of \textsf{MDPP} with $V=0$ and 1, their performance is not good as can be seen in Table~\ref{t1}. As shown in Table~\ref{t2}, all policies except for the heuristic in \citep{li2019agent} have no lost customers with fleet sizes of 1,600 vehicles or more. The charger chasing policy achieves better mean waiting times and mean numbers of waiting customers than the heuristic of \citep{li2019agent} but with a much higher dispatch cost. In contrast, \textsf{MDPP} with $V=0.001$, 0.01, and 0.1 have similar performance, and they all outperform other policies in terms of mean waiting times, mean numbers of waiting customers, numbers of lost customer, and also the dispatch costs. In other words, compared with other policies, the \textsf{MDPP}s offer better service to the customers with a lower costs to the operator. On the other hand, Table~\ref{t2} shows that after the fleet size reaches 1,800 vehicles, a further increase in fleet size brings little improvement to the system performance for all policies except the heuristic in \citep{li2019agent}. This implies that more investment in fleet expansion beyond 1,800 vehicles is not cost-effective for the operator.

To investigate the optimal fleet configuration in this scenario (under real-world charging station distributions and customer demands), we take \textsf{MDPP} with $V=0.001$ as a representative policy, and simulate the system under different battery capacities and fleet sizes. We tested battery capacities of 20 kwh, 40 kwh, 60 kwh, 80 kwh, while $\infty$ kwh represents gasoline-powered vehicles. Fig.~\ref{F:manhattan_mdp0001} shows the results of different performance indices. 
\begin{figure}[h!]
	\centering
	\subfloat[]{\includegraphics[scale=0.28]{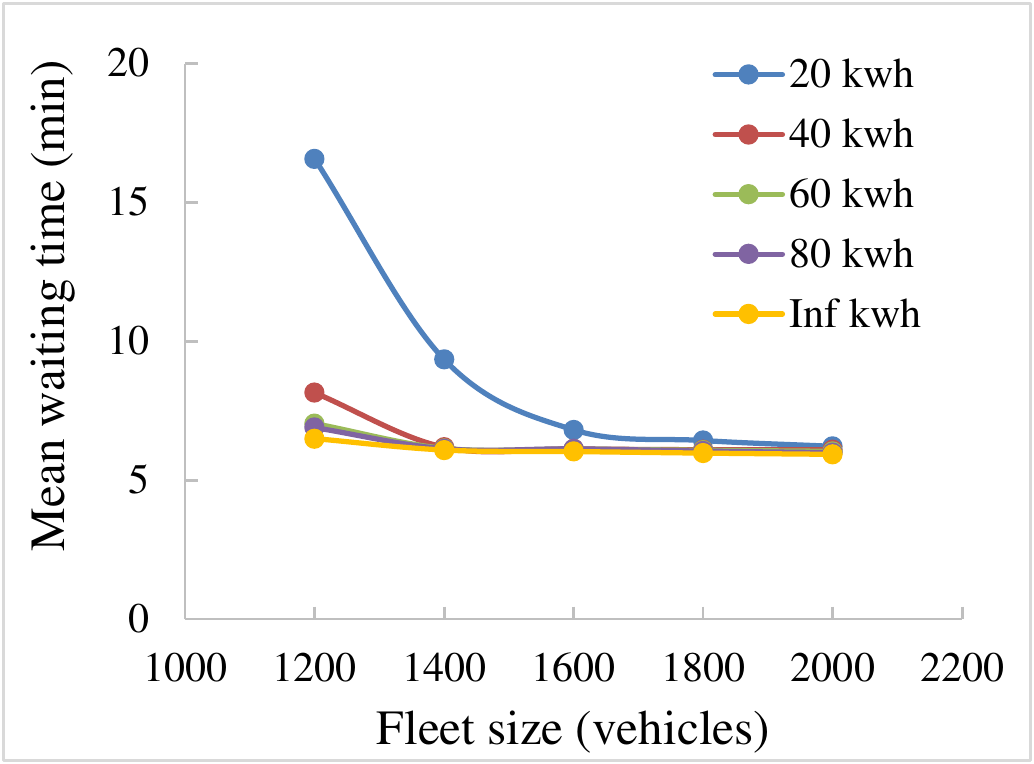}\label{F:manhattan_mdp0001_a}} 
	\subfloat[]{\includegraphics[scale=0.28]{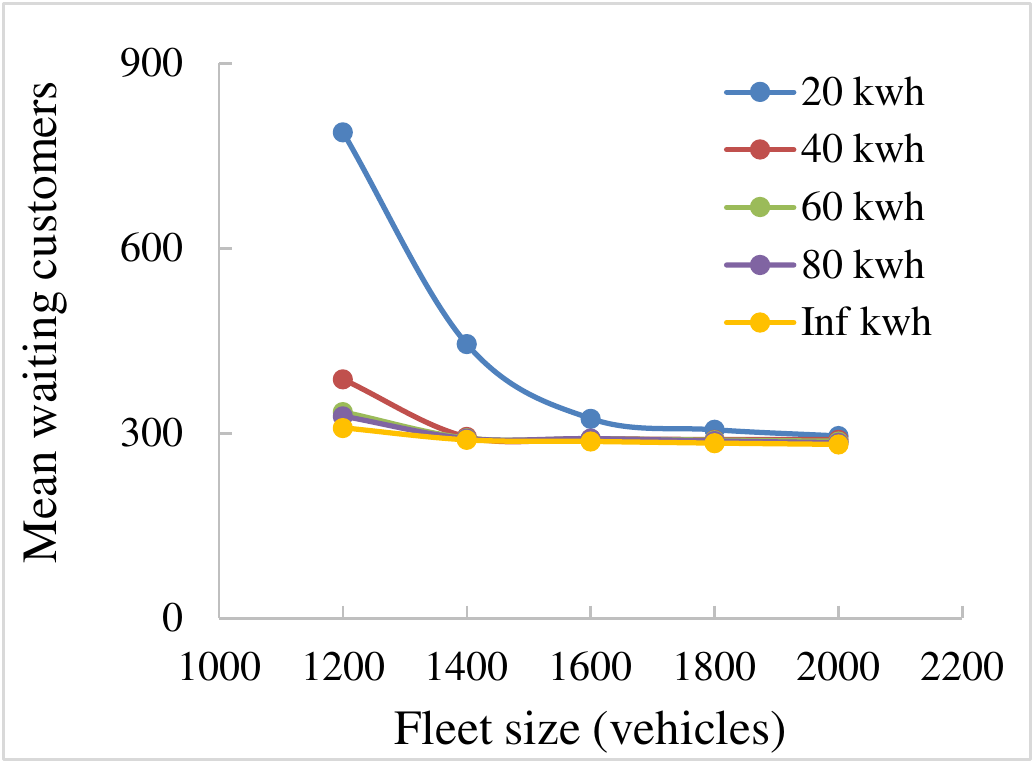}\label{F:manhattan_mdp0001_b}} \\
	
	\subfloat[]{\includegraphics[scale=0.28]{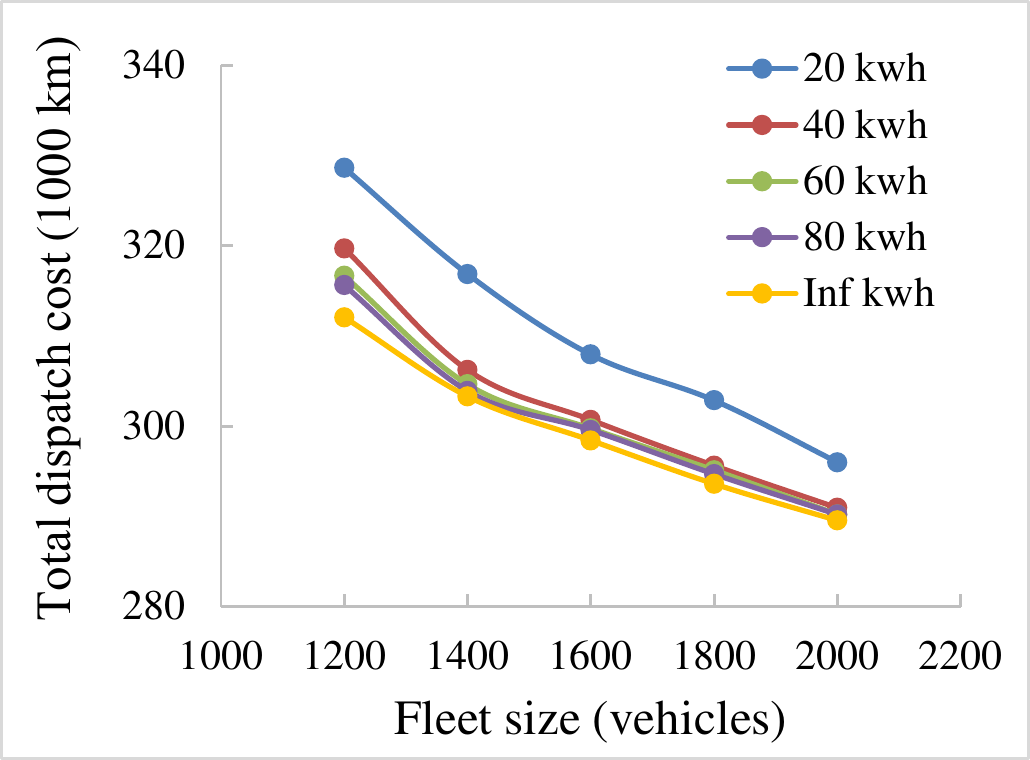}\label{F:manhattan_mdp0001_c}} 
	\subfloat[]{\includegraphics[scale=0.28]{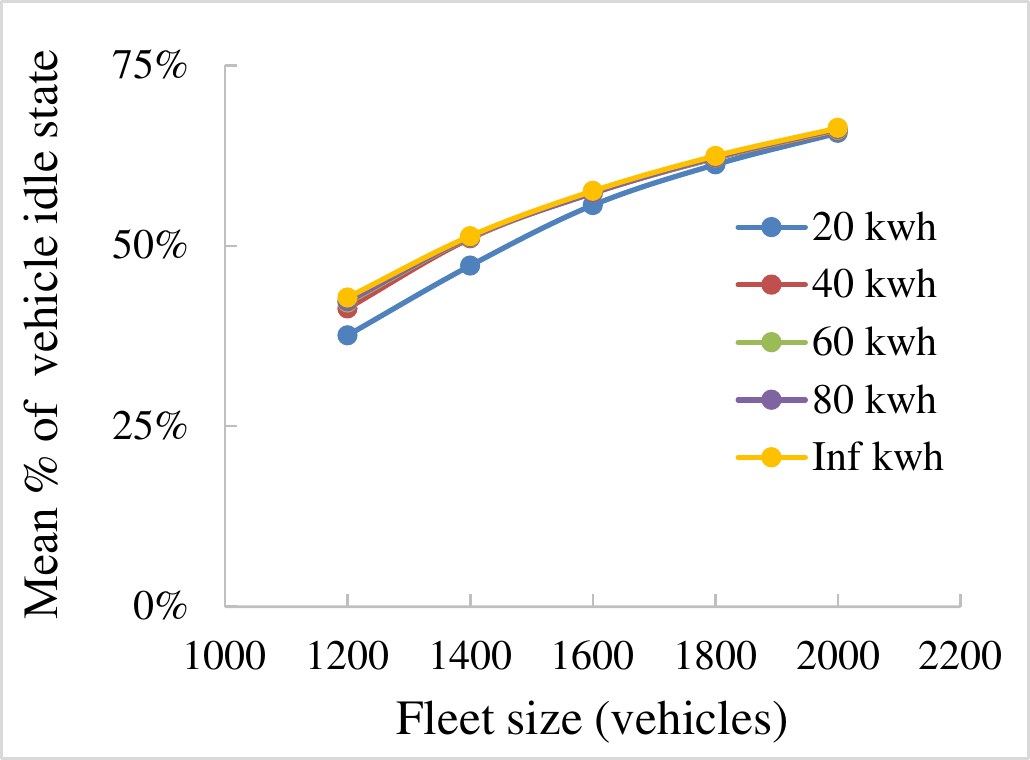}\label{F:manhattan_mdp0001_d}}
	\caption{System performances under different battery capacities and fleet sizes, \textsf{MDPP} with V=0.001, Manhattan, NY.\label{F:manhattan_mdp0001}}		
\end{figure}
Since we have already seen the results for maximum waiting time of 30 min in Tables~\ref{t1} and \ref{t2}, Fig.~\ref{F:manhattan_mdp0001} shows the results for case of no maximum waiting time (i.e., no customer abandonment). We find that when battery capacities increase from 20 kwh to 40 kwh, the total dispatch costs see dramatic reductions for all tested fleet sizes from 1,200 to 2,000 vehicles. The corresponding reductions in mean waiting times and mean numbers of waiting customers are obvious with fleet sizes of 1,400 vehicles or less. When the battery capacities 40 kwh or more, introducing larger battery capacities brings small benefit to the operator for all tested fleet sizes. This may be because all trips are short in this scenario. Hence, the key bottleneck that influences the system's performance is not battery capacity since a capacity of 40 kwh is sufficient.

Fig.~\ref{F:manhattan_mdp0001_d} shows the mean percentage idle time for all vehicles. Here, a vehicle is considered to be idle when it is not assigned to any customer. When the fleet size is 1,600 vehicles or more, vehicles are idle for more than half of the time on average, implying a waste of resources. Fig.~\ref{F:vd} further shows the vehicle dynamics (corresponding to partial simulations in Fig.~\ref{F:manhattan_mdp0001}) on a typical working day (Wednesday, June 6, 2018). 
\begin{figure}[h!] 
	\centering
	
	\resizebox{1.0\textwidth}{!}{%
		\includegraphics[scale=0.18]{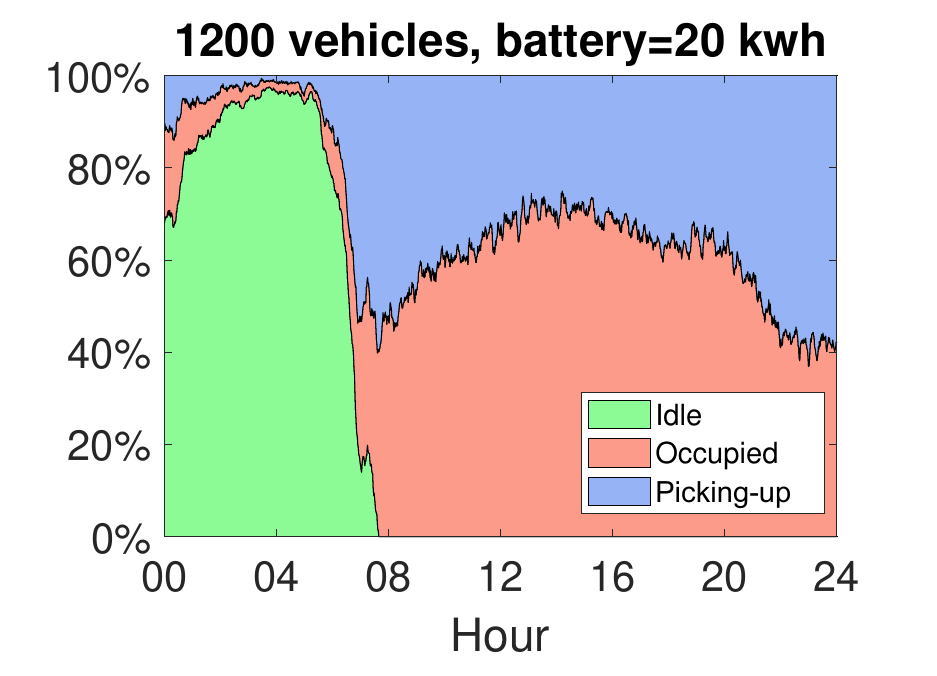}
		\includegraphics[scale=0.18]{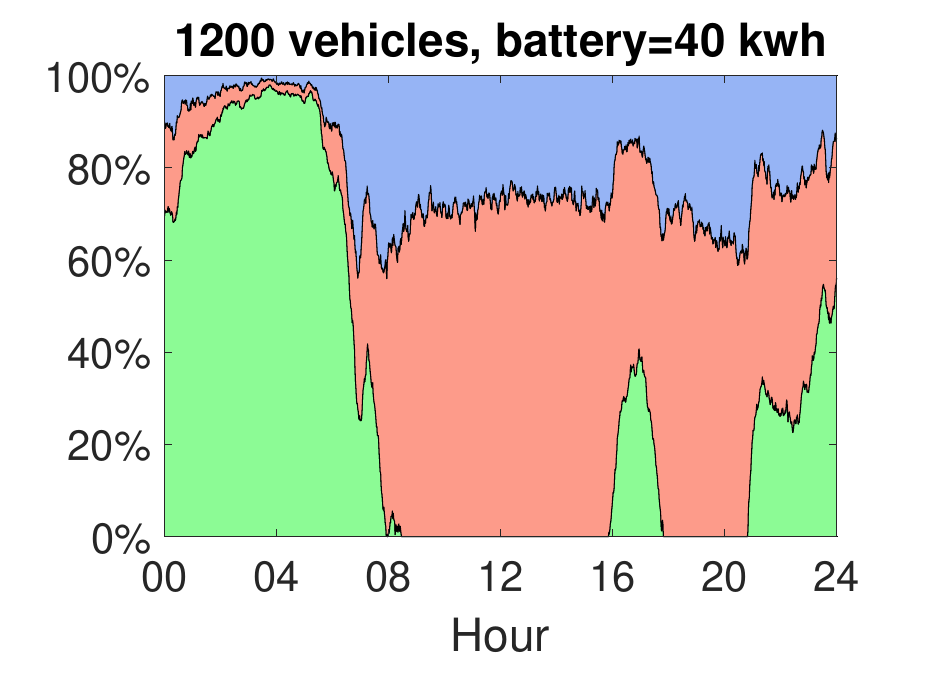}
		\includegraphics[scale=0.18]{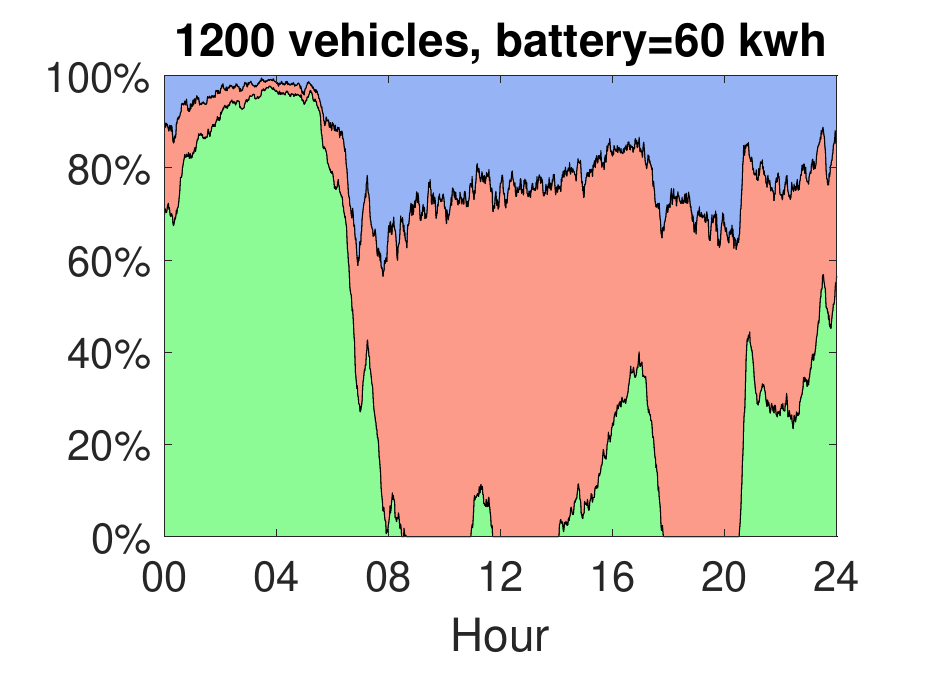}}	
	
	\resizebox{1.0\textwidth}{!}{%
		\includegraphics[scale=0.18]{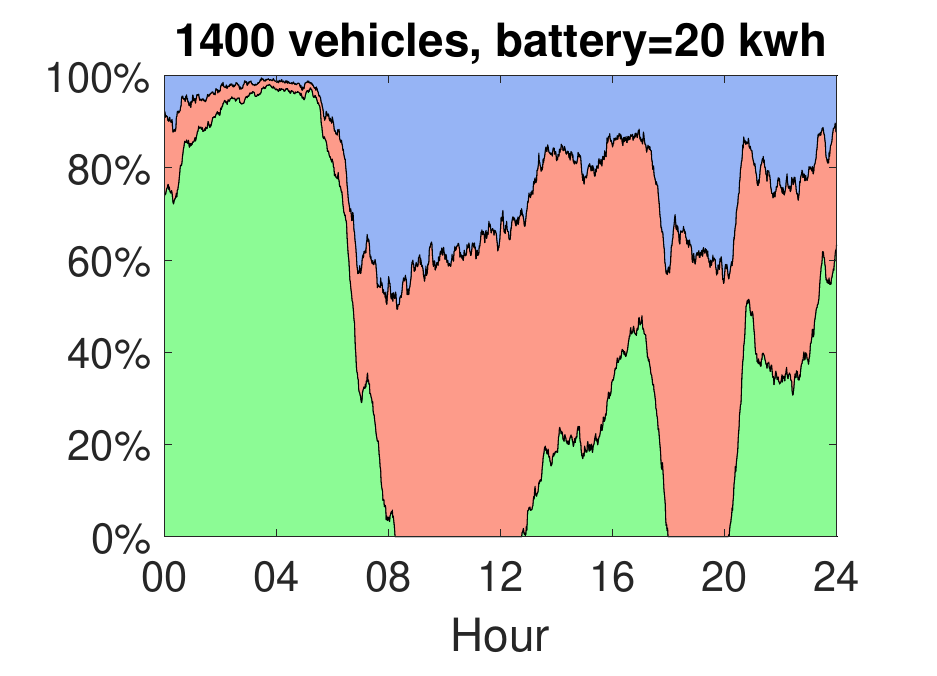}
		\includegraphics[scale=0.18]{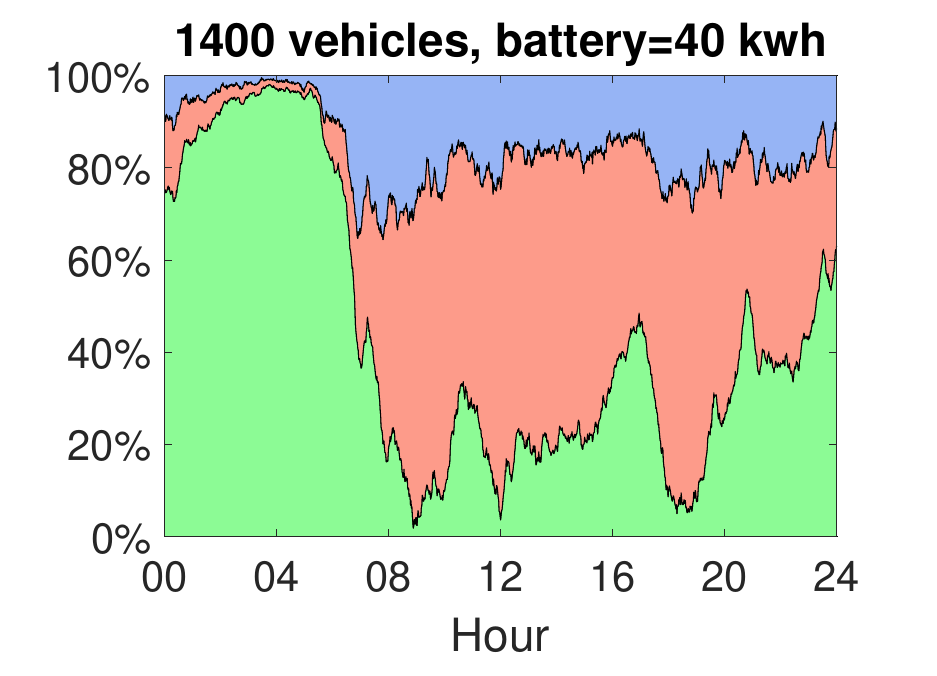}
		\includegraphics[scale=0.18]{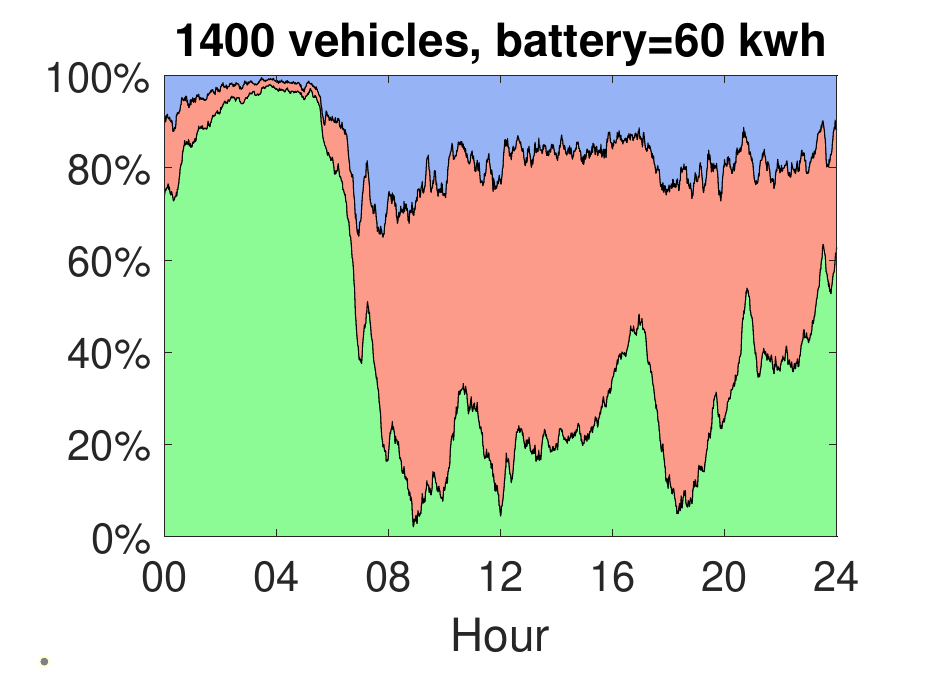}}
	
	\resizebox{1.0\textwidth}{!}{%
		\includegraphics[scale=0.18]{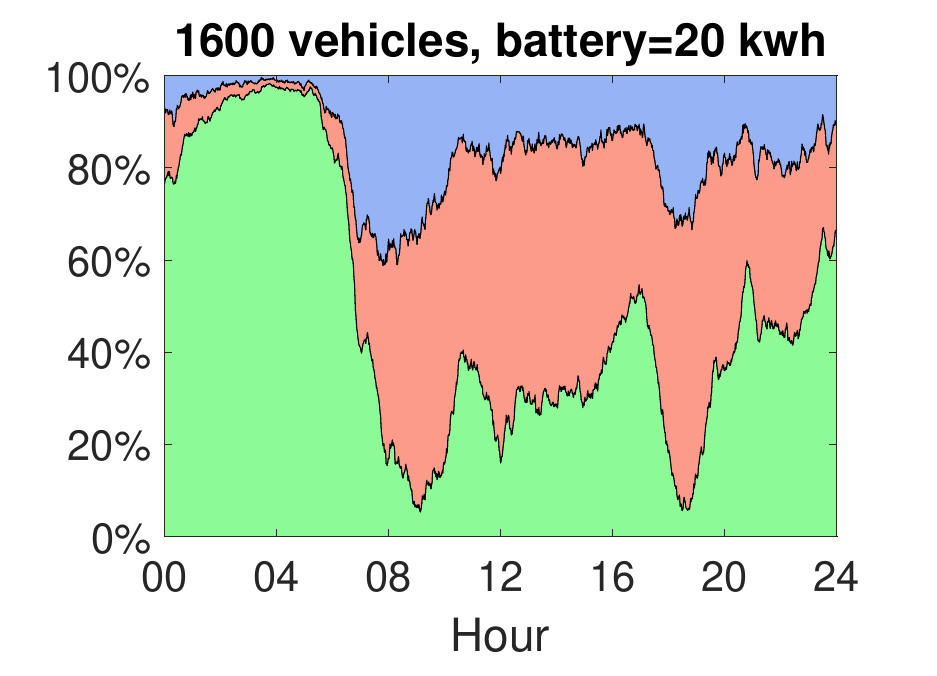}
		\includegraphics[scale=0.18]{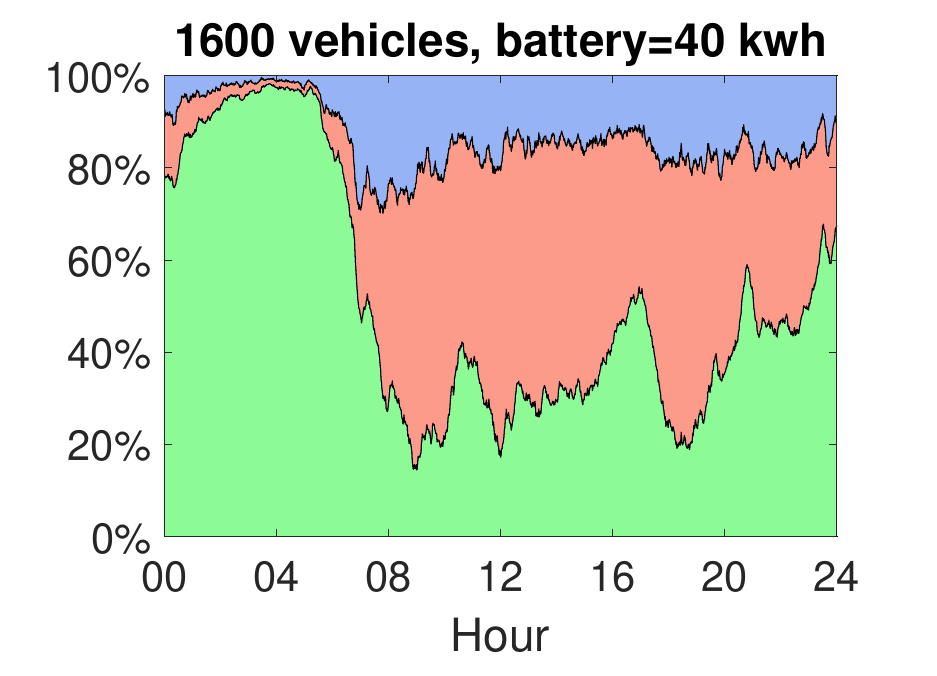}
		\includegraphics[scale=0.18]{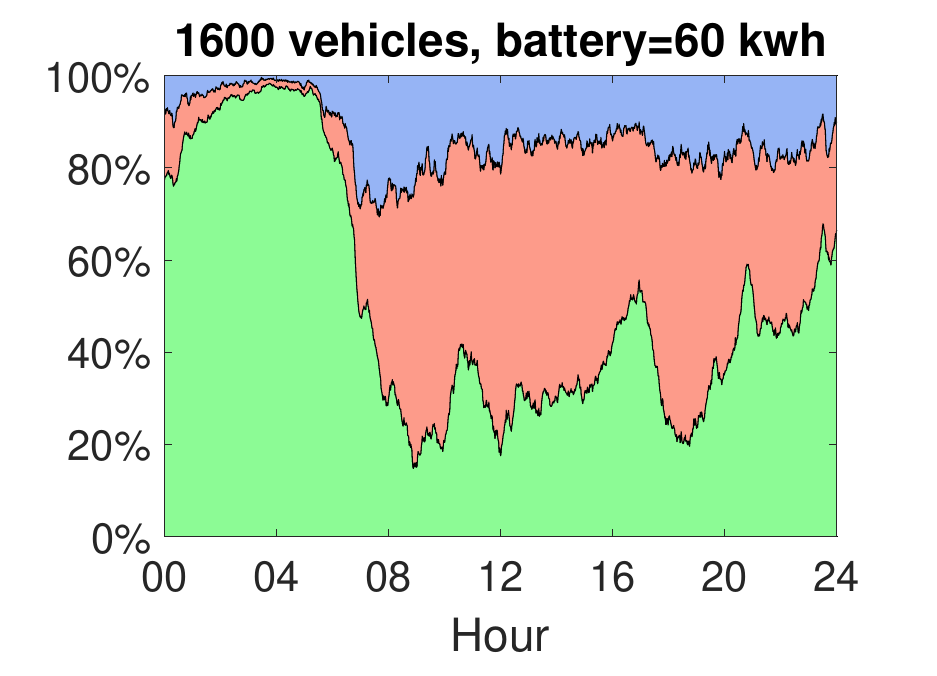}}
	
	\caption{Vehicle dynamics under different battery capacities and fleet sizes, \textsf{MDPP} with V=0.001, Manhattan, NY.}
	\label{F:vd}
\end{figure}
We can always find times when the system has no idle vehicles when the fleet size is 1,200 vehicles. This is not the case with a fleet size of 1,400 vehicles and battery capacity of 40 kwh. A higher battery capacity of 60 kwh makes little difference when the fleet size is 1,400 vehicles. When we have 1,600 vehicles, battery improvement from 20 kwh to 40 kwh makes a small difference, while an increase of battery capacity from 40 kwh or 60 kwh brings little benefit.

Fig.~\ref{F:wcd} shows the dynamics of waiting customers against time, corresponding to simulations in Fig.~\ref{F:manhattan_mdp0001}, with a comparison to the NonEV NoReb policy. We find that with a battery capacity of 20 kwh, even a fleet size of 1,600 will result in a surge in waiting customers on the last day of the simulated week. Combining the information from Fig.~\ref{F:manhattan_mdp0001} and Fig.~\ref{F:vd}, we recommend an optimal fleet configuration of fleet size = 1,400 vehicles and battery capacity = 40 kwh for this scenario when implementing \textsf{MDPP} with $V=0.001$.

\begin{figure}[h!]
	\centering
	
	\resizebox{1.0\textwidth}{!}{%
		\includegraphics[scale=0.22]{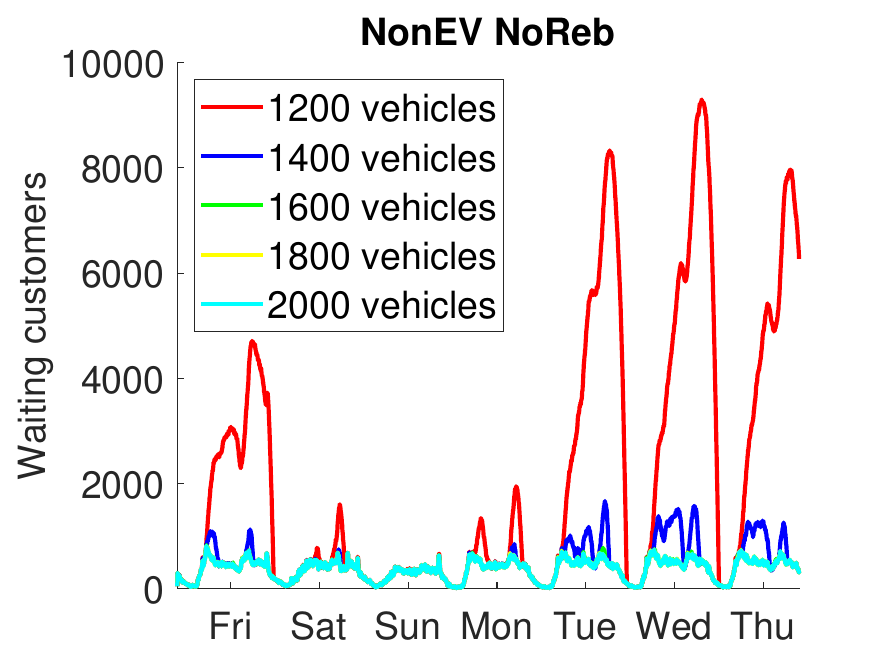}
		\includegraphics[scale=0.22]{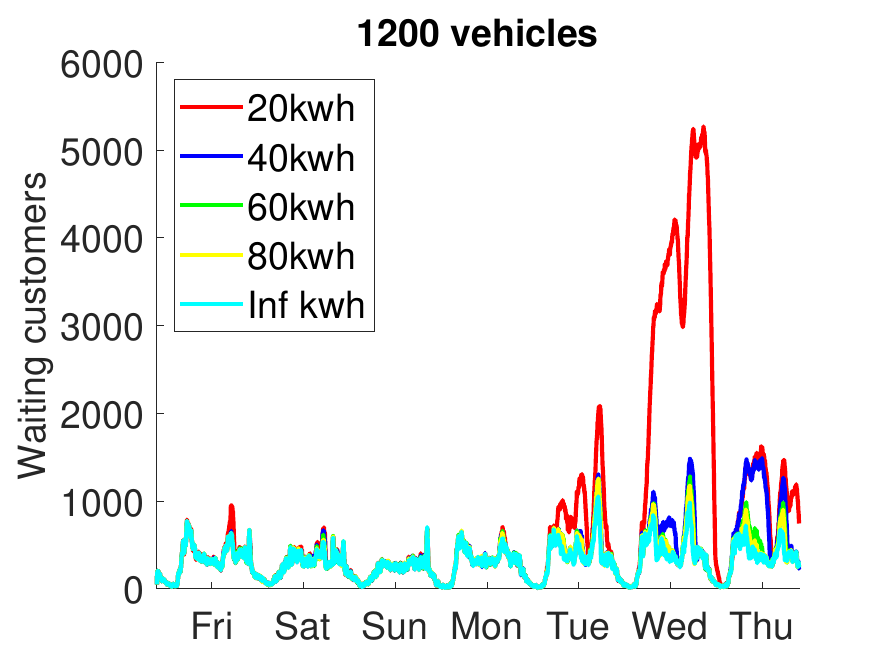}
		\includegraphics[scale=0.22]{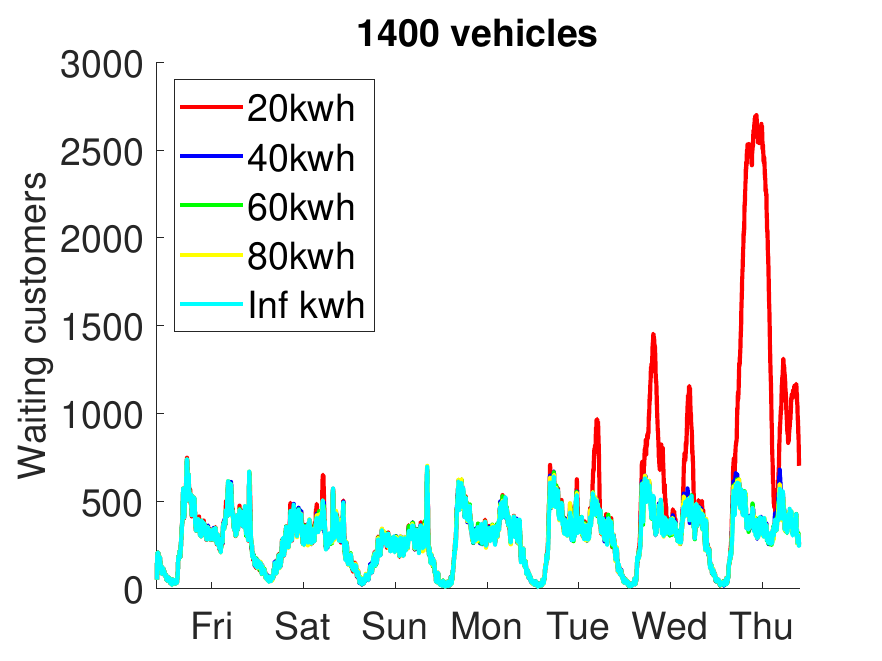}}
	
	\resizebox{1.0\textwidth}{!}{%
		\includegraphics[scale=0.22]{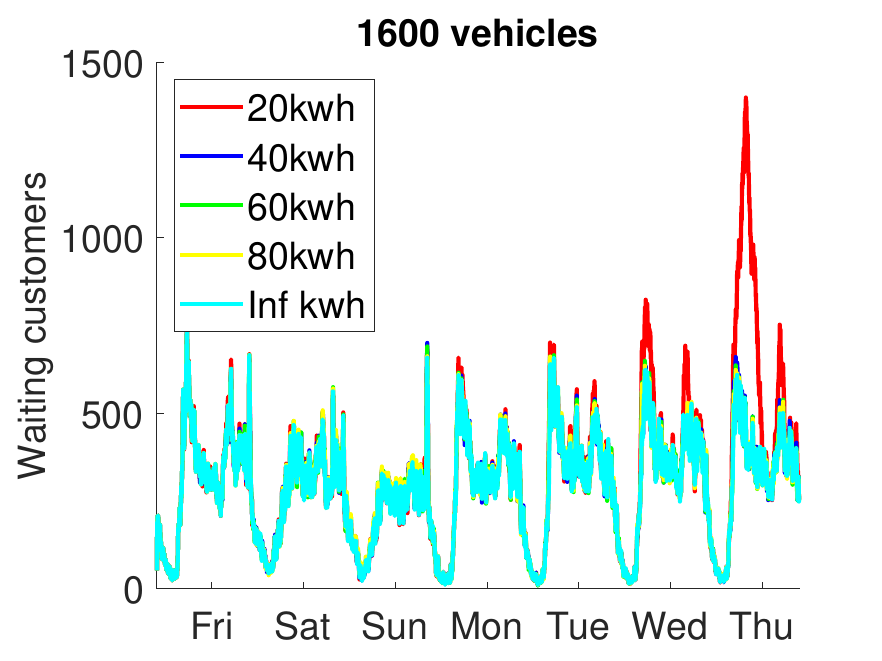}
		\includegraphics[scale=0.22]{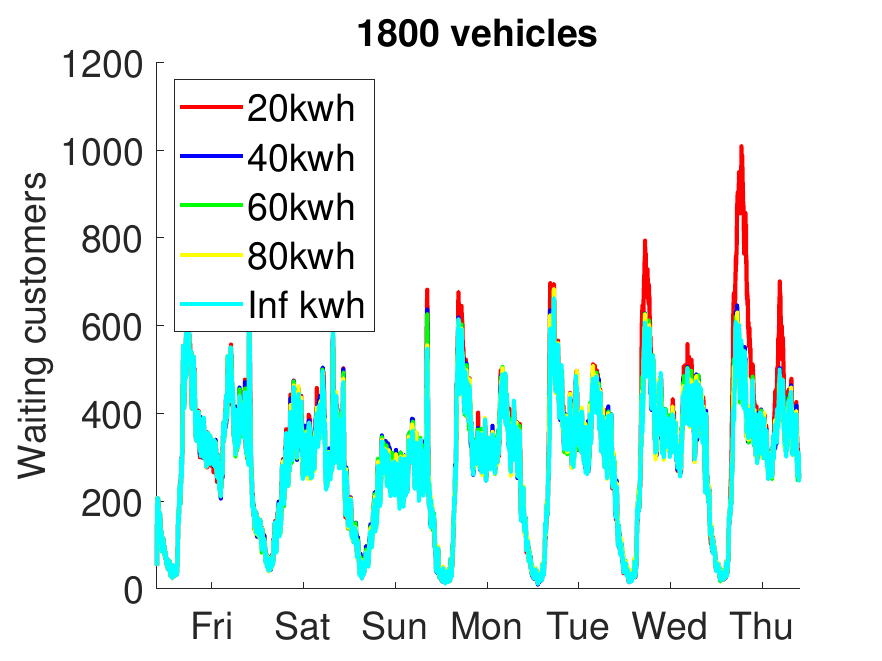}
		\includegraphics[scale=0.22]{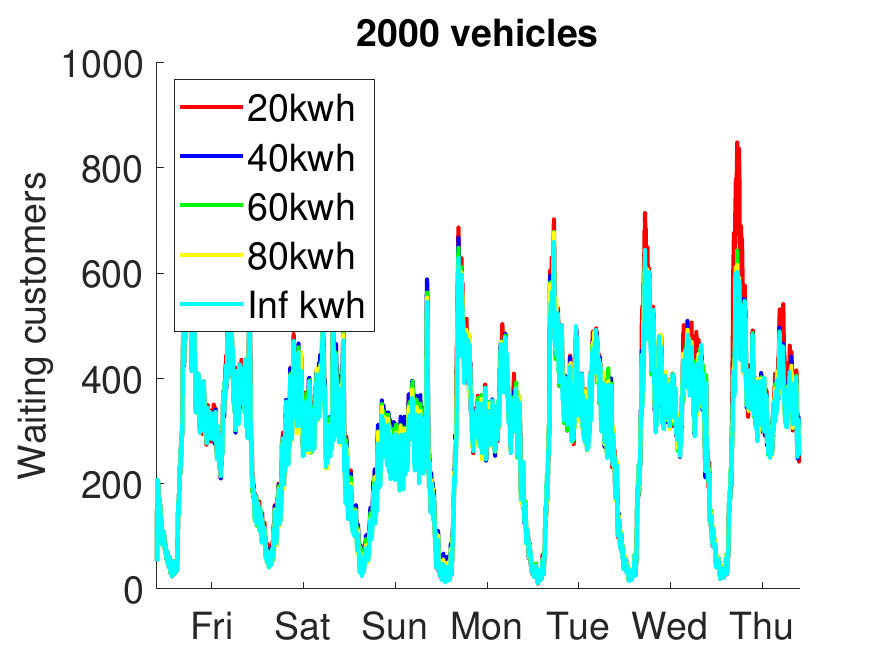}}
	
	\caption{Numbers of waiting customers against time under different battery capacities and fleet sizes, \textsf{MDPP} with V=0.001, Manhattan, NY.}
	\label{F:wcd}
\end{figure}

\section{Conclusion and outlook}
\label{S:Conc}
We propose a minimum drift plus penalty (\textsf{MDPP}) scheduling policy that can be implemented in real-time for large networks for vehicle dispatching in Shared Automated Electric Vehicle (SAEV) systems. The proposed approach has four main merits: (i) It does not require a priori knowledge of customer arrival rates to the different parts of the system.  The algorithm only requires knowledge of the waiting times of head-of-line customers and dispatch costs of vehicles at the time of assignment.  In other words, the algorithm does not need to anticipate customer arrivals or vehicle returns, both very difficult to gauge in practice. (ii) The approach ensures the stability of customer waiting times: we analytically demonstrated in Sec.~\ref{sec:stability} that as long as there is a way to ensure that customer waiting times do not explode, our real-time algorithm will find it. (iii) The algorithm ensures that the deviation of dispatch costs from a target dispatch cost can be controlled.  We use the dispatch cost that can be achieved with an \textsf{S-only} algorithm as our target, noting that \textsf{S-only} algorithms are known to be network stabilizing, but require information that is not available in practice. (iv) The proposed solution technique has a computational time-complexity that allows for real-time implementation.  By frequently updating the assignment solutions, we are able to achieve a time complexity that grows linearly with the number of occupied customer nodes in the system in the worst case.

Based on real demand from the BMW ReachNow car-sharing project in Brooklyn, NY and the Yellow Cab data in Manhattan, NY we test for both low and high demand scenarios, with long trips and short trips, respectively. The charging station locations are based on a real world distribution. Comparisons with other policies under different settings (battery capacities, charge powers, and fleet sizes) indicate that \textsf{MDPP} with appropriately chosen values for $V$ outperform all other algorithms in terms of waiting time, numbers of waiting customers, and vehicle dispatch cost.

The vehicle-to-customer assignment and vehicle recharging problem are considered together in \textsf{MDPP}, while vehicle rebalancing is not included. Future research can include an improved \textsf{MDPP} that considers vehicle relocation. One feature of the proposed \textsf{MDPP}, which may be considered a limitation, is that it does not provide service in a first-come first served way.  For example, customers at the same location but with different charging requirements are served in parallel, a customer that arrives later may enter service faster.  Also, as vehicles are returned to the system, it is possible that a customer that enters service later completes service earlier as a result of being assigned a newly returned vehicle.  Allowing for vehicle re-assignment while customers are in service (to provide faster service) is one possible way to overcome this. Such considerations could be of practical importance to operators and customers and, thus, deserve to be addressed in future research as well. 
Other improvements to the present approach would include more guidance into choosing the penalty constant $V$.  We observed that it plays a critical role in the performance of our method.

\section*{Acknowledgments} \label{Ack}
This work was supported in part by the New York University Abu Dhabi (NYUAD) Center for Interacting Urban Networks (CITIES), funded by Tamkeen, through the New York University Abu Dhabi (NYUAD) Research Institute Award under Grant CG001, and in part by the Swiss Re Institute through the Quantum Cities\textsuperscript{TM} Initiative. This research was also supported by the C2SMART University Transportation Center. Data was provided by BMW ReachNow car-sharing operations in Brooklyn, New York, USA.  The views expressed in this article are those of the authors and do not reflect the opinions of the sponsors, the funders, or supporting bodies.

\bibliography{refs_R1}
	
\end{document}